\documentclass[pdflatex,sn-mathphys-num]{sn-jnl}%

\usepackage{enumerate}
\usepackage{pgfplots}
\usepackage{graphicx}%
\usepackage{multirow}%
\usepackage{amsmath,amssymb,amsfonts}%
\usepackage{amsthm}%
\usepackage{mathrsfs}%
\usepackage[title]{appendix}%
\usepackage{xcolor}%
\usepackage{textcomp}%
\usepackage{manyfoot}%
\usepackage{booktabs}%
\usepackage{algorithm}%
\usepackage{algorithmicx}%
\usepackage{algpseudocode}%
\usepackage{listings}%
\usepackage{tabularx}%
\usepackage{empheq}
\usepackage{booktabs}
\usepackage{comment}
\usepackage{multirow}
\usepackage{enumitem}
\usepackage[normalem]{ulem}
\pgfplotsset{compat=1.18}

\theoremstyle{thmstyleone}%
\newtheorem{theorem}{Theorem}[section]
\newtheorem{proposition}[theorem]{Proposition}%
\newtheorem{problem}[theorem]{Problem}%

\theoremstyle{thmstyletwo}%
\newtheorem{example}{Example}%
\newtheorem{remark}{Remark}%

\theoremstyle{thmstylethree}%
\newtheorem{definition}{Definition}%

\newtheorem{hypothesis}[theorem]{Hypothesis}
\newtheorem{e-proposition}[theorem]{Proposition}

\newtheorem{e-definition}[theorem]{Definition\rm}
\newcommand{\jump}[1]{\lbrack\! \lbrack#1 \rbrack\!\rbrack}
\newcommand{\norm}[1]{\left\| #1 \right\|}
\newcommand{\abs}[1]{\left\vert #1 \right\vert}
\newcommand{\sigI}{\sigma_{\mathrm{I}}}
\newcommand{\sigII}{\sigma_{\mathrm{II}}}
\newcommand{\sigIII}{\sigma_{\mathrm{III}}}
\newcommand{\pI}{p_{\mathrm{I}}}
\newcommand{\pII}{p_{\mathrm{II}}}
\newcommand{\pIII}{p_{\mathrm{III}}}
\def\argmin{\qopname\relax m{argmin}}

\def\sig{{\boldsymbol{\sigma}}}
\def\veps{{\boldsymbol\varepsilon}}

\def\Gc{{\mathsf{G}_c}}
\def\u{\mathbf u}

\def\n{\mathbf n}
\def\K{\mathbb K}
\def\p{\mathbf p}
\def\ns{\boldsymbol{\nu}}
\def\s{\mathbf s}
\def\x{\mathbf x}
\def\sics{\sigma_c^{\s}}
\def\nus{\boldsymbol\nu^{\mathbf s}}
\def\nuss{\boldsymbol\nu_*^{\mathbf s}}

\def\S2{\mathbb S}
\def\Q{\mathbf Q}
\def\Orth3{SO(3)}
\def\T{\mathbf T}
\def\t{\mathbf t}
\def\v{\mathbf v}
\def\k{\mathsf k}
\def\A{\mathsf{A}_0}
\def\H{\textrm{H}}
\def\a{\mathsf a}
\def\w{\mathsf w}
\def\R{\mathbb R}
\def\dV{\mathrm{d}V}
\def\dS{\mathrm{d}S}

\graphicspath{{figures/}}
\newif\ifcorrred
\corrredtrue
\newcommand{\corrdate}{31 October 2026}
\ifcorrred
  \newcommand{\corr}[1]{\textcolor{red}{#1}}
\else
  \newcommand{\corr}[1]{#1}
\fi
\newcommand{\corrnote}{\noindent{\ifcorrred\color{red}\fi\textbf{Note.} This version, of \corrdate, corrects the text submitted on 27 June 2025 (arXiv:2506.22558v1). The corrections were made before receiving the review reports\ifcorrred{} and are printed in red\fi.}}
\begin{document}
\title{A variational approach to fracture incorporating any convex strength criterion}

\author[1]{\fnm{Blaise} \sur{Bourdin}}\email{bourdin@mcmaster.ca}

\author*[2,3]{\fnm{Jean-Jacques} \sur{Marigo}}\email{jean-jacques.marigo@polytechnique.edu}

\author[2]{\fnm{Corrado} \sur{Maurini}}\email{corrado.maurini@sorbonne-universite.fr}

\author[2,4]{\fnm{Camilla} \sur{Zolesi}}\email{camilla.zolesi@ensta.fr}

\affil[1]{\orgdiv{Department of Mathematics \& Statistics}, \orgname{McMaster University},  \city{Hamilton}, \postcode{L8S-4K1}, \state{Ontario}, \country{Canada}}

\affil*[2]{\orgdiv{Institut Jean Le Rond d'Alembert}, \orgname{Sorbonne Université}, \orgname{CNRS}, \city{Paris}, \postcode{F-75005}, \country{France}}

\affil[3]{\orgdiv{Laboratoire de Mécanique des Solides}, \orgname{Institut Polytechnique de Paris},  \orgname{CNRS} \city{Palaiseau}, \postcode{F-91128}, \country{France}}

\affil[4]{\orgdiv{Institut des Sciences de la Mécanique et Applications Industrielles}, \orgname{ENSTA, Institut
Polytechnique de Paris},  \orgname{CNRS} \city{Palaiseau}, \postcode{F-91762}, \country{France}}

\date{\today}
\abstract{
    We propose a variational phase-field model of fracture capable of accounting for arbitrary closed convex strength domains. 
    Unlike traditional models based on Ambrosio and Tortorelli regularization, the phase-field variable does not affect the material stiffness. 
    Instead, our elastic energy exhibits linear growth outside a strength domain, which shrinks to 0 as the phase-field variable goes to 1.
    We characterize this model through a fundamental problem on a \corr{periodic cubic cell under an imposed average deformation}.
    We show that the solution of this problem is a transverse cohesive crack, provided that the applied load and the direction of the displacement jumps satisfy a compatibility criterion, which we formulate in terms of Mohr's circles for isotropic strength domains.
    This allows us to derive a hierarchy of strength criteria for which fracture is never possible, sometimes possible or always possible, depending on the direction of the stress tensor.
    We discuss the properties of the model and postulate a ``sharp-interface'' limit in the form of a cohesive law that can be explicitly derived from the form of the phase-field model. \corr{This cohesive law is the sharp-interface limit when flat cracks are energetically stable against zigzagging; otherwise, the limit is obtained by relaxing it.}
    The proposed framework unifies within a single consistent variational theory key concepts developed over the centuries to predict or prevent material failure: Griffith and cohesive crack models, damage models, plasticity, strength criteria, and limit analysis.
}
\maketitle


\corrnote



\newpage
\tableofcontents
\newpage

\section{Introduction}
Fracture is governed by two fundamental material properties: toughness, which characterizes the energy associated with surfaces of displacement discontinuity, and strength, which defines the maximum stress a material can sustain before failure. 
Toughness determines the threshold for the propagation of existing cracks, whereas crack nucleation in pristine materials relates to the strength.

The construction of a model capable of accounting for both behaviors and material properties \emph{independently} is a well--known challenge. 
Introducing a material strength, \emph{i.e.} postulating that at all times, the stresses remain within a given bounded elastic domain, implies a bulk energy with at most linear growth at infinity for which strains may concentrate along sets of co-dimension 1 (curves and surfaces in two and three dimensions respectively)~\cite{Suq81,DalDeSMor06,Francfort2015}.
In areas of strain localization, the bulk energy degenerates into a \emph{cohesive} surface energy depending on the magnitude of the displacement jump~\cite{BouBraBut95}, which is energetically favorable over the creation of a crack whose energy is independent of the aperture and given by the toughness.
This is essentially incompatible with the fundamental assumption of a Griffith-like surface energy proportional to the crack length or surface, and independent of their opening, and suggests that reconciling strength and toughness necessarily leads to a cohesive fracture energy~\cite{ConFocIur24,barenblatt1962}.

Existing variational phase-field formulations capture crack nucleation through an effective strength criterion which emerges from the loss of stability of {homogeneous damage states}~\cite{Pham-Marigo-EtAl-2011a}.
In some regimes, these models are in good quantitative agreement with experiments in a wide range of materials~\cite{BouMarMau14,TanLiBou18}.
However, they offer limited and indirect control over the actual strength domain—an issue that becomes particularly pronounced under multiaxial loading~\cite{comiFractureEnergyBased2001,AMOR20091209,FREDDI20101154,KUMAR2020104027,VincentiniZolesi24}.  
In this work, we extend the variational theory of fracture by incorporating classical strength criteria, thus unifying key concepts developed to model material failure within a consistent variational framework.
Drawing from concepts of perfect plasticity~\cite{Suq81,DalDeSMor06} and limit analysis~\cite{Sal83,Salencon1993CISM}, our model uses {a non-linear elastic energy} with linear growth at infinity in order to enforce general strength criteria.
A phase-field variable contracting the admissible stress domain homothetically without degrading the elastic stiffness leads to vanishing stresses along the areas of strain localization.
Hence, we devise a multiaxial model problem on a cube to characterize the sharp interface limit of our model as a \emph{cohesive} fracture energy converging to a Griffith-like energy for ``large'' openings.
We characterize the link between {non-linear} bulk energy, cohesive surface energy, strength domains, compatible displacement jumps, and admissible traction vectors in the multi-axial regime, {borrowing concepts from limit analysis.}

Before introducing our model and main results, we present a brief review of variational and phase-field approaches in fracture and plasticity.
We also summarize relevant notations in Table~\ref{tab:notations}.

\subsection{Classical phase-field models of fracture}
\label{sec:classicalPhaseField}
Consider a solid occupying a domain $\Omega\subset \mathbb{R}^n$ ($n=2,\ 3$) in its reference configuration and subject to a combination of time-dependent boundary and body loads and
 displacement boundary conditions.
Let $\A$ be the fourth-order elasticity tensor of the pristine material and $\Gc$ its fracture toughness.
Let $\veps(\u) := \mathrm{sym}(\nabla\u) \in  \mathbb{M}^3_s$ be the linearized strain associated with a displacement field $\u$\corr{, $\mathbb{M}^3_s$ being} the  space of $3\times 3$ symmetric matrices.
Phase-field models of fracture were initially introduced in~\cite{BouFraMar00} as a regularization of the Griffith fracture energy
\begin{equation}
    \mathcal{F}(\u,\Gamma):=    \int_{\Omega\setminus\Gamma} 
    \frac{1}{2}{\A}\veps(\u)\cdot\veps(\u)\,\dV 
    +
    {\Gc} \mathcal{H}^{n-1}(\Gamma),
\label{eq:griffith}
\end{equation}
where $\Gamma$ is the crack set where displacement may jump and $\mathcal{H}^{n-1}(\Gamma)$  denotes its surface measure. 
Following~\cite{Braides-1998a} and introducing a regularization parameter $\ell >0$ and a function $\alpha$ representing the fracture surfaces, one defines
\begin{equation}
  \mathcal{F}_\ell(\u,\alpha):= \int_\Omega 
  \left( 
    \frac{\a(\alpha)}{2}{\A}\veps(\u)\cdot\veps(\u)
    + \frac{{\Gc}}{4c_\w}\left(\frac{\w(\alpha)}{\ell}+\ell\,\nabla\alpha\cdot\nabla\alpha\right)
  \right)
  \mathrm{d}V,
\label{eq:classical-PF}
\end{equation}
where $\a$ and $\w$  are continuous monotonic functions such that $\a(0) = \w(1) = 1$ and $\a(1) = \w(0) = 0$, and $c_{\w} = \int_0^1 \sqrt{\w(s)}\, \mathrm{d}s $ is a normalization factor.
Here and henceforth, we denote by \( \sig \cdot \veps = \mathrm{tr}(\sig^T \veps) = \sigma_{ij} \varepsilon_{ij} \) the scalar product between two second-order tensors.

The link between phase-field and Griffith energy is established in a series of mathematical works proving the $\Gamma$-convergence of $\mathcal{F}_\ell$ to $\mathcal{F}$ as $\ell \to 0$~\cite{Chambolle-2004b,Chambolle-2005,DalIur13}, from which one derives that global minimizers of~\eqref{eq:classical-PF} converge to those of~\eqref{eq:griffith}.

For fixed $\ell$ and when \emph{local} minimization is substituted for \emph{global} minimization, however, a finer analysis of these models reveals a very different behavior with respect to crack nucleation~\cite{PhaMar10a,PhaMar10b,SicMarMau14,MarMauPha16,Mar23}.
In this setting, the phase-field variable $\alpha$ is seen as a damage variable, $\alpha = 0$ representing pristine material and $\alpha = 1$ a fully damaged one.
The first-order optimality condition for the functional~\eqref{eq:classical-PF} at $\alpha = \alpha_0 = \mathrm{cst}$ implies that 
\begin{equation}
  {\A}\veps(\u)\cdot\veps(\u) \leq
\frac{\Gc}{2c_\w\ell}\frac{\w'(\alpha_0)}{\corr{(-\a'(\alpha_0))}}
\corr{\quad\Longleftrightarrow\quad
\A^{-1}\sig\cdot\sig
\leq
\frac{\Gc}{2c_\w\ell}\,\frac{\w'(\alpha_0)}{\mathsf s'(\alpha_0)},}
\label{eq:elasticDomainAT}
\end{equation}
\corr{where $\sig=\a(\alpha_0)\A\veps(\u)$ is the stress and $\mathsf s=1/\a$ the compliance function, with $\mathsf s'=-\a'/\a^2$.} 
This expression  defines a strength domain of the form
\begin{equation}
  \label{eq:defK0AT}
  \K_0 := \left\{ \sig \in \mathbb{M}^3_s;\ {\A^{-1}}\sig\cdot\sig \leq
\frac{\Gc}{2c_\w\ell}\left(\max_{0 \le \alpha_0 \le 1}\corr{\frac{\w'(\alpha_0)}{\mathsf s'(\alpha_0)}} \right)\right\},
\end{equation}
whose characteristic size scales as $\sqrt{\mathsf{E}_0 \Gc / \ell}$, $\mathsf{E}_0$ denoting the Young's modulus of the considered material.

In the tensile regime, for instance, this expression links tensile strength $\corr{\sigma_c}$ and the length scale $\ell$ through the relation
\[
  \corr{\sigma_c} = \sqrt{\frac{\Gc \mathsf{E}_0}{2c_\w\ell}\left(\max_{0 \le \alpha_0 \le 1} \corr{\frac{\w'(\alpha_0)}{\mathsf s'(\alpha_0)}}\right)}.
\]
In light of this observation, it becomes natural to view~\eqref{eq:classical-PF} as a \emph{gradient damage model}, where the length scale $\ell$ acts as an intrinsic material parameter~\cite{PhaMar10a,PhaMar10b,Mar23}, instead of a regularization of~\eqref{eq:griffith}.
This paradigm has proven effective in describing a broad range of fracture processes in the tensile regime, from crack nucleation to final failure, with good agreement with experimental observations~\cite{BouMarMau14,SicMarMau14,TanLiBou18}.

In the general case, however,~\eqref{eq:elasticDomainAT} defines domains defined in terms of critical strain energy density and it is not clear that it can be tailored to match experimental observations of general strength domains.
For instance, in the incompressible limit,~\eqref{eq:defK0AT} leads to infinite strength under hydrostatic tension or compression~\cite{KUMAR2020104027}.
 
Multiple attempts have been made to modify the form of~\eqref{eq:classical-PF} by introducing ``energy splits'', \emph{i.e.} replacing the strain energy density $\frac{\a(\alpha)}{2}\A\veps(\u)\cdot\veps(\u)$ with more complex expressions leading to more flexibility in the definition of the elastic domain~\cite{VincentiniZolesi24,ZOLESI2024105802}.
Whether doing so while preserving the link with the Griffith energy is feasible or not remains an open problem. 
Instead,~\cite{KUMAR2020104027} introduced  a ``driving force'' which is added to the optimality conditions for~\eqref{eq:classical-PF} with respect to $\alpha$, hence renouncing its variational nature and severing the connection with~\eqref{eq:griffith} and the mathematical and thermodynamical foundation upon which phase-field models of fracture were built. 

Several works seek to build phase-field approximations of cohesive fracture models~\cite{lorentz2011,ConFocIur16,ConFocIur24}.
However, their generalization to specific strength domains or cohesive laws is an open problem.
An interesting recent work of~\cite{Feng-Hai-2025a} incorporates realistic strength criteria in a variational cohesive phase-field model by introducing the information on the crack normal as structured deformations~\cite{FREDDI20101154}.

\subsection{Variational models of perfect plasticity}

Accounting for an arbitrary closed convex domain of admissible stresses $\K_0 \subset \mathbb{M}^3_s$ can be achieved through the well-established theory of perfect plasticity, in which the constitutive law for the stress tensor  is of the form $\sig = \A\left(\veps(\u) - \p\right)$, where $\p \in \mathbb{M}^3_s$ is the \emph{plastic strain}.
In a static variational framework (see~\cite{Han-Reddy-1999a,DalDeSMor06} amongst others), the deformation $\u$ and plastic strain $\p$ are  given as the minimizers of the functional
\begin{equation}
  \label{eq:VarPlasticity}
  \int_\Omega \frac12 \mathsf{A}_0(\veps(\u)-\p)\cdot(\veps(\u)-\p) +  \mathrm{H}_{\K_0}(\p)\, \dV,
\end{equation}
where $\H_{\K_0}(\p):= \sup_{\sig\in\K_0} \p\cdot \sig$ is the \emph{support function} of $\K_0$. 

Existence results are obtained for  $\u \in BD(\Omega; \mathbb{R}^3)$ in the space of functions of bounded deformation, and $\p \in \mathcal{M}(\Omega; \mathbb{M}^3_s)$ as a Radon measure; see~\cite{Suq81,Tem85,DalDeSMor06}. 
Formally, because $\H_{\K_0}$ grows linearly at infinity, $\veps$ and $\p$ may concentrate along regions of dimension strictly less than $n = 3$.
It is then natural to introduce an additive decomposition of the strain into a regular ($\mathrm{R}$) and a singular ($\mathrm{S}$) part:
\begin{equation}
  \label{eq:strain-SR-decomposition}
  \veps(\u) = \veps^\mathrm{R}(\u) + \veps^\mathrm{S}(\u),\quad \p = \p^\mathrm{R} + \p^\mathrm{S},\quad \text{with} \quad \veps^\mathrm{S}(\u) = \p^\mathrm{S},
\end{equation}
where the regular parts are assumed to be square-integrable, \emph{i.e.}, $\veps^\mathrm{R}(\u), \p^\mathrm{R} \in L^2(\Omega; \mathbb{M}^3_s)$, and the singular parts may concentrate on sets of dimension strictly less than $n = 3$, \emph{i.e.}, $\veps^\mathrm{S}(\u), \p^\mathrm{S} \in \mathcal{M}(\Omega; \mathbb{M}^3_s)$. 
The requirement of finite energy implies that $\veps - \p$ must be  square-integrable. The singular parts of the linear and plastic strains must therefore coincide.

Concentrations of the singular part of the strain $\veps^\mathrm{S}(\u)$ along sets of co-dimension 1 correspond to discontinuities of the displacement field.
For a given stress state $\sig$ on the boundary $\partial \K_0$ of $\K_0$, whether displacement jumps are possible or not and their direction (tangential, normal) depend on the geometry of the normal to the strength domain, namely whether the outer normal $\boldsymbol{\nu}$ to $\K_0$ at $\sig$ is of the form
\begin{equation}  
\label{eq:compatibilityIntro}
  \ns=\mathbf d\odot\mathbf n:=\frac12\left(\mathbf d\otimes\mathbf n+\mathbf n\otimes\mathbf d\right),\quad 
    \text{with}\quad
    \mathbf d\cdot\n\ge0,
\end{equation}
$\mathbf{d}$ denoting the direction of the displacement jump and $\n$ the normal to its jump set (see~\cite{Francfort2015,FRANCFORT2016125} for the von Mises strength criterion).
Normal stresses do not necessarily vanish in these regions, so that they cannot be seen as fractures.

In the absence of an irreversibility criterion,~\eqref{eq:VarPlasticity} is usually referred to as Hencky plasticity, which can be seen as a special case of a non-linear, non-smooth hyperelastic model~\cite{DuvautLions1976}.
Extending it to a proper quasi-static evolution requires the introduction of a concept of irreversibility for the plastic strain~\cite{DalDeSMor06}.

\subsection{Proposed phase-field model and statement of the main results}
Classical phase-field models can account for well-defined cracks and finite fracture toughness, and variational models of plasticity can account for arbitrary strength domains.
Our model is inspired by both approaches.
We propose an energy of the form
\begin{equation}
\mathcal E_\ell(\u,\mathbf p,\alpha):=
\int_\Omega\left(
\frac{\A}{2}(\veps(\u) -\p)\cdot(\veps(\u) -\p)+\k(\alpha)\H_{\K_0}(\p)\right) \dV 
+ \frac{\Gc}{4c_\w}\int_\Omega\left(\frac{\w(\alpha)}{\ell}+\ell\,\nabla\alpha\cdot\nabla\alpha\right) \dV,
\label{eq:plastic-damage-energyIntro}
\end{equation}
where $\k, \w: [0,1] \to [0,1]$ are respectively a monotonically decreasing and a monotonically increasing continuous function of the scalar phase-field variable $\alpha \in [0,1]$ satisfying $\w(0) = \k(1) = 0$ and $\w(1) = \k(0) = 1$.
Note that $\K_0$ may be \emph{any} closed convex subset of the stress space, independent of both the linear elastic stiffness tensor $\A$ and the fracture toughness $G_c$.
We formulate a quasi-static evolution for~\eqref{eq:plastic-damage-energyIntro} including an irreversibility condition of $\alpha$ within the classical framework of rate-independent processes~\cite{Mie05}.

In contrast with variational models of plasticity, we view $\p$ as a non-linear elastic strain, which is not subject to any irreversibility condition.
Unlike classical phase-field models, the linear elastic stiffness tensor $\A$ is not degraded by $\alpha$ and remains constant. 
Instead, the phase-field variable modulates the material strength: it scales the strength domain $\K_0$ by a homothety, thereby reducing the admissible stress set as $\alpha$ increases.
Indeed, for $\alpha = 0$, this energy reduces to perfect plasticity without irreversibility so that $\p$ can be interpreted as a nonlinear elastic contribution to the total strain $\veps(\u)$.

This model  can be seen as a generalization of the models coupling damage and plasticity introduced in~\cite{AleMarVid14,AleMarVid15}. For antiplane deformations and von Mises strength domains,~\citet{DalOrlToa16} have shown the Gamma-convergence of this class of models to cohesive fracture models as $\ell\to 0$. Yet, their behavior remains poorly understood for two- or three-dimensional elasticity and general strength domains.

\medskip
Section~\ref{sec:Model} presents the variational formulation of the model.
In Section~\ref{sec:fundamentalProblem}, \corr{for a finite $\ell$,} we solve an evolution problem in a cube under prescribed average deformation and fixed loading direction $\mathbf{s}$ in stress space, with $\s \cdot \s = 1$.
The magnitude of the average deformation serves as a scalar loading parameter $t$.
The solution remains linearly elastic, with $\p_t = 0$, $\alpha_t = 0$, and homogeneous strain as long as the stress $\sig_t = \sigma_t \s$ remains in the interior of the undamaged strength domain $\K_0$. 
The stress reaches the boundary of the domain at a critical load $t$ such that $\sig_t = \sigma_c^\s \s \in \partial \K_0$, where $\sigma_c^\s$ denotes the strength in the direction $\s$.
We show that, for loadings beyond this threshold, two qualitatively distinct behaviors may arise, depending on whether the loading direction $\s$ is compatible with displacement discontinuities in the sense that the outer normal $\boldsymbol{\nu}^\s$ to $\K_0$ at $\sig_t$ satisfies~\eqref{eq:compatibilityIntro}.
For loading directions compatible with displacement discontinuities, we construct a solution in which the plastic strain localizes as a measure supported on a cross-section $J_\u$ of the cube and vanishes elsewhere. 
The associated damage field $\alpha_t$ is regular away from the jump set $J_\u$, where it attains a constant maximal value $\bar{\alpha}_t$ and exhibits a discontinuity in its derivative induced by the concentration of plastic strain.
This solution represents the nucleation and progressive evolution of a cohesive crack.
In contrast, when the loading direction $\s$ is incompatible with a displacement jump, the solution remains spatially homogeneous, with both the plastic strain $\boldsymbol{p}_t$ and the damage field $\alpha_t$ uniform throughout the domain. 
We further analyze the stability of these two solution regimes, demonstrating that the homogeneous solution is unstable when displacement jumps are admissible.

In Section~\ref{sec:strength}, we reformulate classical strength criteria based on Mohr's representation of the stress vector. 
For isotropic strength criteria (see~\eqref{eq:defIsotropy})
we describe the set of admissible stress vectors in terms of an intrinsic domain $\mathcal{K}_0$ in Mohr's representation.
We prove that displacement jumps along a surface with normal vector $\n$ are possible if and only if $\sig \in \partial \K_0$ \emph{and} $\sig\n \in \partial \mathcal{K}_0$.

This concept of compatibility is consistent with classical results from limit analysis~\cite{DruPraGre52,Sal83}.
It introduces a hierarchy of strength domains for which displacement jumps are never, sometimes, or always possible.
We discuss this classification for the classical Drucker-Prager and Mohr-Coulomb families of criteria.

In Section~\ref{sec:SharpLimit}, we study the sharp interface limit of~\eqref{eq:plastic-damage-energyIntro}. 
As $\ell \to 0$, finite energy requires that $\alpha \to 0$ almost everywhere in $\Omega$. 
This constraint and our explicit construction of a fundamental solution suggest the following limiting energy functional:
\begin{equation}
  \label{eq:limit-energy}
  \mathcal{E}_0(\u,\hat\alpha)=
  \int_{\Omega\setminus J_\u} 
 \psi_0(\veps)
  \,\mathrm{d} V
  +
  \int_{J_\u}\corr{\phi_\n(\hat\alpha,\jump{\u})} \,\mathrm{d}S,
\end{equation}
where
\begin{equation*}
\psi_0(\veps)=
\inf_{\p\in{\mathbb{M}_s^3}}\corr{\tfrac12}\A(\veps-\mathbf p)\cdot(\veps-\mathbf p)+\H_{\K_0}(\p)
\end{equation*}
is an elastic energy density with linear growth at infinity in directions where $\K_0$ is bounded.
The surface energy density \corr{$\phi_\n$} can be expressed in terms of the displacement jump and a local damage variable ${\hat{\alpha}}$ supported on $J_\u$:
\begin{equation}
  \corr{\phi_\n(\hat\alpha,\jump{\u})} = 
  \hat\k(\hat\alpha)
  \H_{\K_0}(\n \odot \jump{\u})
  + \Gc\,\hat{\alpha},
  \label{eq:surface-energy-densities-limit-model}
\end{equation}
where
\begin{equation*}
  {\hat{\alpha}} (\alpha):= \frac{\int_0^{\alpha} \sqrt{\w(\beta)}\, \mathrm{d}\beta}{\int_0^1 \sqrt{\w(\beta)}\, \mathrm{d}\beta},\qquad
  \hat\k({\hat{\alpha}}):=\k(\alpha({\hat{\alpha}})).
\end{equation*}
We refer to $\hat{\alpha}$ as a \emph{local damage variable} as it is not subject to any gradient regularization terms, and carries a clear physical interpretation: it represents the fraction of dissipated surface energy relative to $G_c$, the critical energy release rate required to create a fully developed stress-free crack. \corr{The energy~\eqref{eq:limit-energy} is the candidate limit suggested by flat cracks. The cohesive model so obtained may fail to be lower semicontinuous, because its surface energy, minimized in $\hat\alpha$, need not be convex in the crack normal at fixed jump~\cite{AmbBra90,AmbFusPal00,DalOrlToa17,FriPerSol21}. Depending on the constitutive functions $\k$, $\w$ and on the strength domain, the effective limit energy may then require a further relaxation of the surface energy~\cite{BouBraBut95,AmbBra90,BouFonMas98}, which is not discussed here (Remark~\ref{rem:convexity-normal} in Section~\ref{sec:SharpLimit}).} 

The form of the cohesive law~\eqref{eq:surface-energy-densities-limit-model} highlights the influence of the strength domain $\K_0$ on the cohesive law, through the jump compatibility condition. 
Indeed, we express the cohesive law~\eqref{eq:surface-energy-densities-limit-model} in terms of the intrinsic domain $\mathcal{K}_0$, which is consistent with~\cite[Proposition 9]{ChaLavMar06} for variational cohesive models.
The dependence of the surface energy on $\K_0$ implies that the bulk and surface terms in~\eqref{eq:limit-energy} are linked and cannot be chosen arbitrarily, as established in~\cite{BouBraBut95} for anti-plane cohesive energy. 

The quasi-static formulation for the limit model \corr{of Section~\ref{sec:SharpLimit}} accounts for the irreversibility of damage by enforcing a growth condition on the local damage variable $\hat{\alpha}$. 

We illustrate the properties of the resulting cohesive model using three classical strength criteria: von Mises, incompressible Drucker-Prager, and Tresca with a tension cut-off.

For the simple example
\begin{equation*}
  \k(\alpha) = \corr{1 - \alpha^\zeta}, \quad \w(\alpha) = \alpha^{2}, \quad 1 \leq \zeta <2,
\end{equation*}
the function $\hat \k$ defining the cohesive law~\eqref{eq:surface-energy-densities-limit-model} writes as
\[
\hat\alpha=\alpha ^{2},\qquad\hat\k(\hat\alpha)=\corr{1-\hat\alpha^{\zeta/2}},\qquad\corr{\text{that is }\hat\k(\hat\alpha)=1-\sqrt{\hat\alpha}\text{ for }\zeta=1}.
\]

All these results  remain valid even in the degenerate case of a linearly rigid material, where the linear elastic stiffness $\A$ is replaced by the rigidity constraint $\veps(\u) = \p$, as in the classical theory of limit analysis for rigid-perfectly plastic materials~\cite{Sal83}.

Section~\ref{sec:Conclusion} concludes the paper and outlines directions for future work.

\begin{table}[h]
  \caption{Overview of the key notation employed in the paper.}
  \label{tab:notations}
  \begin{tabularx}{\textwidth}{cp{0.1\textwidth}X}
  \toprule
&\textbf{Symbol} & \textbf{Description} \\
\hline\\
  \multicolumn{2}{l}{{{General notation}}}\smallskip   \\
  &$\sigma_1,\sigma_2,\sigma_3$& Eigenvalues of the tensor $\sig$ (not ordered)\\
  &$\sigI,\sigII,\sigIII$& Ordered eigenvalues of the tensor $\sig$, $\sigI\ge\sigII\geq\sigIII$\\
  &$\mathbf{a} \wedge \mathbf{b}$ & The cross product of two vectors.\\
  &$\mathbf{a} \odot \mathbf{b}$ & The symmetrized tensor product of two vectors. $\mathbf{a} \odot \mathbf{b} = \frac12\left(\mathbf{a} \otimes \mathbf{b} + \mathbf{b} \otimes \mathbf{a}\right)$.\\

  \multicolumn{2}{l}{{{State variables and fields}}}  \smallskip  \\
  &$\u$ & Displacement vector $\in \mathbb{R}^n$\\ 
  &$\p $ & Nonlinear strain \\ 
  &$\sig$ & Stress tensor\\
  & $\veps$  & Linearized strain, $\veps(\u):={(\nabla \u + \nabla^T \u)}/2$\\
  & $\alpha$ & Phase-field variable\\
  & $\hat\alpha$ & Renormalized damage variable representing the fraction of dissipated energy in the sharp-interface cohesive model\\
  & $\bar\alpha$ & Trace of the phase-field on the jump set $J_\u$\\
  & $\jump{\u}$ &   Displacement jump\\
  & $\p^S$, $\veps^S$  & Singular part of the strains \\ 
  & $\delta_J$ & Dirac surface measure concentrated on the surface $J$ \\ 
  &$\Sigma$ & Normal stress $\Sigma:=\sig\n\cdot\n$\\
  &$\T$ & Tangential or shear stress, $\T:=\sig\n-\Sigma\,\n$\\

  \multicolumn{2}{l}{{{Material constants and functions}}}  \smallskip  \\
  &$\sigma_c$& \corr{Tensile strength parameter}\\
  &$\tau_c$& \corr{Shear strength parameter}\\
  &$\Gc$& Fracture toughness\\
  &$\mathsf{E}_0$& Young modulus\\
  &$\nu_0$& Poisson ratio\\
  &$\ell$& Regularization length\\
  &$\A$, $\A^{-1}$& Linear stiffness and compliance 4th-order tensors\\
  &$A_0^\s$& Linear stiffness  in the direction $\s$, with $A_0^\s:=\corr{1/(\A^{-1}\s\cdot\s)}$\\
  &$\boldsymbol{\nu}$& Normal to the yield surface $\partial\K_0$\\
  &$\boldsymbol{\nu}^\s, {\nu}^\s$&  Normal to the yield surface $\partial\K_0$ in the direction $\s$ and its $\s$-projection $\nu^\s:=\boldsymbol{\nu}^\s\cdot\s$\\
  &$\sigma_c^\s$& Maximum allowable stress in the direction $\s$\\
  &$\ell_\mathrm{ch}^\s$& Elasto-cohesive length in the direction $\s$, $\ell_\mathrm{ch}^\s:={\Gc A_0^\s}/{(\sigma_c^\s)^2}$\\
  &${\varepsilon}_e^\s$& {Elastic limit deformation} in the direction $\s$, ${\varepsilon}_e^\s=\sigma_c^\s/A_0^\s$\\
  &${\varepsilon}_c^\s$& {Critical deformation} in the direction $\s$ for the damage criterion\\
  &\corr{${\varepsilon}_u^\s$}& \corr{Ultimate deformation of the homogeneous response, at which $\alpha=1$}\\
  &\corr{${\varepsilon}_f^\s$}& \corr{Deformation at which the stress of the localized response vanishes}\\
  &$\k(\alpha)$& Strength degradation function\\ 
  &$\hat\k(\hat\alpha)$& Renormalized strength degradation function\\ 
  &$\w(\alpha)$& Local dissipation function\\ 
  &$c_\w$& Normalization constant, $c_\w:=\int_0^1\sqrt{\w(\alpha)}\mathrm{d}\alpha$\\
  &$c^*_\w$& Normalization constant, $c^*_\w:=\int_0^1\frac{1}{\sqrt{\w(\alpha)}}\mathrm{d}\alpha$\\
  
  &$k$& Slope parameter of the Drucker-Prager strength domain\\
  \multicolumn{2}{l}{{{Energy and energy densities}}}  \smallskip  \\
  &$\varphi_{0}(\veps,\p)$&Elastic energy density with explicit dependence on the nonlinear deformation\\
  &$\psi_{0}(\veps)$&Elastic energy density as a function of the geometric strain $\veps$ only, $\psi_{0}(\veps):=\min_{\p\in\mathbb M^3_s}\varphi_{0}(\veps,\p)$\\
  &$D_\ell(\alpha,\nabla\alpha)$&Dissipated energy density of the phase-field model\\
  &$\phi_\n(\hat\alpha_t, \jump{\u}) $&Surface energy density of the sharp-interface cohesive model\\
  &$\Phi(\bar\alpha) $&\corr{Surface energy density of the localized solution}\\
  \multicolumn{2}{l}{{{Sets}}}  \smallskip  \\
  &$\mathbb{M}^3_s$ & set of $3\times 3$ symmetric matrices with real coefficients\\
  & $\mathbb{S}^2$ & Unit sphere in $\mathbb{R}^3$\\
  & $\Omega\subset \mathbb{R}^3$ & Reference domain\\
  & $\Omega_L\subset \mathbb{R}^3$ & Cube of side $L$\corr{, periodic cell}\\
  & $\partial\Omega$ & Boundary of $\Omega$\\  
  & $J_\u\subset\Omega$ & Jump set for the displacement field $\u$\\
  &$\K_0, \H_{\K_0}(\p)$ & Initial strength domain, convex subset of $\mathbb M^3_s$ and its support function \\ 
  &$\partial\K_0$ & Initial yield surface, boundary of $\K_0$ \\
  &$\mathring\K_0$& Interior of a set $\K_0$, $\mathring\K_0=\K_0 \setminus \partial\K_0$\\
  & $\mathrm{N}_{\K_0}(\sig)$ & Normal cone to the convex set $\K_0$ in $\sig\in\partial\K_0 $\\
  &$\mathcal{K}_0, \H_{\mathcal{K}_0}(\mathbf{d})$ & Intrinsic domain in the Mohr's space $(\Sigma,T_1,T_2)$ and its support function for isotropic materials \\

   \bottomrule
  \end{tabularx}
\end{table}

\section{Formulation of the phase-field model}
\label{sec:Model}
\subsection{Fundamental requirements}
\label{sec:fundamentalRequirements}

We aim at formulating a regularized fracture model that incorporates the following fundamental  phenomenological features:
\begin{enumerate}
  \item[R1.] \emph{Linear elasticity at small strains} in the bulk, governed by a symmetric positive-definite elasticity tensor $\A$.

  \item[R2.] A \emph{strength criterion} in the bulk, described by a closed convex set $\K_0 \subset \mathbb{M}^3_s$ that defines the domain of admissible stress tensors $\sig$ and such that the material behaves in a linearly elastic manner when the stress lies in the interior of $\K_0$.
  \item[R3.] A \emph{finite fracture toughness} $\Gc$ on the crack surface, which is represented by a codimension-one set $\Gamma$ in the reference configuration $\Omega$. The displacement field $\u$ may be discontinuous across this surface, and the jump set $J_\u \subset \Gamma$ consists of points where the displacement jump $\jump{\u}$ is non-zero. 
  \item[R4.] \emph{Irreversibility of the crack set} $\Gamma$, translating the condition that cracks, once created, cannot self-heal.
\end{enumerate}

In order to satisfy R1 and R2, our model leverages the framework of variational models of plasticity in the bulk~\cite{Han-Reddy-1999a,DalDeSMor06} (see Section~\ref{sec:bulkBehaviour}).

The requirement of finite fracture toughness $\Gc$ implies that, for sufficiently large displacement jumps on a set of normal $\n$, the traction vector  $\sig\n$ must vanish in admissible jump directions, since the cohesive traction law is related to the derivative of the surface energy. 
In order to satisfy R3 and R4, we introduce an irreversible scalar phase-field $\alpha$ that progressively reduces the strength domain $\K_0$ as it grows. The phase-field variable is associated with a dissipated energy density, a gradient-type regularization, and an irreversibility condition, as is customary in phase-field fracture models (Section~\ref{sec:PFstrengthdegradation}).

\subsection{Constitutive model}
\subsubsection{Elastic material behavior in the bulk}
\label{sec:bulkBehaviour}
Given the \emph{linear elastic stiffness} as a fourth-order symmetric positive-definite tensor $\mathsf{A}_0$, and the \emph{strength domain} $\K_0$  as a closed convex subset of $\mathbb{M}^3_s$, we introduce the \emph{elastic energy density} associated with the second-order symmetric tensor $\veps\in \mathbb{M}^3_s$
\begin{equation}
   \label{eq:defPsi0}
  \psi_{0}(\veps) := \min_{\p \in \mathbb{M}^3_s} \varphi_{0}(\veps, \p),
\end{equation}
where
\begin{equation}
  \label{eq:defPhi0}
  \varphi_{0}(\veps, \p) := \frac{1}{2} \mathsf{A}_0(\veps - \p) \cdot (\veps - \p) + \H_{\K_0}(\p),
\end{equation}
and 
\[
\H_{\K_0}(\p) := \sup_{\sig \in \K_0} \sig \cdot \p
\]
denotes the support function of ${\K}_0$.

In all that follows, we assume that $\K_0$ and $\A$ verify the following basic properties.

\begin{hypothesis}[Constitutive assumption: elasticity and strength domain]
    We assume that the elastic tensor $\A$ is a symmetric positive definite fourth-order tensor and that the strength domain $\K_0$ is a closed convex subset of $\mathbb{M}^3_s$, that $\mathbf{0} \in \mathring{\K}_0:=\K_0\setminus\partial\K_0$, and that $\K_0$ is bounded in at least one direction.
\end{hypothesis}

\begin{remark}[Isotropic materials and limiting behaviors.]
In the examples we will specifically focus on the case of isotropic materials, for which $\A$ is parameterized by the Young's modulus $0 < \mathsf{E}_0 < \infty$ and the Poisson ratio $-1 < \corr{\nu_0} < 1/2$. In the present framework, the limiting cases of linear \corr{incompressibility} ($\corr{\nu_0} = 1/2$) and infinite stiffness ($\corr{\mathsf{E}_0} \to \infty$) can be treated without additional difficulty (see the example in Section~\ref{sec:example-DP-incomp}).
\end{remark}

We recall the following classical result, which can be proved using standard tools of convex analysis (see~\cite{DuvautLions1976,EkeTem76,Boyd-Vandenberghe-2004a} for instance). 
\begin{proposition}
    \label{prop:plasticityConstitutive}
    Let $\K_0\subset \mathbb{M}^3_s$ be a closed convex set such that $\mathbf{0} \in \mathring{\K}_0$.
    Consider $\veps \in \mathbb{M}^3_s$. Then, there exists a unique $\p_{\veps}$ achieving minimality in~\eqref{eq:defPsi0}.
    Let $\sig_{\veps} := \mathsf{A}_0(\veps-\p_{\veps})$, then 
    \begin{equation*}
        \sig_{\veps} \in \K_0   
    \end{equation*}
     and $\p_{\veps} \in \mathrm{N}_{\K_0}(\sig_{\veps})$, the normal cone to $\K_0$ at $\sig_{\veps}$, \emph{i.e.}
    \begin{equation}
      \label{eq:pNormalCone}
        (\sig_{\veps} - \sig)\cdot \p_{\veps} \ge 0,\quad \ \forall \sig \in \K_0.
    \end{equation}
    Furthermore, $\p_{\veps} = \mathbf{0}$ if $\mathsf{A}_0\veps \in \K_0$, and $\sig_\veps$ is a subgradient of the convex function $\psi_0$ at $\veps$: $\sig_\veps \in \partial \psi_0(\veps)$.
\end{proposition}

The result above implies in particular that the material model defined by~\eqref{eq:defPhi0} satisfies the requirements R1 and R2 of Section~\ref{sec:fundamentalRequirements}: the stress tensor $\sig$ must lie in $\K_0$ which can be interpreted as a strength-domain.

The normality rule~\eqref{eq:pNormalCone} implies that there exists an outer normal $\boldsymbol{\nu}_\sig$  to $\partial\K_0$ at a point $\sig\in\partial\K_0$, such that $\p=\lambda\,\boldsymbol{\nu}_\sig$ for some $\lambda\geq 0$. Hence, the inverse stress-strain relationship is in the form
\begin{equation}
  \label{eq:elastoplasticlaw}
\veps=\A^{-1}\sig+\lambda\,\boldsymbol{\nu}_\sig,\quad\text{with}\quad\lambda \ge 0,
\end{equation}

where $\A^{-1}$ denotes the compliance tensor.
\begin{remark}
   Proposition~\ref{prop:plasticityConstitutive} implies that $\psi_0(\veps) = \frac12\A\veps\cdot \veps$ if $\A\veps \in \K_0$, since $\p_\veps = \mathbf{0}$ if $\A\veps \in \K_0$.

  Consider a strain tensor  of the form $\veps_\lambda := \A^{-1} \sig_0 + \lambda \,\boldsymbol{\nu}_0$ where $\sig_0 \in \partial \K_0$, $\boldsymbol{\nu}_0 \in \mathrm{N}_{\K_0}(\sig_0)$ and $\lambda >0$. Optimality in~\eqref{eq:defPsi0} is achieved by $\p_{\veps_\lambda} = \lambda\,\boldsymbol{\nu}_0$, so that $\A(\veps_\lambda - \lambda \,\boldsymbol{\nu}_0) = \sig_0 \in \K_0$ and $\mathrm{H}_{\K_0}(\corr{\p_{\veps_\lambda}}) = \lambda \,\sig_0 \cdot \boldsymbol{\nu}_0$, from which we conclude that 
    \[
      \psi_0(\veps_\lambda) = \corr{\tfrac12}\A^{-1}\sig_0\cdot \sig_0 + \lambda \,\sig_0 \cdot \boldsymbol{\nu}_0,
    \]  
    \emph{i.e.} that $\psi_0$ grows linearly at infinity away from directions along which $\K_0$ is bounded.
\end{remark}

\begin{remark}[Analogy with elasto-plasticity]
The constitutive model above coincides with that of a \emph{standard} elasto-plastic material without hardening \emph{à la} Hencky, \emph{i.e.}~without irreversibility condition on $\p$~\cite{DuvautLions1976}. 
In all that follows, we refer to $\p$ as the \emph{non-linear deformation}.
This model leads to a nonlinear and non-smooth, yet reversible, elastic behavior in the bulk, with an energy that grows linearly for large strains in the directions where $\K_0$ is bounded. 
This choice simplifies both the presentation and the analysis, but it can be readily modified to incorporate the effects of residual deformations and plastic dissipation.
\end{remark}

\subsubsection{Phase-field strength degradation}
\label{sec:PFstrengthdegradation}
To model cracks with finite fracture toughness, the cohesive traction—and hence the stress (or at least some components of it)—must vanish for sufficiently large strains. Within a phase-field like approach, this may be achieved by introducing a scalar damage variable $\alpha$ that modifies the pristine elastic energy~\eqref{eq:defPhi0} through two basic mechanisms:
\begin{itemize}
  \item \emph{Stiffness degradation}, \emph{i.e.}, replacing the linear stiffness tensor $\A$ with a degraded elastic stiffness tensor $\mathsf{A}(\alpha) \leq \A$, which decreases as $\alpha$ increases;
  \item \emph{Strength degradation}, \emph{i.e.}, replacing the strength domain $\K_0$ with a degraded strength domain $\K(\alpha) \subseteq \K_0$, which  shrinks as $\alpha$ increases.
\end{itemize}

Classical phase-field models of fracture typically assume an unbounded strength domain $\K_0\equiv\mathbb{M}_s^3$ and introduce stiffness degradation~\cite{BouFraMar00,TanLiBou18}. In contrast, we adopt here strength degradation without modifying the stiffness tensor, thereby degrading the elastic energy through a damage-like variable that preserves the linear elastic stiffness near the origin but reduces the energy growth rate at infinity, as illustrated in Figure~\ref{fig:strength-degradation}.
\begin{figure}[tp]
  \centering
  \includegraphics[width=0.43\textwidth]{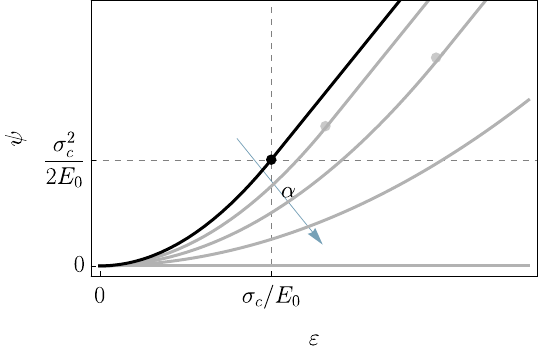}
  \includegraphics[width=0.43\textwidth]{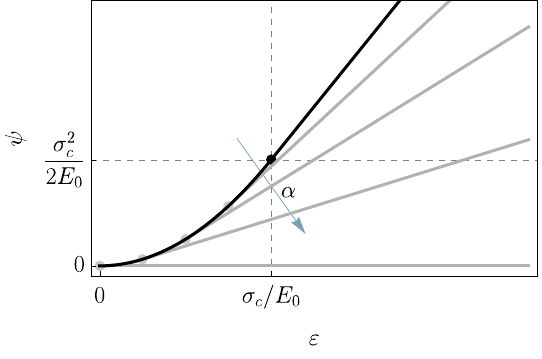}
\caption{Plot of the elastic energy $\psi$ as a function of the strain $\varepsilon$ for uniaxial tension: Stiffness degradation (left) reducing the curvature of the function close to the origin,  vs. strength degradation (right) reducing the slope at infinity. Here, $\varepsilon$ denotes strain, $\sigma_c$ the uniaxial strength, and $\mathsf{E}_0$ the initial stiffness. The black lines represent the energy defined by Equation~\eqref{eq:defPhi0} for $\alpha = 0$. Gray lines show the energy for increasing values of the damage parameter $\alpha$. For each curve, the dots mark the transition between the regime with quadratic growth and linear growth. For increasing $\alpha$, these critical strains \corr{and stresses} go to zero for the strength degradation\corr{; for} the classical stiffness degradation\corr{, the critical strains diverge to infinity and the critical stress remains $\sigma_c$}. }

  \label{fig:strength-degradation}
\end{figure}

For simplicity, we assume that the degradation occurs through a {homothetic contraction} of the initial strength domain $\K_0$, governed by a monotonically decreasing \emph{degradation function} $\k: \alpha \mapsto \k(\alpha) \in [0,1]$ such that $\k(0) =1$ and $\k(1) = 0$. Accordingly, we define
\[
\K(\alpha) = \k(\alpha) \K_0.
\]

Since $\H_{\K(\alpha)}(\p)=\k(\alpha)\H_{\K_0}(\p)$, we introduce the degraded elastic energy density
\begin{equation}
  \label{eq:phialpha}
\psi(\veps,\alpha)=\min_{\p\in\mathbb{M}^3_s}\varphi(\veps,\p,\alpha),
\quad
\varphi(\veps,\p,\alpha)=\dfrac{\A}{2}(\veps-\p)\cdot(\veps-\p)+\k(\alpha)\H_{\K_0}(\p)\corr{,}
\end{equation}
\corr{with the convention $0\cdot(+\infty)=+\infty$ at $\alpha=1$, which is the limit of a residual strength $\eta\K_0$ as $\eta\to0$. At full damage the constraint $\H_{\K_0}(\p)<\infty$ remains: when $\K_0$ contains all hydrostatic compressions, it forbids a negative inelastic volume change.}

As material degradation reduces the elastic energy density, we postulate that it is energetically offset by a phase-field approximation of a surface energy, given by
\begin{equation}
  \label{eq:Dell}
D_\ell(\alpha, \nabla\alpha) = \frac{\Gc}{4c_\w} \left( \frac{\w(\alpha)}{\ell} + \ell\, \nabla\alpha \cdot \nabla\alpha \right),
\end{equation}
where $\w$ is a monotonically increasing local dissipation potential satisfying $\w(0) = 0$ and $\w(1) = 1$, $\Gc>0$ is a material constant, $\ell > 0$ is the regularization length scale, and $c_\w:=\int_0^1\sqrt{\w(\alpha)}\mathrm{d}\alpha>0$ is a dimensionless normalization constant. As in classical phase-field models, we will show in subsequent sections that $\Gc$ can be interpreted as the fracture toughness, \emph{i.e.}, the energy dissipated per unit surface in crack-like solutions.

The total  energy density is the sum of~\eqref{eq:phialpha} and~\eqref{eq:Dell}:
\begin{equation}
  \label{eq:Well}
  W_\ell(\veps,\p,\alpha, \nabla\alpha)=\varphi(\veps,\p,\alpha)+D_\ell(\alpha, \nabla\alpha).
\end{equation}
The properties of the model depend on the strength degradation function $\k$ and the dissipation potential $\w$. We ask that they satisfy the following minimal requirements.

\begin{hypothesis}[Constitutive assumption: strength degradation and dissipation potential]
\label{hyp:kw}
We assume that $\K(\alpha) = \k(\alpha) \K_0$ and the constitutive functions $\k: \alpha \in [0,1] \to [0,1]$ and $\w: \alpha \in [0,1] \to [0,1]$ are \corr{of class $C^1$ on $[0,1]$ and of class $C^2$ on $(0,1)$}, with
\begin{equation}
  \label{eq:kwassumptions}
 \w'(\alpha)> 0,\quad\k'(\alpha)< 0,\quad \corr{\w''(\alpha)\k'(\alpha)-\w'(\alpha)\k''(\alpha)<0,
 \quad \w'(\alpha)\k'(\alpha)-2\w(\alpha)\k''(\alpha)<0}\qquad\forall{\alpha\in(0,1)};
\end{equation}
$\k(\alpha)$ being monotonically decreasing from $\k(0)=1$ to $\k(1)=0$ and $\w(\alpha)$ monotonically increasing from $\w(0)=0$ to $\w(1)=1$\corr{, and $\sqrt{\w(\alpha)}/\vert\k'(\alpha)\vert\to0$ as $\alpha\to0$}.
\end{hypothesis}

In all that follows, we use the following two one-parameter families of constitutive laws to illustrate the properties of our model. In both cases $\zeta$ is a scalar model parameter:
\begin{alignat}{6}
\text{Family } \mathsf{M1}: \quad
& \k(\alpha) &=& 1 - \alpha, \qquad
& \w(\alpha) &= (1 - \zeta)\alpha + \zeta \alpha^2,& \quad&& \zeta \in (0, 1]
\label{eq:model-1} \\
\text{Family }\mathsf{M2}: \quad
& \k(\alpha) &=& 1 - \alpha^{\zeta}, \qquad
& \w(\alpha) &= \alpha^2,& \quad &&\zeta \in [1, 2)
\label{eq:model-LS}
\end{alignat}
Both models \corr{satisfy} Hypothesis~\ref{hyp:kw}\corr{: the last two quantities in~\eqref{eq:kwassumptions} are $-2\zeta$ and $-(1-\zeta)-2\zeta\alpha$ for $\mathsf{M1}$, and $2\zeta(\zeta-2)\alpha^{\zeta-1}$ and $2\zeta(\zeta-2)\alpha^{\zeta}$ for $\mathsf{M2}$; the ratio $\sqrt{\w}/\vert\k'\vert$ is $\sqrt{(1-\zeta)\alpha+\zeta\alpha^2}$ for $\mathsf{M1}$ and $\alpha^{2-\zeta}/\zeta$ for $\mathsf{M2}$, and vanishes as $\alpha\to0$}. Moreover, the two model families above intersect when $\zeta = 1$ in both~\eqref{eq:model-1} and~\eqref{eq:model-LS}, yielding
\begin{equation}
   \k(\alpha) = 1 - \alpha, \quad 
   \w(\alpha) = \alpha^2, \quad
   c_\w = \frac{1}{2}.
   \label{eq:model-LS-zeta-1}
\end{equation}
\begin{remark} 
  \corr{The last two conditions in~\eqref{eq:kwassumptions} are equivalent to} 
\begin{equation}
  \dfrac{\mathrm{d}}{\mathrm{d}\alpha}\dfrac{{\w'(\alpha)}}{\k'(\alpha)}=\frac{\w''(\alpha)\k'(\alpha)-\w'(\alpha)\k''(\alpha)}{\k'(\alpha)^2} < 0,\qquad
  \dfrac{\mathrm{d}}{\mathrm{d}\alpha}\dfrac{\sqrt{\w(\alpha)}}{\k'(\alpha)}
  =
  \corr{\frac{\w'(\alpha)\k'(\alpha)-2\w(\alpha)\k''(\alpha)}{2\sqrt{\w(\alpha)}\k'(\alpha)^2}}
  < 0,
  \label{eq:slopesfromkw}
\end{equation}
\corr{that is, to the increase of $\w'/|\k'|$ and of $\sqrt{\w}/|\k'|$, used} in Section~\ref{sec:fundamentalProblem}. \corr{They hold in particular if $\k''\ge0$ and $\w''\ge0$, with $\k''$ and $\w''$ not simultaneously vanishing.}
These two quantities determine the dependence on $\alpha$ of the slope of the stress-strain response in the homogeneous solution and that of the cohesive law in the crack-like one, respectively, see Section~\ref{sec:model-problem-hom} and Section~\ref{sec:model-problem-loc}.
\end{remark}
\subsection{Static variational formulation}
\label{sec:model-static}
Let us consider a structure occupying the domain $\Omega\subset\R^3$, with imposed Dirichlet boundary data $\u^D$ on a part of the boundary $\partial_D\Omega$. 
We define the phase-field energy functional
\begin{equation}
  \label{eq:phase-field-energy}
  \mathcal{E}_\ell(\u,\alpha):=\int_\Omega {\psi(\veps(\u),\alpha)+D_\ell(\alpha,\nabla\alpha)}\,\mathrm{d}V
  =\min_{\p\in\mathbb{M}_s^3}\mathcal{E}_\ell(\u,\p,\alpha), 
\end{equation}
with
\begin{align}
  \notag
  \mathcal{E}_\ell(\u,\p,\alpha)&:= \int_\Omega  \varphi(\veps(\u),\p,\alpha) + D_\ell(\alpha,\nabla\alpha)\,\mathrm{d}V\\
  & = \int_\Omega \left(\dfrac{\A}{2}(\veps\corr{(\u)}-\p)\cdot(\veps\corr{(\u)}-\p)+\k(\alpha)\H_{\K_0}(\p)+\frac{\Gc}{4c_\w} \left( \frac{\w(\alpha)}{\ell} + \ell\, \nabla\alpha \cdot \nabla\alpha \right)\right)\,\mathrm{d}V.
  \label{eq:phase-field-energy3}
\end{align}
A formal variational formulation of the static phase-field fracture problem can be written as a minimization problem for the total energy functional, in one of the two equivalent forms:
\begin{equation*}
  \inf_{\u \in \mathcal{C},\, 
  \alpha \in \mathcal{D}} \mathcal{E}_\ell(\u,\alpha) \equiv
  \inf_{\u \in \mathcal{C},\, 
  \alpha \in \mathcal{D},\, \p \in \mathcal{P}} \mathcal{E}_\ell(\u,\p,\alpha),
\end{equation*}
where $\mathcal{C}$, $\mathcal{D}$, and $\mathcal{P}$ denote suitable functional spaces for admissible displacement fields, phase-fields, and nonlinear deformations, respectively, with  $\mathcal{C}$ incorporating the Dirichlet boundary condition on $\u$.

Hints for the appropriate functional setting are provided by the variational theory of perfect plasticity, which corresponds to a special case of the problem above when $\alpha = 0$:
\begin{equation}
  \label{eq:plasticityVariational}
  \inf_{\u \in \mathcal{C}, \,\p \in \mathcal{P}} \int_\Omega \frac{\A}{2}(\veps\corr{(\u)} - \p) \cdot (\veps\corr{(\u)} - \p) + \H_{\K_0}(\p) \, \dV.
\end{equation}
As for variational plasticity models, $\veps$ and $\p$ may admit singular parts which must coincide so that $\veps-\p$ remains square-integrable.
Therefore, we seek solutions of~\eqref{eq:plasticityVariational} with
  $\mathcal{C}\subseteq BD(\Omega; \mathbb{R}^3)$ and  $\mathcal{P}\subseteq\mathcal{M}(\Omega; \mathbb{M}^3_{\mathrm{s}})$, where $BD$ is the space of functions with bounded deformation and $\mathcal{M}$ is the set of Radon measures~\cite{Suq81,DalDeSMor06}. For the space of admissible phase-fields $\mathcal{D}$, we take  the functions with square-integrable first derivatives, taking values in the interval $[0,1]$: $\mathcal{D} \subseteq H^1(\Omega; [0,1])$.
Furthermore, admissible displacement fields must satisfy prescribed Dirichlet boundary conditions on $\partial_D\Omega$. As in the theory of perfect plasticity, the formulation of such boundary conditions requires particular care due to the potential concentration of the nonlinear strain $\p$ on $\partial_D\Omega$. These technical aspects are not addressed in detail here; we refer the interested reader to~\cite{DalDeSMor06,Francfort2015} for a comprehensive treatment.

\begin{remark}[Cantor-type terms]
The structure of the space $BD$ allows for an additional Cantor-type term in the strain decomposition~\eqref{eq:strain-SR-decomposition}; see~\cite{DalDeSMor06}. In the present work, we neglect this term and proceed with a formal analysis, focusing on specific classes of solutions in which Cantor terms do not appear.
\end{remark}

\subsection{Jump conditions}
\label{sec:jumpconditions}

Crack-like solutions include a displacement field with jump discontinuities along a set $J_\u$ of co-dimension 1 (\emph{i.e.} a union of surfaces in dimension 3 and curves in dimension 2). In the present modeling framework, they are related to a singular concentration of the  strain field on  $J_\u$.
Assuming that $J_\u$ is regular enough that it admits a normal vector at almost every point, integrability of the strain energy density~\eqref{eq:phialpha} and the positive definiteness of $\A$ then mandate that 
\begin{equation}
  \label{eq:psingular}
  \p^S = \veps^S(\u) = (\jump{\u} \odot \n)\, \delta_{J_\u}
\end{equation}
almost everywhere on $J_\u$, $\mathbf{a} \odot \mathbf{b} := \frac{1}{2}(\mathbf{a} \otimes \mathbf{b} + \mathbf{b} \otimes \mathbf{a})$ denoting the symmetrized tensor product of two vectors $\mathbf{a}, \mathbf{b} \in \R^n$, and $\delta_{J_\u}$ the Dirac-delta measure concentrated on $J_\u$.

The geometric compatibility condition~\eqref{eq:psingular} implies  a specific structure of the singular part of the strain tensor. This is formalized in the following definition that will play a pivotal role in disclosing the properties of the proposed phase-field modeling and the related crack nucleation conditions.

\begin{definition}[Direction of deformation compatible with a displacement jump]
    \label{def:jumpcompatibility} A non-zero  deformation $\ns$ is said to be compatible with a displacement jump if there exists a unit vector $\mathbf n$ and a vector $\mathbf d$  such that 
    \begin{equation}
    \label{eq:ndelta}
    \ns=\mathbf d\odot\mathbf n:=\frac12\left(\mathbf d\otimes\mathbf n+\mathbf n\otimes\mathbf d\right),\quad 
    \text{with}\quad
    \mathbf d\cdot\n\ge0,
    \end{equation}
  where the  condition $\mathbf{d} \cdot \n \ge 0$ encodes the non-interpenetration of matter in the geometrically linear setting.
\end{definition}

Because of Proposition~\ref{prop:plasticityConstitutive}, the nonlinear deformation $\p$ can be different from zero only for a stress state $\sig\in\partial\K(\alpha)$ and must satisfy the normality condition~\eqref{eq:pNormalCone}, \emph{i.e.} it must be collinear to one of the normals (or the normal if the boundary is regular) to $\partial\K(\alpha)$, see $\eqref{eq:elastoplasticlaw}$. 

Let us consider a unit-norm stress tensor $\s$ defining a direction in the stress space along which $\K_0$ is bounded, \emph{i.e.}, $\exists\,\sigma_c^\s > 0: \, \sigma_c^\s\,\s \in \partial\K_0$. Hence, we refer to $\s$ as a loading direction in the stress space compatible with a displacement jump if there exists a normal $\boldsymbol{\nu}^\s$ to $\partial\K_0$ at $\sigma_c^\s\s$ that is compatible with a displacement jump.
We can therefore conclude that, for a stress state \( \sig \in \partial \K_0 \), a displacement jump is admissible only if there exists a unit normal \( \nus \) to \( \partial \K_0 \) at \( \boldsymbol{\sigma} \), and two vectors \( \mathbf{n} \) and \( \mathbf{d} \), with $\n\cdot\n=1$ and $\mathbf{d}\not = \mathbf{0}$, such that $ \nus $ is in the form~\eqref{eq:ndelta}. Depending on the structure of the admissible normals, the jump may be purely tangential to the jump set $J_\u$—as in the case of a slip line—or may also include an opening component normal to $J_\u$.

We will provide a  detailed characterization of these properties in Section~\ref{sec:strength}.

\begin{remark}[On the analogy with limit analysis]
Solutions with displacement jumps are widely used in the theory of \emph{limit analysis} to derive bounds on the maximum supportable loads in rigid-perfectly plastic materials~\cite{DruPra52,DruPraGre52,Sal83}. The limit analysis problem coincides with the limit of the elasto-plastic problem~\eqref{eq:plasticityVariational} for $\A\to\infty$, which formally consists in replacing the quadratic term in~\eqref{eq:plasticityVariational} with the  constraint $\veps(\u)=\p$. We will further draw on fundamental concepts from this theory in Section~\ref{sec:strength}.
More recently,~\cite{Francfort2015,FRANCFORT2016125} discussed in detail the role of  displacement jump discontinuities in von Mises elasto-plasticity.
We do not view the jump discontinuities of such elasto-plastic solutions as cracks because stresses remain unaffected.
In a fracture setting, \emph{cohesive cracks} are characterized by displacement discontinuities and stress vectors that decrease and vanish as the opening increases.
In contrast, stresses vanish along \emph{Griffith cracks}, regardless of the magnitude of the displacement jumps.
\end{remark}

\subsection{Quasi-static variational formulation}
\label{sec:evolution-phase-field}

In the rest of this work, we consider solutions to quasi-static evolution problems in which the loading is parameterized by a scalar time-like variable $t$. To fix ideas, let us consider the case of a prescribed displacement $\u^D_t$ on the Dirichlet boundary $\partial_D\Omega$ and denote by $\mathcal{C}_t$ the corresponding affine space of admissible displacements. Let $\mathbf U_t$ be a particular element of $\mathcal{C}_t$, respecting the condition $\mathbf U_t=\u^D_t$ on $\partial_D\Omega$. Hence, we denote by $\mathcal{C}^0$ the vector space of admissible variations such that $\v+\mathbf U_t\in\mathcal{C}_t$, $\forall \v\in\mathcal{C}^0$.

When including the presence of external dead loads modeled by the linear functional $\mathcal{W}_{\mathrm{ext}}^t(\mathbf u)$, we denote the energy of the structure at time $t$ in the state $(\mathbf u,\p,\alpha)$ as
\begin{equation}
  \label{eq:Et}
  \quad\mathcal E_\ell^t(\mathbf u,\p,\alpha):=\int_\Omega W_\ell(\veps(\mathbf u),\p,\alpha,\nabla\alpha)\,\dV - \mathcal{W}_{\mathrm{ext}}^t(\mathbf u - \mathbf U_t).
\end{equation} 

We seek, at each time $t \ge 0$, the \corr{displacement, nonlinear deformation, and phase-field} $(\u_t,\p_t,\alpha_t)\in\mathcal C_t\times\mathcal{P}\times \mathcal{D}$ that satisfy the following three governing principles:
\begin{enumerate}[label=(\textbf{\alph*}), ref=\textbf{\alph*}, leftmargin=*, widest=EB]
\begin{subequations}
  \label{eq:QS-evolution}
  \item[(IR)] {\it Irreversibility condition}: 
    \begin{equation*} 
      t \mapsto \alpha_t \text{ is non-decreasing and takes values in } [0,1];
    \end{equation*}
  \item[(ST)] {\it Stability criterion}: For  $(\mathbf{v},\mathbf{q},\beta)\in\mathcal{C}_t\times\mathcal{P}\times \mathcal{D}$; 
  \begin{equation}
    \label{eq:ST}
    \mathcal E_\ell^t(\u_t,\p_t,\alpha_t)\le \mathcal E_\ell^t(\mathbf v,\mathbf{q},\beta), \quad \forall (\mathbf v,\mathbf{q},\beta) \text{ in a neighborhood of } (\u_t,\p_t,\alpha_t)\text{ with }\alpha_t\le\beta\le 1;
  \end{equation}
  \item[(EB)] {\it Energy balance}:
  \begin{equation}
    \label{eq:EB}
     \mathcal E_\ell^t(\corr{\u_t},\p_t,\alpha_t)=\mathcal E_\ell^0(\u_0,\p_0,\alpha_0)+\int_0^t\frac{\partial\bar{\mathcal E}_\ell^\tau}{\partial \tau}(\u_\tau-\mathbf U_\tau,\corr{\p_\tau,}\alpha_\tau)\,\mathrm{d}\tau,
  \end{equation}
  where, for $\bar{\mathbf v} \in \mathcal C^0$, $\bar{\mathcal E}_\ell^t(\bar{\mathbf v},\mathbf{q},\beta) := \mathcal E_\ell^t(\bar{\mathbf v} + \mathbf U_t,\mathbf{q},\beta)$\corr{, $\bar\u_t:=\u_t-\mathbf U_t\in\mathcal C^0$,} and, assuming the time-dependence of the data is regular,
  \end{subequations}
  \begin{equation*}
  \frac{\partial\bar{\mathcal E}_\ell^t}{\partial t}(\bar\u_t,\mathbf{p}_t,\alpha_t) = \int_\Omega \sig_t \cdot \veps(\dot{\mathbf U}_t)\,\dV - \dot{\mathcal{W}}_{\mathrm{ext}}^t(\bar{\u}_t),
  \end{equation*}
  with \corr{$\dot{\mathcal{W}}_{\mathrm{ext}}^t$ the time derivative of the loading functional at fixed argument, and}
  $$\sig_t ={\A}(\veps(\u_t)-\p_t)\corr{,\quad \p_t=\argmin_{\p\in\mathbb{M}^3_s}}\left(\frac{\A}{2}(\veps(\u_t)-\p)\cdot(\veps(\u_t)-\p)+\k(\alpha_t)\H_{\K_0}(\p)\right).
  $$
\end{enumerate} 

\begin{remark}[On the stability condition]
The notion of proximity introduced in the stability condition must be made precise, which requires the choice of a norm on the space of displacement and phase-fields. We do not specify this choice here. 
\end{remark}

Using standard localization arguments, one can show that the stability principle~\eqref{eq:ST} requires the following first-order necessary stability conditions to be satisfied almost everywhere in  $\Omega$:
\begin{itemize}
  \item \emph{Equilibrium equation}, obtained by imposing stability with respect to the displacement $\mathbf{u}$: 
  $$
  \corr{\mathbf{div}\sig_t+\mathbf{f}_t=\mathbf{0}},\quad\sig_t=\A(\veps(\mathbf{u}_t)-\p_t),
  $$
  where $\mathbf{f}_t$ \corr{is the bulk force density, $\mathcal{W}_{\mathrm{ext}}^t(\u)=\int_\Omega\mathbf{f}_t\cdot\u\,\dV+\dots$}, supposed null in the rest of this work.
  \item   \emph{Criterion for the evolution of the nonlinear deformation}, obtained by imposing stability with respect to $\p$: 
  $$
  \corr{\sig_t\in\K(\alpha_t),\qquad}
  \p_t=\begin{cases}
        \mathbf{0},&\text{if\quad}\sig_t\in\corr{\mathring{\K}(\alpha_t)}\\
        \lambda_t\,\corr{\boldsymbol{\nu}_{\sig_t/\k(\alpha_t)}},&\text{if\quad}\sig_t\in\corr{\partial\K(\alpha_t)}\\
        \end{cases},\quad\lambda_t\geq 0\corr{,}
  $$
  \corr{for $\alpha_t<1$, where $\boldsymbol{\nu}_{\sig_t/\k(\alpha_t)}$ is an outer normal to $\partial\K_0$ at $\sig_t/\k(\alpha_t)$, \emph{i.e.} to $\partial\K(\alpha_t)$ at $\sig_t$.}
  \item \emph{Damage criterion}, obtained as stability with respect to the phase-field variable $\alpha$
  \begin{equation}
  \k'(\alpha_t)\H_{\K_0}(\p_t) 
    + 
    \dfrac{\Gc}{4c_\w\ell}{\w'(\alpha_t)}-
    \dfrac{\Gc\ell}{2c_\w}\,\Delta\alpha_t
    \geq 0.
    \label{eq:stab-alpha}
  \end{equation}
\end{itemize}
Furthermore, \corr{combined with (IR) and the first-order stability conditions,} the energy balance~\eqref{eq:EB} implies the following \corr{consistency} condition, \corr{which} states that the damage field can evolve only  where the damage criterion is satisfied as an equality
\begin{equation}
    \left(
      \k'(\alpha_t)\H_{\K_0}(\p_t) 
    + 
    \dfrac{\Gc}{4c_\w\ell}{\w'(\alpha_t)}-
    \dfrac{\Gc\ell}{2c_\w}\,\Delta\alpha_t
 \right) \,\dot\alpha_t =0.
  \label{eq:eb-alpha}
\end{equation}

In the following Section~\ref{sec:fundamentalProblem}, we will discuss and solve for a specific multiaxial model problem the set of evolution conditions  above together with specific boundary conditions. 

\begin{remark}[Two-field formulation]
As an alternative, it can be useful to define the energy and the evolution in terms of the two fields $(\u,\alpha)$, by  defining $\mathcal{E}_\ell^t(\u,\alpha)$ after the elimination of $\p_t$ from the formulation above.
\end{remark}
\begin{remark}[On external forces and global stability]
When the stability condition (ST) is expressed as a global minimization of the energy functional $\mathcal{E}_\ell^t$, external forces cannot be included, as they may drive the energy to $-\infty$, preventing the existence of a minimizer.
Additionally, global stability can be physically unrealistic, as it neglects energy barriers and metastable states that naturally arise in rate-independent or irreversible processes. 
\end{remark}

\section{Multiaxial model problem}
\label{sec:fundamentalProblem}

In order to establish the key properties of the phase-field model proposed in Section~\ref{sec:Model}, we consider two representative classes of solutions under a constant stress: \emph{homogeneous} solutions with uniform strain and damage, and \emph{localized} solutions involving a displacement discontinuity and damage localization. To construct them, we formulate an ad-hoc boundary value problem on a cube \corr{with periodic boundary conditions} that models the evolution of internal fields far from external boundaries, imposing the direction of the loading in the stress space and using the average deformation of the cube as the loading parameter. 
This formulation enables the extension of techniques for the construction of explicit analytical solutions, developed previously for one-dimensional problems~\cite{Pham-Marigo-EtAl-2011a,AleMarVid14}, to the multi-axial setting.

This multiaxial model problem is introduced in Section~\ref{sec:model-problem-formulation}. 
Sections~\ref{sec:model-problem-hom} through~\ref{sec:model-problem-loc} present the homogeneous and crack-like solutions to the associated evolution problem. As discussed in Section~\ref{sec:jumpconditions}, the nucleation of a crack is linked to the admissibility of displacement discontinuities, which depends on the stress orientation and the properties of the strength surface. 
In Section~\ref{sec:model-problem-stab}, we discuss the stability of the solutions.  The analysis of the energy evolution in the localized regime will identify an equivalent cohesive law that characterizes the crack behavior. 

\subsection{Problem formulation} 
\label{sec:model-problem-formulation}

\begin{figure}[t] 
   \begin{center}
    \includegraphics[width=2.5in]{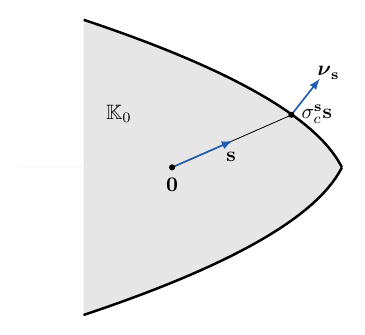} \includegraphics[width=3.in]{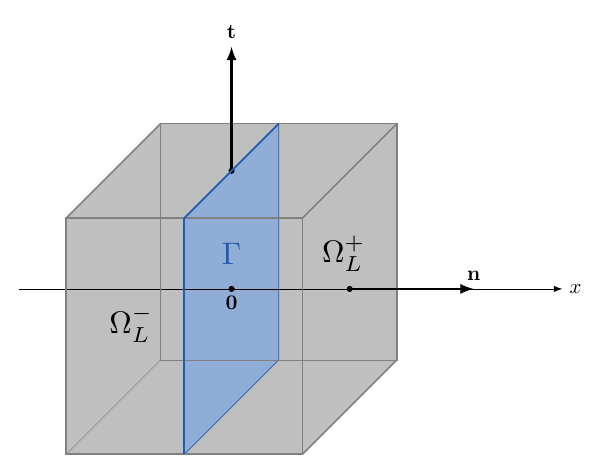} 
   \caption{   \label{ref:F-PBCube}
The multiaxial model problem. Left: A strength domain $\K_0$ in the stress space showing the loading direction $\s$, intersecting the boundary $\partial\K_0$ at the point $\corr{\sig_0} = \sigma_c^\s \s$ with normal $\boldsymbol{\nu}^\s$. Right: the domain $\Omega_L$ of the model problem; if the normal $\boldsymbol{\nu}^\s$ is compatible with a displacement jump, one face of the cube is oriented along a direction $\n$ such that $\boldsymbol{\nu}^\s = \mathbf{d} \odot \n$. In this case, we construct an explicit solution with a displacement jump $\mathbf{d}$ across the surface $\Gamma$, where the phase-field reaches its maximum and the nonlinear deformations concentrate as a measure; the corresponding bulk deformations and stress are uniform on $\Omega_L\setminus\Gamma$.
   }
   \end{center}
\end{figure}
We consider a cube $\Omega_L$ with side length $L$, centered at the origin.
Let $\mathbf{n}$ be a unit vector orthogonal to one of its faces and $x$ be the coordinate in the direction of $\mathbf{n}$ (see Figure~\ref{ref:F-PBCube}-right). \corr{The cube is a periodic cell: admissible displacements are of the form $\u(\x)=\mathbf E\x+\mathbf w(\x)$, with $\mathbf E\in\mathbb{M}^3_s$ a constant homogeneous deformation and $\mathbf w$ $\Omega_L$-periodic, and admissible phase-fields are $\Omega_L$-periodic.  Periodicity models a material point far from external boundaries.}
Given a displacement field $\u$ and $\s\in \mathbb{M}^3_s$ such that  $\s\cdot\s=1$,  we denote by
\begin{equation}
   \label{eq:eps_s}
   \bar\varepsilon^\s(\u):=\frac{1}{L^3} \int_{\Omega_L} \s \cdot \veps(\u) \, \mathrm{d}V 
\end{equation}
the average deformation of the cube in the direction $\s$. \corr{Here $\veps(\u)$ is the strain measure of $\u$, jump part included (see~\eqref{eq:strain-SR-decomposition}); by the divergence theorem and the periodicity of $\mathbf w$, $\bar\varepsilon^\s(\u)=\s\cdot\mathbf E$.}

The following problem plays a central role in establishing the fundamental properties of the model and in characterizing its asymptotic behavior as $\ell \to 0$.

\begin{problem}{(Multiaxial model problem)}
   \label{pb:fundamental-problem} 
   Given 
 the model parameters $( \A,\K_0,\k,\w,\Gc,\ell)$,
 a cube $\Omega_L$ of size $L$ and orientation $\n$,
 a loading direction $\s\in \mathbb{M}^3_s$ with $\s\cdot\s=1$,
 an undamaged initial state $\mathbf{z}_0=(\u_0=\mathbf{0},\p_0=\mathbf{0},\alpha_0=0)$,
 determine a solution $\mathbf{z}_t=(\u_t,\p_t,\alpha_t)$ of the quasi-static evolution problem ~\eqref{eq:QS-evolution} for the energy functional $\mathcal{E}_\ell^t(\mathbf{v}, \mathbf{q}, \beta) $ with null body forces \corr{and periodic boundary conditions}
under the constraint of imposed \emph{average deformation $t$ in the direction $\s$}:
\begin{equation}
   \label{eq:loading-condition}
\u_t \in\mathcal{C}_{t} = \left\{ \mathbf{v}\corr{=\mathbf E\x+\mathbf w:\ \mathbf E\in\mathbb{M}^3_s,\ \mathbf w\ \Omega_L\text{-periodic},}\ \bar\varepsilon^\s (\mathbf{v})=t  \right\}.
\end{equation}
\end{problem}

The first-order stability condition implies equilibrium, and we recover $\textbf{div}\, \sig_t = \mathbf{0}$ in $\Omega_L$\corr{, with $\sig_t\bar{\mathbf n}$ anti-periodic on $\partial\Omega_L$, $\bar{\mathbf n}$ being the outward unit normal. The competitor $(\u_t+h\,\mathbf G\x,\p_t,\alpha_t)$, with $\mathbf G\in\mathbb{M}^3_s$ and $\s\cdot\mathbf G=0$, satisfies the loading condition~\eqref{eq:loading-condition} for every $h\in\mathbb R$; at first order in $h$, \eqref{eq:ST} gives $\big(\int_{\Omega_L}\sig_t\,\dV\big)\cdot\mathbf G=0$. Hence the average stress $\bar\sig_t:=L^{-3}\int_{\Omega_L}\sig_t\,\dV$ is $\bar\sig_t=\sigma_t\s$, with $\sigma_t\in\mathbb R$: the loading condition~\eqref{eq:loading-condition} fixes its direction $\s$, and the average deformation $t$ controls its amplitude $\sigma_t$. A uniform stress field $\sigma_t\s$ satisfies these conditions.} \corr{In the evolution problem~\eqref{eq:QS-evolution}, the particular element of $\mathcal C_t$ is $\mathbf U_t=t\,\s\x$, and the space of admissible variations is $\mathcal C^0=\{\mathbf E\x+\mathbf w:\ \mathbf E\in\mathbb{M}^3_s,\ \s\cdot\mathbf E=0,\ \mathbf w\ \Omega_L\text{-periodic}\}$. With null external loads, the energy balance~\eqref{eq:EB} involves $\partial_t\bar{\mathcal E}_\ell^t=\int_{\Omega_L}\sig_t\cdot\s\,\dV=L^3\sigma_t$.}\\

\noindent We will establish analytical solutions to Problem~\ref{pb:fundamental-problem} with the following procedure:
\begin{enumerate}
   \item We select a stress direction $\s$ along which $\K_0$ is bounded. We denote by $\sig_0$ the intersection of $\partial \K_0$ with the line through $\mathbf{0}$ in the direction $\s$.
   Let $\boldsymbol{\nu}^\s$ be a unit normal to $\partial \K_0$ at $\sig_0$ and $\sics = \|\sig_0\|$ (see Figure~\ref{ref:F-PBCube}-left):
   \begin{equation*}
      \sig_0 = \sics \s \in \partial \K_0.
   \end{equation*}
   \item We then identify two classes of solutions that satisfy the first-order stability condition, the irreversibility condition, and the energy balance for Problem~\ref{pb:fundamental-problem}:
   \begin{itemize}
      \item the \emph{homogeneous response}, \emph{i.e.}~a solution with spatially constant strain and damage; under our assumptions on the model parameters, this solution is unique and independent of the cube orientation $\n$ and size $L$.
      
      \item solutions with \emph{displacement jumps}, which exist only if the loading direction $\s$ is compatible with a displacement discontinuity, as defined in Section~\ref{sec:jumpconditions}, \emph{i.e.}, if there exist two vectors $\n$ and $\mathbf{d}$ such that the normal $\nus$ to the yield surface $\partial \K_0$ at $\boldsymbol{\sigma}_0 = \sics \s$ is in the form $\nus = \mathbf{d} \odot \mathbf{n}$. We construct an explicit analytical solution with displacement jumps by aligning the cube with one of these two vectors, say $\n$, and choosing a cube of sufficiently small size $L$. 
      The solution will be valid for a regularization length $\ell$ that is sufficiently small compared to $L$.
   \end{itemize}

   \item Finally, we analyze the stability of both the homogeneous and the localized solutions, verifying whether they satisfy the full minimality condition~\eqref{eq:ST}, and not merely its first-order version.
\end{enumerate}

Here and henceforth, we denote by
\begin{equation}
   A_0^\s:=\corr{\frac{1}{\mathsf{A}_0^{-1}\s \cdot\s}},\qquad  \nu^\s=\boldsymbol{\nu}^\s \cdot\s, 
\label{eq:A0sC0s}
\end{equation}
\corr{the stiffness in the direction $\s$ and the projection along $\s$ of the normal to the strength domain.} 
We introduce the \emph{elasto-cohesive length} in the direction $\s$:
\begin{equation*}
\ell_\mathrm{ch}^\s:=\frac{\Gc A_0^\s}{(\sigma_c^\s)^2}.
\end{equation*}

\begin{remark}[On periodic boundary conditions]
\corr{With free faces, small blocks cut off at the corners of the cube and translated by large jumps would satisfy the loading condition~\eqref{eq:loading-condition} at a vanishing surface cost. In the periodic cell, the loading condition involves only the mean deformation $\mathbf E$, to which such jumps do not contribute.}
\end{remark}

\subsection{The homogeneous response}
\label{sec:model-problem-hom}

We look for solutions with stress, strain, and phase-field  uniform in space and that evolve with the loading $t$.  
The stress tensor is collinear with $\s$ and lies within the convex set $\k(\alpha_t)\K_0$:
$$
\sig_t(\mathbf{x}) = \sigma_t \s \quad \text{with} \quad \sigma_t \le \k(\alpha_t) \sics,
$$
where $\sigma_t$ is the unknown stress level.

The strain is obtained from the constitutive law and the normality condition, as in~\eqref{eq:elastoplasticlaw}:
$$
\veps_t(\mathbf{x}) = \sigma_t \A^{-1} \s + \lambda_t\, \boldsymbol{\nu}^\s \quad \text{with} \quad
\begin{cases}
\lambda_t  = 0 & \text{if } \sigma_t < \k(\alpha_t) \sics, \\
\lambda_t  \ge 0 & \text{if } \sigma_t = \k(\alpha_t) \sics,
\end{cases}
$$
where $\boldsymbol{\nu}^\s$ is a normal to $\partial \K_0$ at $\sics\s$ and $\lambda_t$ is the scalar plastic multiplier, giving the direction and the amplitude of the nonlinear contribution to the strain, respectively.
The loading condition~\corr{\eqref{eq:loading-condition}} simplifies to:
$$
 t=\s \cdot \veps_t =\sigma_t/A_0^\s + \lambda_t\,\nu^\s.
$$

The evolution of the phase-field is governed by the irreversibility condition, the first-order stability criterion, and the consistency condition, as in classical approaches~\cite[see \emph{e.g.}][]{MarMauPha16,AleMarVid15}. As long as $\alpha_t < 1$, these read as:
\begin{equation}
\dot\alpha_t \ge 0, \quad 
\k'(\alpha_t) \lambda_t \sics {\nu}^\s + \frac{\Gc}{4c_\w\ell} \w'(\alpha_t) \ge 0, \quad
\left( \k'(\alpha_t) \lambda_t \sics {\nu}^\s + \frac{\Gc}{4c_\w\ell} \w'(\alpha_t) \right) \dot\alpha_t = 0.
\label{eq:damage-criterion-H}
\end{equation}

Our goal is to determine the evolution of $\sigma_t$, $\lambda_t$, and $\alpha_t$ for $t \ge 0$, given that $\alpha_0 = 0$.
The response proceeds through three distinct phases:
(i) a purely elastic phase, where the stress remains inside $\K_0$ and no damage occurs;  
(ii) a nonlinear elastic phase, where the stress reaches the yield surface and nonlinear strains develop without damage;  
(iii) a damaging phase, where the phase-field evolves until complete failure.  

Assuming that ${\w'(\alpha)}/{\vert\k'(\alpha)\vert}$ is a monotonically increasing function of $\alpha$, the response can be described as follows:

\begin{enumerate}
\item \textit{Linear elastic response.}  
For small loading the response is linear and elastic, with vanishing  nonlinear deformation and damage: 
$$
t\in(0,{\varepsilon}_e^\s]:\qquad
\sigma_t =  A_0^\s\, {t},
\qquad  
\lambda_t = 0,\qquad\alpha_t = 0, \quad\text{with}\quad {\varepsilon}_e^\s=\sics/A_0^\s,
$$
where the \emph{linear elastic limit}  for the deformation in the direction $\s$, ${\varepsilon}_e^\s$, is such that $\sigma_t<\sics$.
\item \textit{Nonlinear elastic response without damage.}  
The stress state reaches the yield surface at the stress level $\sics\,\s$, and the imposed average deformation $t$ is accommodated at constant stress by the nonlinear strain of amplitude $\lambda_t$ along the direction $\boldsymbol{\nu}^\s$.  
In this damage-free regime, the nonlinear strain accumulates to satisfy the condition
$
t = \varepsilon_e^\s + \lambda_t\,{\nu}^\s.
$
Thus, the evolution is given explicitly by:
$$
t \in (\varepsilon_e^\s,\varepsilon_c^\s]:\qquad \sigma_t = \sics, \qquad
\lambda_t = \frac{t - {\varepsilon}_e^\s}{\nu^\s}, \qquad \alpha_t = 0.
$$
The upper bound for this regime, denoted as the \emph{critical deformation} ${\varepsilon}_c^\s$, is determined by the damage criterion:
$$
\k'(0)  \sics \,\lambda_t{\nu}^\s + \frac{\Gc}{4c_\w\ell} \w'(0) = 0,
$$
giving:
\begin{equation}
   \label{eq:epsc}
   {\varepsilon}_c^\s = {\varepsilon}_e^\s + \frac{\Gc}{4c_\w\ell\sics }\frac{\w'(0)}{|\k'(0)|},
\end{equation}
with ${\varepsilon}_c^\s\to\infty$ for $\ell \to 0$ \corr{if $\w'(0)>0$; if $\w'(0)=0$, damage starts at ${\varepsilon}_c^\s={\varepsilon}_e^\s$}.

\item \textit{Response with damage.}  
When the loading reaches the critical deformation ${\varepsilon}_c^\s$, the damage  must evolve to satisfy the conditions~\eqref{eq:damage-criterion-H}. In this regime the solution is given by 
\begin{equation}
   \label{eq:talphahom}
t \in (\varepsilon_c^\s,\varepsilon_u^\s):\qquad \sigma_t = \k(\alpha_t)\sics, \qquad
\lambda_t = \frac{t - \corr{\k(\alpha_t)}{\varepsilon}_e^\s}{\nu^\s}, \qquad t=\k(\alpha_t){\varepsilon}_e^\s+\frac{\Gc}{4c_\w\ell\sics } \frac{\w'(\alpha_t)}{\vert\k'(\alpha_t)\vert}.
\end{equation}
The latter equation represents the damage criterion and can be rewritten in the form:
$$
t=\frac{\Gc}{\ell\sics }
\left(
   \k(\alpha_t) \frac{\ell}{\ell_{\mathrm{ch}}^\s}
   +
   \frac{1}{4c_\w} \frac{\w'(\alpha_t)}{\vert\k'(\alpha_t)\vert}
   \right).
$$
Its derivative with respect to $t$ gives
$$
\dot\alpha_t=
   \frac{\ell\sics}{ \Gc}
   \frac{1}
   {\left(\frac{1}{4c_\w }
 \frac{\mathrm{d}}{\mathrm{d}\alpha}\left.\left(\frac{\w'(\alpha)}{\vert\k'(\corr{\alpha})\vert}\right)\right|_{\alpha = \alpha_t}
- \vert\k'(\alpha_t)\vert\frac{\ell}{\ell_{\mathrm{ch}}^\s}\right)}
   .
$$
To obtain a smooth homogeneous evolution respecting the irreversibility condition $\dot\alpha_t\geq 0$, the following condition must be verified:
\begin{equation}
\dfrac{\ell}{\ell_{\mathrm{ch}}^\s}
\leq
\inf_{\alpha\in(0,1)}
\left(
\dfrac{1}{4c_\w\vert \k'(\alpha)\vert}
 \dfrac{\mathrm{d}}{\mathrm{d}\alpha}
\left(\dfrac{\w'(\alpha)}{\vert\k'(\alpha)\vert}\right)\right),
\label{eq:homogeneous-snapback-condition}
\end{equation}
where, $\forall \alpha\in(0,1)$, $\k'(\alpha)<0$ and $\frac{\mathrm{d}}{\mathrm{d}\alpha}
\left(\frac{\w'(\alpha)}{\vert\k'(\alpha)\vert}\right)>0$ thanks to Hypothesis~\ref{hyp:kw},~Equation~\eqref{eq:slopesfromkw}. If the condition above is not verified, the solution can exhibit a snap-back in the homogeneous response.

This solution is valid until the \emph{ultimate deformation} $t=\varepsilon_u^\s$ such that $\alpha_t=1$, with
\begin{equation}
   \label{eq:epsu}
\varepsilon_u^\s=
\begin{cases}
   \frac{\Gc}{\sics \ell}\frac{1}{4c_\w}\frac{\w'(1)}{\vert\k'(1)\vert} &\text{if}\quad \vert \k'(1)\vert >0\\
   +\infty &\text{if}\quad \vert \k'(1)\vert =0.
\end{cases}
\end{equation}
\item \emph{Fully damaged state.} For $t>\varepsilon_u^\s$, the solution remains with $\alpha_t=1$ and $\sigma_t=0$ for any $t$. The nonlinear deformation accommodates the imposed strain.
\end{enumerate}

We detail below and in Figure~\ref{fig:modelproblem-solution} (gray lines) the homogeneous response for the model $\mathsf{M1}$ presented in~\eqref{eq:model-1}.
\begin{figure}[t]
   \centering
   \includegraphics[height=0.35\textwidth]{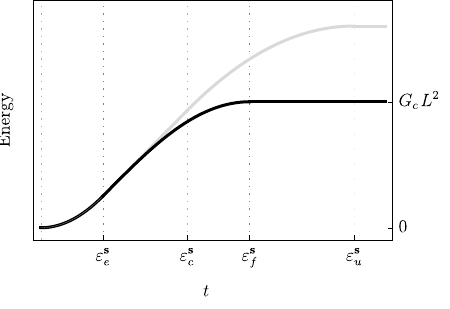}
   \includegraphics[height=0.35\textwidth]{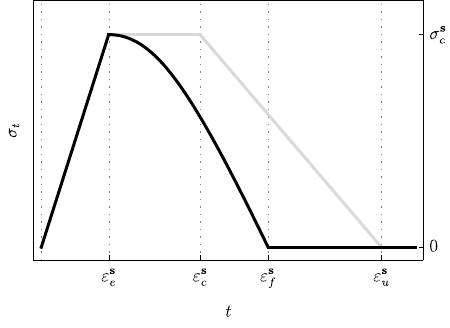}
   \caption{Solution of the multiaxial model problem: energy (left) and stress (right) versus the applied average strain $t$. Black: localized solution; Gray: homogeneous solution. 
   The numerical results are for $\k(\alpha)=1-\alpha$ and $\w(\alpha)=(\alpha+\alpha^2)/2$, corresponding to the model $\mathsf{M1}$ in~\eqref{eq:model-1} with $\zeta=1/2$, with the critical strains given in equations~\eqref{eq:model1-critical-strains}.  
   For $\zeta=1$ the phase with constant stress and linear energy in the homogeneous response degenerates, with $\varepsilon_c^\s = \varepsilon_e^\s = \sigma_c^\s / A_0^\s$.}
   \label{fig:modelproblem-solution}
 \end{figure}
  In this case, the damage criterion in~\eqref{eq:talphahom}  writes as:
   $$
   t = (1 - \alpha_t)\varepsilon_e^\s + \frac{\Gc}{4c_\w\ell\sics }(1 - \zeta + 2\zeta \alpha_t).
   $$
   The corresponding damage and stress evolutions are piecewise linear:
   $$
   \alpha_t=
   \begin{cases}
      0&\text{for}\quad t\in(0,\varepsilon_c^\s]\\
      \dfrac{t -\varepsilon_c^\s}{\varepsilon_u^\s- \varepsilon_c^\s}&\text{for}\quad t\in(\varepsilon_c^\s,\varepsilon_u^\s].\\
      1&\text{for}\quad t>\varepsilon_u^\s
   \end{cases}
   \qquad
   \sigma_t=
   \begin{cases}
      A_0^{\s}{t}&\text{for}\quad t\in(0,\varepsilon_e^\s ]\\
      \sics&\text{for}\quad t\in(\varepsilon_e^\s ,\varepsilon_c^\s]\\
      \sics\,\dfrac{\varepsilon_u^\s -t  }{\varepsilon_u^\s- \varepsilon_c^\s}\,&\text{for}\quad t\in(\varepsilon_c^\s,\varepsilon_u^\s].\\
      0&\text{for}\quad t>\varepsilon_u^\s
   \end{cases}
   $$
   where 
   \begin{equation}
   \varepsilon_c^\s = {\varepsilon}_e^\s + (1-\zeta)\frac{\Gc}{4c_\w\ell\sics },
   \qquad
   \varepsilon_u^\s=  (1 + \zeta)\frac{\Gc}{4c_\w\ell\sics }.
   \label{eq:model1-critical-strains}
   \end{equation}
   Figure~\ref{fig:modelproblem-solution} plots the total energy and stress-strain response for the case $\zeta=1/2$ (gray lines). 
   The evolution is monotonic in $t$ provided that $\varepsilon_u^\s > \varepsilon_c^\s$, which for $\zeta > 0$ holds if $\ell$ is sufficiently small.  
   If we take $\zeta = 0$, then $\varepsilon_u^\s < \varepsilon_c^\s$ for all $\ell$, and a snap-back necessarily occurs.  
   For $\zeta = 1$, $\varepsilon_c^\s = \varepsilon_e^\s$, and damage initiates immediately after the elastic phase. Snap-back does not occur if 
   $
   {\ell}<{\ell_{\mathrm{ch}}^\s}
   $. 

   For the model $\mathsf{M2}$ given in~\eqref{eq:model-LS} the damage criterion writes as   
   $$
t=\frac{\Gc}{\ell\sics }
\left(
   (1-\alpha_t^\zeta) \frac{\ell}{\ell_{\mathrm{ch}}^\s}
   +
    \frac{\alpha\corr{_t}^{2-\zeta}}{\zeta}
   \right).
$$
The homogeneous response is without snap-back provided that 
$$
\frac{\ell}{\ell_{\mathrm{ch}}^\s}< \frac{2-\zeta}{\zeta^2}.
$$
For $\zeta=1$, the solution is the same as for the model $\mathsf{M1}$; for $\zeta>1$ the solution for the damage field is not explicit and it is not reported here. 

\subsection{A crack-like solution when $\nus$ is compatible with displacement jumps}
\label{sec:model-problem-loc}
When the loading direction $\s$ is such that there exists a normal $\nus$ to the strength surface $\partial\K_0$ at the stress state $\sig= \sics\s$  compatible with a displacement jump, 
there exist two vectors $\n$ and $\mathbf{d}_0$ such that $\nus=(\n\odot\mathbf{d}_0)/\Vert\n\odot\mathbf{d}_0\Vert$, see~\eqref{eq:ndelta}. The vector $\n$ will represent the normal to the crack set and $\mathbf{d}_0$ the direction of the displacement jump. Here, we assume that both vectors have unit norm.

We determine below a crack-like solution  to the multiaxial model problem of Problem~\ref{pb:fundamental-problem} featuring displacement discontinuities and localized damage for load levels $t$ exceeding the elastic threshold, $t \ge \varepsilon_e^\s$.
To this end, we consider a cube $\Omega_L$ having a face with a normal $\n$ and start from a reduced set of assumptions:
      \begin{enumerate}
          \item The stress field is homogeneous: $\sig_t(\mathbf x) = \sigma_t \s$, for all $t \ge \varepsilon_e^\s$;
          \item The phase-field vanishes at $t = \varepsilon_e^\s$, and is no longer $0$ thereafter: $\alpha_{t} = 0$ for $t= \varepsilon_e^\s$, and $\alpha_t \not\equiv 0$ for all $t > \varepsilon_e^\s$;
          \item The phase-field attains its maximum on the cross-section $\Gamma$:
          $$
          0 \le \alpha_t(\mathbf x) < \bar\alpha_t = \max_{\mathbf y \in \Gamma} \alpha_t(\mathbf y), 
          \quad \forall t > \varepsilon_e^\s, 
          \quad \forall \mathbf x \not\in \Gamma:= \{\x\in\Omega_L: \x \cdot \n = 0 \};
          $$
          \item The evolution is smooth in time.
      \end{enumerate}

We will then show that for models respecting Hypothesis~\ref{hyp:kw} with
\begin{equation}
   \label{eq:kw-conditions-1}
   c^*_{\w}:=\int_0^1 \frac{1}{\sqrt{\w(\beta)}}\mathrm{d}\beta<+\infty,\quad
\end{equation}
and for $\ell/L$ and $L/\ell_\mathrm{ch}^\s$ small enough:
        
\begin{itemize}
\item The displacement experiences a uniform jump $\mathbf{d}_t=d_t\mathbf{d}_0$ on the planar surface $\Gamma$, where the damage is constant and equal to $\bar\alpha_t$ with
\begin{equation} \label{eq:delta-alpha-loc1}
   t = \frac{\Gc}{\sics L} \left(\k(\bar{\alpha}_t) \frac{L}{\ell_{\mathrm{ch}}^\s} + \frac{\sqrt{\w(\bar{\alpha}_t)}}{c_\w|\k'(\bar{\alpha}_t)|}\right),
   \qquad
   {d_t}=\frac{1}{\s\n\cdot\mathbf{d}_0}\frac{\Gc}{\sics }\frac{\sqrt{\w(\bar\alpha_t)}}{c_\w\vert\k'(\bar\alpha_t)\vert}.
\end{equation}
\item The deformations are uniform on $\Omega_L\setminus\Gamma$ and equal to
$$
\veps_t=\corr{\sigma_t}\,\A^{-1}\s.
$$
\item There exists a smooth evolution with maximal damage $\bar\alpha_t$ monotonically increasing from $0$ to $1$ for $t\in(\varepsilon_e^\s,\varepsilon_f^\s)$
with 
\begin{equation}
   \varepsilon_f^\s=
   \begin{cases}
      \frac{\Gc}{\sics L }\frac{1}{c_\w|\k'(1)|}&\text{if}\quad|\k'(1)|>0, \\
     +\infty&\text{if}\quad|\k'(1)|=0,
   \end{cases}
   \label{eq:epsf}
\end{equation}
provided that  the cube size $L$ is sufficiently small compared to the elasto-cohesive length $\ell_{\mathrm{ch}}^\s:={\Gc A_0^\s}/{(\sigma_c^\s)^2}$, namely:
       \begin{equation}
          \frac{L}{\ell_{\mathrm{ch}}^\s}\leq
          \rho_c:=\inf_{\alpha\in(0,1)}\frac{1}{c_\w\vert\k'(\alpha)\vert}\frac{\mathrm{d}}{\mathrm{d}\alpha}\left(\frac{\sqrt{\w(\alpha)}}{{\vert\k'(\alpha)\vert}} \right),
          \label{eq:loc-snapback-condition1}
       \end{equation}
where $\rho_c$ is non-negative because of Hypothesis~\ref{hyp:kw}.
\end{itemize} 
If $c^*_{\w}=+\infty$, equations~\eqref{eq:delta-alpha-loc1}-\eqref{eq:loc-snapback-condition1} will be valid only asymptotically for $\ell/L\to 0$. Remark~\ref{rem:AT2-loc-sol} will discuss this latter case  with a specific example. \\

\noindent To construct the solution we proceed with the following steps:
\begin{itemize}
   \item \emph{Strain and stresses on $\Omega_L\setminus\Gamma$.}
Consider a time $t>\varepsilon_e^\s=\sics/A_0^\s$ such that $\bar\alpha_t<1$. Due to the homogeneity of stresses, the monotonicity of the function $\k$, and the maximality of the phase-field at $x=0$, we have
$$\sigma_t\le \k(\bar\alpha_t)\sics<\k(\alpha_t(\mathbf x))\sics\le\sics,\quad \forall \mathbf x\not\in\Gamma.$$
Consequently, nonlinear strain vanishes, and total deformation remains uniform everywhere except on $\Gamma$:
\begin{equation*}
   \p_t(\mathbf x)=\mathbf 0,\quad \veps_t(\mathbf x)=\sigma_t\A^{-1}\s, \quad\forall\mathbf x\not\in\Gamma.
\end{equation*}

\item \emph{Displacement jump on $\Gamma$.}
Let us decompose the loading condition~\eqref{eq:loading-condition} into the contributions coming from $\Omega_L\setminus\Gamma$ where the strains are constant, and $\Gamma$, where a displacement jump $\jump{\u_t}$ is possible:
\begin{equation}
t=\frac{1}{L^3}\int_{\Omega_L\setminus\Gamma}\s\cdot\veps_t\,\mathrm{d}V
+
\frac{1}{L^3}\int_\Gamma \s\n\cdot\jump{\u_t}\,\mathrm{d}S
=\frac{\sigma_t}{A_0^\s}
+
\frac{1}{L^3}\int_\Gamma \s\n\cdot\jump{\u_t}\,\mathrm{d}S
\leq
\k(\bar\alpha_t)\frac{\sigma_c^\s}{A_0^\s}+\frac{1}{L^3}\int_\Gamma \s\n\cdot\jump{\u_t}\,\mathrm{d}S,
\label{eq:loading-condition-jump}
\end{equation}
where we used that $\sigma_t \leq \k(\bar\alpha_t)\sics$. 
Hence, for $t>\varepsilon_e^\s=\sics/A_0^\s$, the jump term on $\Gamma$ must be strictly positive and monotonically increasing with $t$, because $\k(\bar\alpha_t)$ cannot increase with $t$. 
The evolution of the jump also requires the yielding condition to be satisfied as an equality: $\sigma_t=\k(\bar\alpha_t)\sics$.
Since strains are identical on both sides of $\Gamma$, the displacement jump must be a rigid-body motion\corr{, $\boldsymbol{d}_t+\boldsymbol\omega_t\wedge\mathbf x$. Since $\u_t=\mathbf E_t\x+\mathbf w_t$ with $\mathbf w_t$ periodic, and since $\veps(\mathbf w_t)$ is uniform on $\Omega_L\setminus\Gamma$, the periodicity of $\mathbf w_t$ in the two directions parallel to $\Gamma$ imposes the same rotation on both sides of $\Gamma$: $\boldsymbol\omega_t=\mathbf 0$, and the jump is a uniform translation}:
\begin{equation}
   \label{sautu}
   \jump{\u_t}(\mathbf x)=\boldsymbol{d}_t,\quad\forall\mathbf x\in\Gamma.
\end{equation}
Because $\veps_t$ has a singular part concentrated on $\Gamma$ and $\veps_t-\p_t$ must belong to $L^2(\Omega_L)$, the field $\p_t$ must also be a measure (equal to $\veps_t$) on $\Gamma$. Furthermore, as $\p_t$ must be collinear with $\boldsymbol{\nu^\s}=\mathbf d_0\odot\mathbf n/\Vert\mathbf d_0\odot\mathbf n\Vert$, the displacement jump must be collinear with $\mathbf d_0$:
\begin{equation}\label{sautud}
   \boldsymbol{d}_t={d_t}\mathbf{d}_0,\quad {d_t}\ge0.\end{equation}
Hence,  the loading condition~\eqref{eq:loading-condition-jump} gives
\begin{equation}
   \label{loaddt}
   t=\k(\bar\alpha_t)\,\frac{\sics}{A_0^\s}+\frac{{d_t}}{L}\s\mathbf n\cdot\mathbf{d}_0.
\end{equation}
\item \emph{Evolution of the phase-field on $\Omega_L\setminus\Gamma$.} 
Since the displacement jump is non-zero and monotonically increasing with $t$ almost everywhere on $\Gamma$, the stresses must satisfy the yielding condition:
$$\sigma_t=\k(\alpha_t(\mathbf x))\sics,\quad \forall \mathbf x\in\Gamma.$$
Therefore, the phase-field attains its maximal value uniformly on $\Gamma$: $\alpha_t(\mathbf x)=\bar\alpha_t$, $\forall\mathbf x\in\Gamma$.
In $\Omega_L\setminus\Gamma$, since $\mathbf p_t=\mathbf 0$, the first-order stability condition~\eqref{eq:stab-alpha} and the energy balance~\eqref{eq:eb-alpha} reduce to:
$$
\w'(\alpha_t)-2\ell^2\Delta\alpha_t\ge0,
\quad 
\big(\w'(\alpha_t)-2\ell^2\Delta\alpha_t\big)\dot\alpha_t=0, 
$$
\corr{with $\alpha_t$ $\Omega_L$-periodic.}
Given that $\alpha_t(\mathbf x)=\bar\alpha_t$ is uniform on $\Gamma$, we look for a phase-field depending only on the coordinate $x$ along $\mathbf n$ and symmetric with respect to $\Gamma$, leading to:
$$
\w'(\alpha_t(x))-2\ell^2\frac{\mathrm{d}^2\alpha_t}{\mathrm{d}x^2}(x)\ge 0,\quad \forall x>0\quad\text{with}
\quad 
\alpha_t(0)=\bar\alpha_t.
$$
The solution of this problem is classical in phase-field fracture models, see \emph{e.g.}~\cite{Pham-Marigo-EtAl-2011a,AleMarVid14}.
We seek a solution such that the phase-field is positive, not increasing with the distance from $x=0$, and increasing in time in the region $x\in(0,\bar \ell_t/2)$, where the inequality above must be verified as an equality. 
The solution must respect the first integral:
\begin{equation}
   \label{eq:first-integral}
\ell^2\left(\frac{\mathrm{d}\alpha_t(x)}{\mathrm{d}x}\right)^2 = \w(\alpha_t(x))-\w(\alpha_t^L),
\end{equation}
where $\alpha_t^{L}=\alpha_t(L/2)$. The value of the first integral is given by evaluating its integrand at $x=L/2$, where $\alpha_t'(L/2)=0$, because the solution is \corr{symmetric with respect to $\Gamma$ and periodic}.
 Consequently,  the phase-field is given by
\begin{equation*}
|x|=\ell\int_{\alpha_t(x)}^{\bar\alpha_t}\frac{\mathrm{d}\beta}{\sqrt{\w(\beta)-\w(\alpha_t^{L})}}\quad\text{for}\quad |x|\le \bar{\ell}_t/2=\ell\int_{\alpha_t^{L}}^{\bar\alpha_t}\frac{\mathrm{d}\beta}{\sqrt{\w(\beta)-\w(\alpha_t^{L})}},
\end{equation*}
and \corr{$\alpha_t(x)=\alpha_t^L$ for $\bar{\ell}_t/2\le|x|\le L/2$}. To ensure irreversibility, it suffices that $\bar\alpha_t$ is non-decreasing, since this guarantees that $\bar{\ell}_t$ and $\alpha_t(x)$ are non-decreasing for all $x$.
If $c^*_{\w}<\infty$, for  $\ell$ small enough  $\bar{\ell}_t\le L$ and $\alpha_t^{L}=0$. 
In this case, the first integral gives the following condition for the jump of the derivative of the phase-field at $x=0$:
\begin{equation}\label{jumpal'}
   \ell\left[\!\!\left[\frac{\mathrm{d}\alpha_t}{\mathrm{d}x}\right]\!\!\right](0) = -2\sqrt{\w(\bar\alpha_t)}.
\end{equation}
If $c^*_{\w}=+\infty$, $\alpha_t^L$ is the solution of 
$$\int_{\alpha_t^{L}}^{\bar\alpha_t}\frac{\mathrm{d}\beta}{\sqrt{\w(\beta)-\w(\alpha_t^{L})}}=\frac{L}{2\ell}$$
and the equation~\eqref{jumpal'} is verified only asymptotically for $\ell/L\to 0$, see the following Remark~\ref{rem:AT2-loc-sol} for the example $\w(\alpha)=\alpha^2$.
\item \emph{Evolution of the phase-field on $\Gamma$.} 
On $\Gamma$, as long as $\bar\alpha_t<1$, the first-order stability condition at $x=0$ formally reads
\begin{equation*}
\k'(\bar\alpha_t)\sics\, \s\cdot\p_t+\frac{\Gc}{4c_\w\ell}\w'(\bar\alpha_t)-\frac{\Gc}{2c_\w}\ell\frac{\mathrm{d}^2\alpha_t}{\mathrm{d}x^2}(0)\ge0,
\end{equation*}
but it must be interpreted in the sense of distributions since $\p_t$ and $\mathrm{d}^2\alpha_t/\mathrm{d}x^2$ are Dirac measures. Expressing it in terms of jumps using~\eqref{sautu} and~\eqref{sautud}, we obtain
\begin{equation}\label{STsaut}
\k'(\bar\alpha_t)\sics\, {d_t}\,\s\n\cdot\mathbf{d}_0-\frac{\Gc}{2c_\w}\ell\left[\!\!\left[\frac{\mathrm{d}\alpha_t}{\mathrm{d}x}\right]\!\!\right](0)\ge0.
\end{equation}
Equality in~\eqref{STsaut} holds whenever $\bar\alpha_t$ increases\corr{; with}~\eqref{jumpal'}\corr{, it reads}
\begin{equation}\label{STsaut2}
\k'(\bar\alpha_t)\sics\, {d_t}\, \s\n\cdot\mathbf{d}_0+\frac{\Gc}{c_\w}\sqrt{\w(\bar\alpha_t)}\ge0.
\end{equation}

We now show that equality holds in~\eqref{STsaut2} as long as $\bar\alpha_t<1$, using the assumption of continuous evolution and the strict monotonicity of $\k$, which implies that $\k'(\bar\alpha_t)<0$ for $\bar\alpha_t<1$, and the increase of $\sqrt{\w}/\vert{\k'}\vert$ guaranteed by Hypothesis~\ref{hyp:kw}, see equation~\eqref{eq:slopesfromkw}. 
At time $t_e={\varepsilon_e^\s}$,~\eqref{STsaut2} holds as an equality since $\bar\alpha_{{t}}=d_{{t}}=0$. Suppose that at some time $t_1$ the inequality is strict while $\bar\alpha_{t_1}<1$. Then there must exist $t_0\in[{\varepsilon_e^\s}, t_1)$ where equality holds but not in the interval $(t_0,t_1)$. During this time, $\bar\alpha_t$ would not evolve, implying by continuity that $\bar\alpha_{t_0}=\bar\alpha_{t_1}$. But then $d_{t_0}$ would be related to $\bar\alpha_{t_0}$ through equality in~\eqref{STsaut2}, and by the inequality at $t_1$, $d_{t_1}\le d_{t_0}$. This contradicts the requirement from~\eqref{loaddt} that ${d_t}$ must increase.
Eliminating ${d_t}$ between~\eqref{loaddt} and~\eqref{STsaut2}, we obtain the relation between $t$ and $\bar\alpha_t$:
\begin{equation}\label{eq:talpha}
t=\k(\bar\alpha_t)\frac{\sics}{A_0^\s}
+
\frac{\Gc}{c_\w\sics L}\frac{\sqrt{\w(\bar\alpha_t)}}{{\vert\k'(\bar\alpha_t)\vert}}\quad \text{for}\quad \bar\alpha_t<1.
\end{equation}
\item \emph{Irreversibility condition for the phase-field.}
To ensure that the implicit relation $t=f(\bar\alpha_t)$ above defines a monotonically increasing $\bar\alpha_t$, we examine the derivative:
$$
f'(\alpha)=
\vert\k'(\alpha)\vert\frac{\Gc}{\sics L}\left(
\frac{1}{c_\w\vert\k'(\alpha)\vert}\frac{\mathrm{d}}{\mathrm{d}\alpha}
\left(\frac{\sqrt{\w(\alpha)}}{{\vert\k'(\alpha)\vert}}
\right)
-\frac{L}{\ell_{\mathrm{ch}}^\s}
\right),
$$
which is non-negative if and only if~\eqref{eq:loc-snapback-condition1} is satisfied, where $\rho_c$ is non-negative thanks to Hypothesis~\ref{hyp:kw}, see~Equation~\eqref{eq:slopesfromkw}.
If $\rho_c>0$, the solution respects the irreversibility condition provided that the ratio $L/\ell_\mathrm{ch}^\s$ is sufficiently small. 
If $\rho_c=0$, the crack-like solution with a smooth evolution in time would respect the irreversibility condition only asymptotically for $L/\ell_\mathrm{ch}^\s\to 0$. The latter condition can be verified also for a finite cube size $L$ if $\ell_\mathrm{ch}^\s\to\infty$, as for the example of  an incompressible material and an isotropic loading direction $\s$.
\end{itemize}
\begin{remark}[Snap-back]
   For ${L}/{\ell_{\mathrm{ch}}^\s} > \rho_c$, the basic properties of the localized solution outlined in this section remain valid; however, it will not be possible to find a smooth evolution with a monotonically increasing maximum damage level $\bar{\alpha}_t$. 
   The response will in this case exhibit a snap-back instability. 
   We do not discuss this case explicitly in order to avoid the theoretical difficulties related to handling discontinuities in time. 
   For the $\mathsf{M1}$ and $\mathsf{M2}$ models given in~\eqref{eq:model-1}-\eqref{eq:model-LS}, the maximal size of the cube to avoid the snap-back is computed to get 
   \[
   \dfrac{L}{\ell_{\mathrm{ch}}^\s} < \rho_c=
\begin{cases}
   \mathsf{M1} \text{ model with $\zeta\in(0,1]$} :&
\dfrac{2 \zeta \sqrt{\zeta}( 1+\zeta)}{
    (1+\zeta)\sqrt{\zeta} - 
    {(1-\zeta)^2 
        \, \operatorname{csch}^{-1} \left( 
            \sqrt{ \dfrac{1}{\zeta} - 1 } 
        \right)
    } 
}\corr{\leq} 2,\\
\mathsf{M2} \text{ model with $\zeta\in[1,2)$}:&
\corr{\dfrac{2 (2 - \zeta)}{\zeta^2}\leq 2}. 
\end{cases}
\]
\end{remark}

The localized evolution is fully parametrized by $\bar\alpha_t$, the uniform value of the phase-field on the jump surface $\Gamma$. The corresponding total energy of the cube at time $t$ decomposes into three terms:
\begin{enumerate}
\item \textit{Elastic energy outside the discontinuity surface.} Since nonlinear strains vanish in $\Omega_L\setminus\Gamma$ and the stress and strains are uniform, the energy reduces to
\begin{equation*}
\mathcal E_t^{\mathrm{bulk}}(\bar\alpha_t) =  \frac{(\k(\bar\alpha_t)\sics)^2}{2A_0^\s} {L^3}
\end{equation*}
and it is proportional to the volume of the cube.

\item \textit{Elastic energy on the discontinuity surface.} As the deformations are localized, the energy is given by
\begin{equation}
   \label{eq:Esur}
   \mathcal E_t^{\mathrm{sur}}(\bar\alpha_t) = 
   L^2 \k(\bar\alpha_t)\, \sics\, {d_t} \,\mathbf s \mathbf n \cdot \mathbf{d}_0
  =\Gc\,L^2 \frac{\k(\bar\alpha_t)}{c_\w} \frac{\sqrt{\w(\bar\alpha_t)}}{\abs{\k'(\bar\alpha_t)}}.
\end{equation}

\item \textit{Energy dissipated in the phase-field.} This corresponds to the integral of the $D_\ell(\alpha_t,\nabla\alpha_t)$ term and, given the phase-field profile computed previously, can also be written as a function of $\bar\alpha_t$:

\begin{equation}\label{eq:Dis}
\mathcal E^\mathrm{diss}(\bar\alpha_t) = L^2 { \Gc}  f_\w(\bar\alpha_t),\qquad\text{with}\quad f_\w(\alpha):=\frac{\int_0^{\alpha} \sqrt{\w(\beta)} \, d\beta}{\int_0^{1} \sqrt{\w(\beta)} \, d\beta}.
\end{equation}
As with the surface elastic energy, it is proportional to the surface area and represents a fraction of the Griffith surface energy, determined by the function $f_\w(\bar\alpha_t)$.
\end{enumerate}

The surface energy is thus composed of an elastic contribution and a dissipation term. Both are  proportional to the surface area and depend neither on $\ell$ nor on the direction $\mathbf s$, but only on the maximum value of the phase-field. We can thus introduce the following function giving  the surface energy density in the localized solution:
\begin{equation}\label{eq:Phi}
\Phi(\bar\alpha):= \Gc\left(
   \frac{\k(\bar\alpha)}{c_\w}\frac{ \sqrt{\w(\bar\alpha)} }{\abs{\k'(\bar\alpha)}} +f_\w(\bar\alpha)
\right).
\end{equation}
Its derivative is given by
\begin{equation*}
   \label{eq:Phiprime}
\Phi'(\bar\alpha) = 
\Gc\frac{\k(\bar\alpha)}{c_\w}\frac{\mathrm{d}}{\mathrm{d}\bar\alpha}\left(\frac{ \sqrt{\w(\bar\alpha)} }{\abs{\k'(\bar\alpha)}}\right) \ge 0
\end{equation*}
so that $\Phi(\bar\alpha)$ increases from 0 to $\Gc$ as $\bar\alpha$ increases from 0 to 1, because of Hypothesis~\ref{hyp:kw}, see Equation~\eqref{eq:slopesfromkw}.\\

Figure~\ref{fig:modelproblem-solution} shows the evolution of the energy and the stress  with the loading parameter $t$ for both the homogeneous response and the response with localized deformation. During the linear elastic phase (i.e., while $t \le {\varepsilon_e^\s}$), both responses are identical, and energy increases quadratically with $t$. After ${\varepsilon_e^\s}$ and until time $\varepsilon_c^\s$ (which is larger for smaller $\ell$), the homogeneous response energy grows linearly. In contrast, the localized response bifurcates from ${\varepsilon_e^\s}$, with slower energy growth, though both curves remain tangent at ${\varepsilon_e^\s}$.
The energy remains below the Griffith surface energy $\Gc L^2$, reaching it at a finite time for the model used, which corresponds to the  model $\mathsf{M1}$~\eqref{eq:model-1} with $\zeta=1/2$. This highlights a scale effect, as these evolutions depend on the ratio $\Gc/\sics L$. 

\begin{remark}
   \label{rem:AT2-loc-sol}
   The evolution conditions~\eqref{eq:delta-alpha-loc1} and the energy expressions presented above are derived under the assumption that the constant \( c^*_{\w} \), as defined in~\eqref{eq:kw-conditions-1}, is finite. This assumption ensures that the width \( \bar{\ell}_t \) of the region where the phase-field evolves remains bounded. While this hypothesis simplifies the formulation, it is not essential. In cases where condition~\eqref{eq:kw-conditions-1} is not satisfied, the evolution equations~\eqref{eq:delta-alpha-loc1} and the associated energy expressions remain asymptotically valid in the limit \( \ell/L \to 0 \).
   To illustrate this, we consider the specific case of the models $\mathsf{M1}$ and $\mathsf{M2}$ with $\zeta = 1$, introduced in~\eqref{eq:model-LS-zeta-1}, for which $c^*_{\w} =+\infty$. The evolution conditions for $\alpha_t$, derived from the irreversibility, stability, and energy balance conditions, reduce to:
   \[
   \ell^2\, \alpha_t''(x) - \alpha_t(x)= 0, \quad \alpha_t(0) = \bar\alpha_t, \quad \alpha_t'\left({L}/{2}\right) = 0,
   \]
   which admits the solution:
   \[
   \alpha_t(x) = \bar\alpha_t  \frac{\cosh\left(\frac{L}{2\ell} - \frac{x}{\ell}\right)}{\cosh\left(\frac{L}{2\ell}\right)}.
   \]
   Accordingly, the first integral~\eqref{eq:first-integral} becomes:
   \begin{equation*}
   \ell^2 \left(\frac{d\alpha_t(x)}{dx}\right)^2 = \alpha_t^2(x)- \frac{\bar\alpha_t^2}{\cosh^2\left(\frac{L}{2\ell}\right)}.
   \end{equation*}
   As a consequence, the conditions~\eqref{eq:delta-alpha-loc1} are not satisfied exactly, but only in the limit $\ell/L \to 0$, with an exponentially vanishing remainder. Similarly, the dissipated energy of the localized solution evaluates to:
   \[
   \mathcal{E}^\mathrm{diss}(\bar\alpha_t) = L^2 \Gc \bar\alpha_t^2 \tanh\left(\frac{L}{2\ell}\right),
   \]
   which converges exponentially to the expression defined in~\eqref{eq:Dis} as $\ell/L \to 0$.      
\end{remark}

\subsection{The stability of homogeneous and localized responses}

\label{sec:model-problem-stab}
The two obtained solutions satisfy the irreversibility condition, the first-order stability condition, and the energy balance. It remains to verify whether they are indeed stable at every instant in the sense of~\eqref{eq:ST}, or even better, if they are globally stable (meaning that any other admissible state $(\mathbf v,\mathbf q,\beta)$ at instant $t$, satisfying $\beta\ge\alpha_t$, has greater energy). 
Although we are currently unable to rigorously prove that the solution exhibiting localized deformation is globally energy-minimizing, we shall establish a \corr{weaker result: the localized solution has minimal energy among the fields whose phase-field attains its maximum uniformly on $\Gamma$, a class that contains the homogeneous response. Global optimality remains open.}

To each $\bar\alpha\in[0,1]$, we associate the set $\mathcal D(\bar\alpha)$ of phase-fields defined by
\[
\mathcal D(\bar\alpha)=\left\{\beta\in \corr{H^1_{\rm per}(\Omega_L)}: 0\le\beta(\mathbf x)\le\bar\alpha\quad \forall \mathbf x\in\Omega_L,\quad \beta(\mathbf x)=\bar\alpha\quad\forall \mathbf x\in\Gamma\right\}.
\]
\corr{Here $H^1_{\rm per}(\Omega_L)$ is the space of $\Omega_L$-periodic $H^1$ functions.} In other words, the phase-fields attain their maximal value uniformly on $\Gamma$  (though they may also take the value $\bar\alpha$ elsewhere). The following stability result for the localized response holds:
\begin{proposition}[Optimality of the localized deformation response]
   Provided that \corr{$c^*_\w<+\infty$, see~\eqref{eq:kw-conditions-1}, that $2c^*_\w\ell\le L$,} and that \corr{$L/\ell_{\mathrm{ch}}^\s\le\rho_c$, condition~\eqref{eq:loc-snapback-condition1}},
   at every instant $t > {\varepsilon_e^\s}$, the response $(\u_t,\mathbf p_t,\alpha_t)$ constructed in Section~\ref{sec:model-problem-loc} has minimal energy among all admissible fields $(\mathbf v,\mathbf q,\beta)$ satisfying $\mathbf v\in \mathcal C_t$ and $\beta\in\corr{\mathcal D_\Gamma}:=\cup_{\bar\alpha\in[0,1]}\mathcal D(\bar\alpha)$.
   \end{proposition}
   
   \begin{proof}
   For given $\bar\alpha\in[0,1]$ and $t>{\varepsilon_e^\s}$, we consider the following minimization problem, whose minimum is denoted by $\bar{\mathcal E}_t(\bar\alpha)$:
   \[
   \bar{\mathcal E}_t(\bar\alpha)=\min_{\mathbf v\in\mathcal C_t,\mathbf q,\beta\in\mathcal D(\bar\alpha)}\mathcal E_\ell^t(\mathbf v,\mathbf q,\beta).
   \]
   Since $\k$ is decreasing and $\H_{\K_0}$ is \corr{nonnegative}, we have $\k(\beta)\H_{\K_0}(\mathbf q)\ge\k(\bar\alpha)\H_{\K_0}(\mathbf q)$ for every $\beta\in\mathcal D(\bar\alpha)$. This yields the lower bound:
   \begin{equation}\label{eq:minmin}
   \bar{\mathcal E}_t(\bar\alpha)\ge
   \min_{\mathbf v\in\mathcal C_t,\mathbf q}\int_{\Omega_L}\left(\frac12\A(\veps(\v)-\mathbf q)\cdot(\veps(\v)-\mathbf q)
   +
   \k(\bar\alpha)\H_{\K_0}(\mathbf q)\right)\mathrm{d}V
   +
   \min_{\beta\in\mathcal D(\bar\alpha)}\int_{\Omega_L} D_\ell(\beta,\nabla\beta)\mathrm{d}V.
   \end{equation}
   The minimization problem in $(\mathbf v,\mathbf q)$ in~\eqref{eq:minmin} is a classical Hencky perfect elasto-plastic minimization problem with yield set $\k(\bar\alpha)\K_0$. By convexity, this problem admits a unique solution in stress. For the case $\bar\alpha=1$, we have $\k(1)=0$, and the minimal energy trivially equals $0$. 
   
   Now consider the case $\bar\alpha<1$. \corr{In this auxiliary problem the damage enters only through the constant $\k(\bar\alpha)$, and its stress solution is unique. A state with uniform strains and uniform stress $\sigma_*\s\in\k(\bar\alpha)\K_0$ satisfies its optimality conditions. Hence} the stress solution is uniform and can be written as $\sig_*(\mathbf x)=\sigma_*\mathbf s$, with $\sigma_*\le\k(\bar\alpha)\sics$. We first show that $\sigma_*>0$. Otherwise, since $\mathbf 0\in\corr{\mathring{\K}_0}$, the nonlinear strain $\mathbf p_*$ would satisfy $\mathbf s\cdot\mathbf p_*\le0$ everywhere in $\Omega_L$, implying the total strain also satisfies $\mathbf s\cdot\veps_*\le0$, contradicting the loading conditions. Next, we show $\sigma_*=\k(\bar\alpha)\sics$. Indeed, if strict inequality were true, the nonlinear strains would vanish, leading to total strain $\veps_*=\sigma_*\A^{-1}\mathbf s$, which implies $t=\sigma_*/{A}_0^\s<{\varepsilon_e^\s}$, contradicting the hypothesis $t>{\varepsilon_e^\s}$.
   
   Thus, $\sigma_*=\k(\bar\alpha)\sics$, and the nonlinear strains satisfy the normality rule \corr{$\mathbf p_*\in\mathrm N_{\K_0}(\sics\s)$, so that $\H_{\K_0}(\mathbf p_*)=\sics\,\s\cdot\mathbf p_*$}. Using the loading condition, we find
   \[
   \corr{\int_{{\Omega_L}}\s\cdot\mathbf p_*\,\mathrm{d}V}=L^3\left(t-\k(\bar\alpha)\frac{\sics}{A_0^\s}\right),
   \]
   which yields the minimal energy
   \begin{equation*}
   \min_{\mathbf v\in\mathcal C_t,\mathbf q}\int_{\Omega_L}\left(\frac12\A(\veps(\v)-\mathbf q)\cdot(\veps(\v)-\mathbf q)+\k(\bar\alpha)\H_{\K_0}(\mathbf q)\right)\mathrm{d}V
   =L^3\sics\left(\k(\bar\alpha)\,t-\frac{\k(\bar\alpha)^2}{2}\frac{\sics}{A_0^\s}\right).
   \end{equation*}
   This expression remains valid for $\bar\alpha=1$ since $\k(1)=0$.
   
The minimization problem in $\beta$ in~\eqref{eq:minmin} decouples and depends only on $\bar\alpha$. Since the phase-field is prescribed to  $\bar\alpha$ on $\Gamma$, this becomes effectively one-dimensional, and the optimal field $\alpha_*$ depends only on the coordinate $x$. \corr{Since $c^*_\w<+\infty$ and $d(\bar\alpha)\le c^*_\w\ell\le L/2$}, the optimal solution is given explicitly by
   \[
   |x|=\ell\int_{\alpha_*(x)}^{\bar\alpha}\frac{\mathrm{d}\beta}{\sqrt{\w(\beta)}}\quad\text{when}\quad |x|\le d(\bar\alpha)=\ell\int_0^{\bar\alpha}\frac{\mathrm{d}\beta}{\sqrt{\w(\beta)}}\quad\text{and}\quad \alpha_*(x)=0\quad\text{when}\quad |x|\ge d(\bar\alpha),
   \]
   leading to the minimal energy:
   \begin{equation}\label{minbeta}
   \min_{\beta\in\mathcal D(\bar\alpha)}\int_{\Omega_L} D_\ell(\beta,\nabla\beta)\mathrm{d}V
   =
   L^2\frac{\Gc}{c_\w}\int_0^{\bar\alpha}\sqrt{\w(\alpha)}\,\mathrm{d}\alpha.
   \end{equation}
   Returning to the global minimization problem over all admissible fields $\mathbf v\in\mathcal C_t$ and $\beta\in\corr{\mathcal D_\Gamma}$, we set
   $\bar{\bar{\mathcal E}}_t=\min_{\bar\alpha\in[0,1]}\bar{\mathcal E}_t(\bar\alpha)$. Combining~\eqref{eq:minmin}--\eqref{minbeta}, we obtain
   \[
   \bar{\bar{\mathcal E}}_t\ge
   \min_{\bar\alpha\in[0,1]}g_t(\bar\alpha)\sics L^3,
   \quad\text{with}\quad 
   g_t(\bar\alpha)=t\,\k(\bar\alpha)-\frac{\k(\bar\alpha)^2}{2}\frac{\sics}{A_0^\s} +\frac{\Gc}{c_\w\sics L}\int_0^{\bar\alpha}\sqrt{\w(\alpha)}\,\mathrm{d}\alpha.
   \]
   Analyzing the differentiable function $g_t$, we find that 
   the unique minimum \corr{is attained at $\bar\alpha_t$ given by~\eqref{eq:talpha} for $t<\varepsilon_f^\s$, and at $\bar\alpha_t=1$ for $t\ge\varepsilon_f^\s$}, provided \corr{$L/\ell_{\rm ch}^\s\le\rho_c$, condition~\eqref{eq:loc-snapback-condition1}}. Rewriting in terms of $\bar\alpha_t$, we obtain the final lower bound:
   \begin{equation*}
   \bar{\bar{\mathcal E}}_t\ge L^3
   \frac{(\k(\bar\alpha_t)\sics)^2}{2A_0^\s} +L^2\frac{\Gc}{c_\w}\k(\bar\alpha_t)\frac{\sqrt{\w(\bar\alpha_t)}}{\abs{\k'(\bar\alpha_t)}}+L^2\frac{\Gc}{c_\w}\int_0^{\bar\alpha_t}\sqrt{\w(\alpha)}\,\mathrm{d}\alpha,
   \end{equation*}
   which matches precisely the energy of the localized deformation solution from Section~\ref{sec:model-problem-loc}, thereby concluding the proof.
   \end{proof}

   Let us now examine the stability of the homogeneous response. We consider an instant $t$ when the response is nonlinear elastic without damage, \emph{i.e.}, $t\in({\varepsilon_e^\s},\varepsilon_c^\s)$. At this instant, the total strains, nonlinear strains, and stresses associated with the homogeneous response are given by
 \begin{equation}\label{homt}
 \mathbf p_t=\lambda_t\nus,\quad 
 \sig_t=\sics\mathbf s,\quad
 \veps_t=\sics\A^{-1}\mathbf s+\lambda_t\nus, 
 \quad\text{with}\quad 
 \lambda_t=\frac{t-{\varepsilon_e^\s}}{\nu^\s}.
 \end{equation}
Since $t<\varepsilon_c^\s={\varepsilon_e^\s}+\dfrac{\w'(0)}{|\k'(0)|}\dfrac{\Gc}{4c_\w\sics\ell}$, we have
 \begin{equation}\label{ineg}
 \H_{\K_0}(\mathbf p_t)=\lambda_t\sics\nu^\s<\frac{\w'(0)}{|\k'(0)|}\frac{\Gc}{4c_\w\ell}.
 \end{equation}
 We introduce a small perturbation such that the displacements, nonlinear strains, and phase-field become $\mathbf v=\u_t+\u_*$, $\veps(\mathbf v)=\veps_t+\veps_*$ with $\veps_*=\veps(\mathbf u_*)$, $\mathbf q=\mathbf p_t+\mathbf q_*$, and $\beta=0+\alpha_*$, where starred quantities denote perturbations. To ensure $\mathbf v\in\mathcal C_t$, we must have \corr{$\u_*=\mathbf E_*\x+\mathbf w_*$ with $\mathbf E_*\in\mathbb{M}^3_s$ and $\mathbf w_*$ $\Omega_L$-periodic, and}
 \begin{equation}\label{admistar}
 \int_{\Omega_L}\mathbf s\cdot\veps_*\,\mathrm{d}V=0.
 \end{equation}
 Let us denote by $\mathcal E_*$ and $\mathcal E_t$ the energies of the perturbed and homogeneous solutions, respectively. Their difference can be expressed as:
 \begin{multline*}
   \mathcal E_* - \mathcal E_t
   =\int_{\Omega_L}\Biggl(
   \underbrace{\H_{\K_0}(\mathbf p_t+\mathbf q_*) - \sics\mathbf s\cdot(\mathbf p_t+\mathbf q_*)}_{\texttt{T}_1}+\underbrace{\frac{\Gc}{4c_\w\ell}\w(\alpha_*) - (1-\k(\alpha_*))\H_{\K_0}(\mathbf p_t+\mathbf q_*)}_{\texttt{T}_2}\Biggr)\,\mathrm{d}V\nonumber\\
   +\int_{\Omega_L}\Biggl(
   \underbrace{\frac12\A(\veps_*-\mathbf q_*)\cdot(\veps_*-\mathbf q_*)}_{\texttt{T}_3}+\underbrace{\frac{\Gc\ell}{4c_\w}\nabla\alpha_*\cdot\nabla\alpha_*}_{\texttt{T}_4}\Biggr)\,\mathrm{d}V,
   \end{multline*}
 where we have simplified the expression of term $\texttt{T}_1$ by using~\eqref{homt}--\eqref{admistar} and the identity $\H_{\K_0}(\mathbf p_t)=\sics\mathbf s\cdot\mathbf p_t$.
 
 To show that the homogeneous state at instant $t$ is unstable, we must exhibit a small perturbation such that the energy difference $\mathcal E_*-\mathcal E_t$ becomes negative. The terms $\texttt{T}_3$ and $\texttt{T}_4$ (which are second-order terms) are clearly nonnegative, and so is term $\texttt{T}_1$ by the definition of the support function. Moreover, this term is of first-order and only vanishes if $\mathbf p_t+\mathbf q_*$ lies in the cone of outward normals to $\K_0$ at $\sics\mathbf s$. \corr{Since $\mathbf p_t$ has no singular part, the singular part of $\mathbf q_*$ must then also belong there.}
 
 The only possible source for a negative contribution is thus term $\texttt{T}_2$, which is a mixed term containing both first- and second-order contributions. Furthermore, due to~\eqref{ineg} and the smallness of $\alpha_*$, this term is positive wherever $\mathbf q_*$ is locally small. To obtain a negative contribution, one must therefore ``concentrate'' $\mathbf q_*$ onto a surface. On this surface, $\mathbf p_t$ becomes negligible compared to $\mathbf q_*$, and $\Gc/\ell$ becomes negligible compared to $\H_{\K_0}(\mathbf q_*)$. The remaining dominant \corr{term of $\texttt{T}_2$ is $-(1-\k(\alpha_*))\H_{\K_0}(\mathbf q_*)$, of second order. Since $\texttt{T}_1$ is of first order, this term dominates only if $\texttt{T}_1$ is small compared to it on the surface, that is, only if the singular part of $\mathbf q_*$ lies close to} the cone of outward normals to $\K_0$ at $\sics\mathbf s$.
 
 However, wherever $\mathbf q_*$ concentrates, the strain perturbation $\veps_*$ must also concentrate so that term $\texttt{T}_3$ remains integrable. This leads to a displacement discontinuity (jump) across the surface\corr{, of normal $\n$, and the singular parts of $\veps_*$ and $\mathbf q_*$ both equal $(\jump{\u_*}\odot\n)\,\delta$, see~\eqref{eq:psingular}}. Moreover, since \corr{this singular part} must \corr{lie close to} the cone of outward normals at $\sics\mathbf s$, there must exist an outward normal direction compatible with such a displacement jump. Thus, we recover precisely the situation previously encountered in constructing the localized deformation response.
 
 In summary, this analysis, which would require rigorous justification, strongly suggests that one cannot expect to find a small perturbation capable of destabilizing the homogeneous response before the instant $\varepsilon_c^\s$, unless there exists a unit outward normal to $\K_0$ at $\sics \mathbf{s}$ that is compatible with a displacement jump. If no such normal exists, the response should remain homogeneous up to the instant $\varepsilon_c^\s$, at which the damage criterion is reached. This critical time diverges as $1/\ell$ when $\ell\to 0$. Hence, one would recover the classical situation of the Griffith model, in which it becomes impossible to initiate a crack under the considered loading conditions.

 We now examine the case where there exists a unit normal $\nuss$ to $\K_0$ at $\sics \s$ that is compatible with a displacement jump. This normal may differ from the normal $\nus$ \corr{of the homogeneous response~\eqref{homt} when the cone of outward normals at $\sics\s$ is not reduced to a half-line}. Let $\mathbf n_*$ and $\mathbf d_*$ be two unit vectors defining this normal:
\[
\nuss = \frac{\mathbf d_* \odot \mathbf n_*}{\Vert\mathbf d_* \odot \mathbf n_*\Vert},\quad \text{with}\quad \mathbf d_* \cdot \mathbf n_*\ge 0.
\] 
\corr{The homogeneous response does not depend on the orientation of the cell. We take a cell with a face normal to $\mathbf n_*$, so that a jump across a cross-section of normal $\mathbf n_*$ is compatible with periodicity.} We will construct a perturbation that effectively leads to a negative energy difference, thus destabilizing the homogeneous response. Let $\Gamma_*$ be the \corr{mid-section of the cell} with normal $\mathbf n_*$, dividing it into two subdomains ${\Omega_L}^\pm$. We define the displacement field $\u_*$ as follows:
\[
\u_*(\x) = \corr{-c_*\,\nus\x}
+\begin{cases}
\mathbf 0,&\text{if } \x\in{\Omega_L}^-,\\[6pt]
{d}_*\mathbf d_*,&\text{if } \x\in{\Omega_L}^+,
\end{cases}
\]
where ${d}_*>0$ \corr{and $c_*=d_*\,\s\mathbf n_*\cdot\mathbf d_*/(L\nu^\s)$}. The associated strains $\veps_*$ are equal to $\corr{-c_*\nus}$ within ${\Omega_L}^\pm$, and the displacement jump $\jump{\u_*}$ equals ${d}_*\mathbf d_*$ across $\Gamma_*$. Thus, the admissibility condition~\eqref{admistar} is satisfied. \corr{The field $\u_*$ is of the form $(\mathbf E_*+\mathbf W_*)\x+\mathbf w_*$, with $\mathbf w_*$ periodic, $\mathbf E_*=-c_*\nus+(d_*/L)\,\mathbf d_*\odot\mathbf n_*$, and $\mathbf W_*$ the skew-symmetric part of $(d_*/L)\,\mathbf d_*\otimes\mathbf n_*$. The infinitesimal rotation $\mathbf W_*\x$ changes neither the strain nor the jump, and can be removed.}

We set $\mathbf q_*=\veps(\u_*)$ (in the sense of distributions), implying that terms $\texttt{T}_1$ and $\texttt{T}_3$ vanish\corr{; $\texttt{T}_1$ vanishes because $\p_t+\mathbf q_*$ equals $(\lambda_t-c_*)\nus$ in ${\Omega_L}^\pm$ and is proportional to $\corr{\nuss}$ on $\Gamma_*$, hence lies in the normal cone to $\K_0$ at $\sics\s$ for $c_*\le\lambda_t$}. For the remaining two terms, we define the phase-field by:
\[
\alpha_*(\x)=h\left(h-\frac{\corr{|\x\cdot\mathbf n_*|}}{L}\right)^+,\quad \x\in{\Omega_L},
\]
where $a^+=\max\{0,a\}$ and $h>0$. Thus, $\alpha_*$ equals $h^2$ on $\Gamma_*$, and its support has width $2hL$ around $\Gamma_*$. It follows that the integrals $\int_{\Omega_L}\nabla\alpha_*\cdot\nabla\alpha_*\,\mathrm{d}V$ and $\int_{\Omega_L}\w(\alpha_*)\,\mathrm{d}V$ are of order $h^3$.

Finally, for the remaining portion of term $\texttt{T}_2$, we have, for the volumetric contribution:
\[
-\int_{{\Omega_L}\setminus\Gamma_*}(1-\k(\alpha_*))\H_{\K_0}(\mathbf p_t+\mathbf q_*)\,\mathrm{d}V
=\sics\left(\frac{{d}_*}{\corr{L}}\s\mathbf n_*\cdot\mathbf d_*-\lambda_t\s\cdot\nus\right)\int_{\Omega_L}(1-\k(\alpha_*))\,\mathrm{d}V,
\]
which is of order $h^3$, whereas the contribution on $\Gamma_*$ is:
\[
-{d}_*(1-\k(h^2))\,\corr{L^2}\sics\s\mathbf n_*\cdot\mathbf d_*.
\]
Therefore, \corr{if $\k'(0)<0$, so that $1-\k(h^2)\sim|\k'(0)|h^2$,} it suffices to choose ${d}_*=\sqrt{h}\corr{L}$ to ensure that $\mathcal E_*-\mathcal E_t<0$ for small $h$\corr{: the contribution on $\Gamma_*$, of order $h^{5/2}$, dominates the terms of order $h^3$, and $c_*\le\lambda_t$ since $c_*$ is of order $\sqrt h$}, thus confirming the instability of the uniform state.

\begin{remark}[\corr{Damaging phase}]
\corr{The same perturbation, with $\beta=\alpha_t+\alpha_*$, destabilizes the homogeneous response in its damaging phase $t\in(\varepsilon_c^\s,\varepsilon_u^\s)$, which starts at $\varepsilon_e^\s$ when $\w'(0)=0$. The damage criterion~\eqref{eq:damage-criterion-H} then holds as an equality, so the volume terms remain of order $h^3$ at most, while the contribution on $\Gamma_*$ is of order $h^{5/2}$, with $|\k'(\alpha_t)|>0$ in place of $|\k'(0)|$. The homogeneous response is therefore unstable for every $t\in(\varepsilon_e^\s,\varepsilon_u^\s)$.}
\end{remark}

\subsection{Summary of findings}
We conclude this Section by summarizing the main results for the solution of the phase-field model for the multiaxial model problem, where we fixed the loading direction $\s$ in the stress space and controlled the loading amplitude via the average strain $t$. 
The results critically depend on the properties of the strength domain $\K_0$ in the direction $\s$:
\begin{itemize}
   \item \emph{If the strength domain is unbounded in the direction $\s$}, the solution remains linear elastic and homogeneous for any loading; the stress amplitude is $\sigma_t = A_0^\s\,t$, with the directional elastic stiffness $A_0^\s$ defined in~\eqref{eq:A0sC0s}.

   \item \emph{If the strength domain is bounded in the direction $\s$}, we denote by $\corr{\sig_0} = \sigma_c^\s\s \in \partial\K_0$ the critical stress at which the loading direction intersects its boundary. Nonlinear elastic deformations start to develop after an average deformation $\varepsilon_e^\s:= \sigma_c^\s / A_0^\s$. Then, we must distinguish two cases, depending on whether the normal $\nus$ of the strength domain $\corr{\K_0}$ at $\corr{\sig_0}$ is compatible with a displacement jump, in the sense of Definition~\ref{def:jumpcompatibility}, or not:
   \begin{itemize}
      \item \emph{If the normal is not compatible with a displacement jump}, homogeneous nonlinear elastic deformations develop without damage at a constant stress amplitude $\sigma_c^\s$. 
      This homogeneous solution remains stable until a critical average deformation  $\varepsilon_c^\s$ defined in~\eqref{eq:epsc}.
      At the critical deformation $\varepsilon_c^\s$, uniform damage begins to develop. 
      The average deformation is accommodated in a homogeneous softening regime with a stress $\k(\alpha_t)\sigma_c^\s$, until reaching the zero-stress level at the ultimate deformation $\varepsilon_u^\s$ given in~\eqref{eq:epsu}. 
      For models with $\w'(0)=0$, see \emph{e.g.}~\eqref{eq:model-LS-zeta-1}, we have $\varepsilon_c^\s = \varepsilon_e^\s$, and damage starts to develop directly at the elastic limit. 
   
      \item \emph{If the normal is compatible with a displacement jump}, localized solutions are possible and can compete with the homogeneous solution from the elastic limit $\varepsilon_e^\s$. 
      The localized solution described in Section~\ref{sec:model-problem-loc} involves a displacement jump in the direction $\mathbf{d}_0$ across a planar surface with normal $\n$ such that $\nus = \n \odot \mathbf{d}_0/\Vert\n \odot \mathbf{d}_0\Vert$. 
      The nonlinear deformations localize on the jump surface, where the damage reaches its maximum value $\bar\alpha_t$. The rest of the domain deforms elastically. 
      The stress decreases from $\corr{\sics}$ to $0$, and $\bar\alpha_t$ increases from $0$ to $1$, while the average deformation increases from $\varepsilon_e^\s$ to $\varepsilon_f^\s$, the latter being defined in~\eqref{eq:epsf}. The localized solutions \corr{have minimal energy among the fields with uniform maximal damage on $\Gamma$, and} a lower energy than the homogeneous solutions. In this case, the homogeneous solutions are unstable after the elastic limit $\varepsilon_e^\s$.
   \end{itemize}
\end{itemize}

The critical deformations $\varepsilon_e^\s$ and $\varepsilon_f^\s$, at which the localized deformation begins to develop and reaches the zero-stress level, are independent of $\ell$, as are their associated energies, for $\ell$ sufficiently small. In contrast, the critical deformation $\varepsilon_c^\s$, marking the onset of the damage process in the homogeneous response, diverges as $\ell \to 0$ \corr{when $\w'(0)>0$, and equals $\varepsilon_e^\s$ when $\w'(0)=0$}. 
In the limit, the homogeneous solution remains stable under any loading in a direction that is not compatible with a displacement jump; see Section~\ref{sec:model-problem-stab}.

The localized solutions represent the phase-field regularization of cohesive cracks. 
The multiaxial behavior of the model strongly depends on the fine properties of the strength surface $\partial\K_0$, which determine whether solutions with displacement jumps are admissible for a given orientation of the stress tensor. 
In the next section, we will introduce a special representation of the stress and the traction vector on a jump surface, which simplifies the characterization of this property and clarifies the link between the strength domain and the properties of the cohesive law. 
This representation provides the basis for the discussion of the model's limiting behavior as $\ell \to 0$ in Section~\ref{sec:SharpLimit}.

\section{Strength criteria, Mohr's representation of the stress, and  intrinsic curves}
\label{sec:strength}
The mathematical and phenomenological properties of the energy of the phase-field model defined by the energy functional~\eqref{eq:phase-field-energy} strongly depend on the choice of the strength criterion $\K_0$. 
This section revisits classical isotropic strength criteria and introduces a complementary characterization of stress states that are compatible with displacement jumps.
To this end, we draw upon and extend tools from the theory of Limit Analysis~\cite{DruPra52,Sal83}, including the representation of stress tensors via Mohr's circles and the description of material strength through their envelopes, referred to as intrinsic domains.  

\subsection{Mohr's circles and their envelopes}
\label{sec:Mohr}
The material of this sub-section is classical within the field of strength of materials, but we are not aware of a self-contained reference that presents it in a form accessible to a mathematically oriented audience.
For this reason, we provide a formal presentation of the key notions of the fundamental properties that will be used throughout the remainder of the paper.

\subsubsection{Mohr's representation of the stress vector }
Let $\sig$ be a stress tensor, $\n\in\mathbb{S}^2$ be a unit vector, and $\sig\n\in\mathbb{R}^3$ be the associated stress vector.
We can write
\begin{equation}
    \label{decomp}
    \sig\n = \Sigma\, \n + \mathbf{T},
\end{equation}
with $\Sigma = \sig\n\cdot \n $ and $\mathbf{T} \in \n^\perp := \left\{ \mathbf{v} \in \mathbb{R}^3; \n \cdot \mathbf{v} = 0\right\}$.
Let $(\mathbf{i},\mathbf{j})$ be an orthonormal basis of $\n^\perp$, and $T_1,T_2$ be the components of $\T$ in this basis.
We have that  
\[
    \|\T\|^2 = T_1^2+T_2^2,\qquad    \|\sig\n\|^2 = \Sigma^2 + \|\T\|^2.
\]

\begin{remark}[Parallel transport of the tangential components of the stress vector]
The orthonormal basis $(\mathbf{i},\mathbf{j})$ is of course not unique, and therefore neither are $T_1$ and $T_2$.
Consider however a unit vector $\n$ on the sphere $\mathbb S^2$ and an orthogonal basis $(\mathbf{i},\mathbf{j})$ of its tangent plane.
Given any unit vector $\n^*$, we can use a parallel transport to build an orthonormal basis $(\mathbf{i}^*,\mathbf{j}^*)$ of its tangent space, which allows us to compare the components $\Sigma, T_1,T_2$ of the stress vector $\sig \n$ when $\sig$ or $\n$ varies.
Up to the definition of a single reference basis $(\mathbf{i},\mathbf{j})$, we can therefore represent $\sig\n $ in terms of the triplet $(\Sigma, T_1, T_2)$ and $\T$ in terms of $(T_1,T_2)$.
We refer to these characterizations as their Mohr's representation.
Note in particular that $(\Sigma, T_1, T_2)$ depends linearly on $\sig$. 
\end{remark}

Let $\{\sigma_1,\sigma_2,\sigma_3\}$ be the eigenvalues of $\sig$ and 
$\{\mathbf{e}_1,\mathbf{e}_2,\mathbf{e}_3\}$ the corresponding eigenvectors, which form an orthonormal basis of $\mathbb{R}^3$. 
Here and henceforth, we use roman numbering to denote eigenvalues and matching eigenvectors ordered so that 
($\sigI\geq\sigII\geq\sigIII$).

By~\eqref{decomp}, we have that
\[
    \left\|\sig \n -\left(\frac{\sigI+\sigIII}{2}\right)\n\right\|^2 = \left| (\sig \n) \cdot \n - \left(\frac{\sigI+\sigIII}{2}\right)\right|^2 + \|\T\|^2.
\]
Writing then 
\[
    \n=n_\mathrm{I}\,\mathbf{e}_\mathrm{I}+n_{\mathrm{II}}\,\mathbf{e}_\mathrm{II}+n_\mathrm{III}\,\mathbf{e}_\mathrm{III},
\]
we have that 
\[
    \sig \n = \sigI n_\mathrm{I}\,\mathbf{e}_\mathrm{I}+ \sigII n_{\mathrm{II}}\,\mathbf{e}_\mathrm{II}+\sigIII n_\mathrm{III}\,\mathbf{e}_\mathrm{III},
\]
from which we derive that 
\[
    \left\|\sig \n -\left(\frac{\sigI+\sigIII}{2}\right)\n\right\|^2 = \left(\frac{\sigI-\sigIII}{2}\right)^2 \left(n_\mathrm{I}^2+\corr{n_\mathrm{III}^2}\right)+\corr{\left(\sigII-\frac{\sigI+\sigIII}{2}\right)^2n_\mathrm{II}^2}. 
\]
Noticing that $\left|\sigII - \left(\frac{\sigI+\corr{\sigIII}}{2}\right) \right| \le \left|\frac{\sigI-\sigIII}{2}\right|$, we obtain that 
\begin{equation}
    \label{MajMohr}
    \left(\Sigma-\frac{\sigI+\sigIII}{2}\right)^2+T_1^2+T_2^2\le \left(\frac{\sigI-\sigIII}{2}\right)^2.
\end{equation}
This proves that in its Mohr representation, the stress vector $\sig\n$ belongs to the \corr{ball} of radius $\frac{\sigI-\sigIII}{2}$ centered at $\left(\frac{\sigI+\sigIII}{2},0,0\right)$. 

Consider a vector $\n$ in the plane $(\mathbf{e}_\mathrm{III}, \mathbf{e}_\mathrm{I})$ and let $\theta$ be the angle between $\n$ and $\corr{\mathbf{e}_\mathrm{I}}$ (\emph{i.e.} the polar angle of $\n$ in the basis $\corr{\mathbf{e}_\mathrm{I},\mathbf{e}_\mathrm{III}}$ is $-\theta$)
\[
    \n = \cos \theta \mathbf{e}_\mathrm{I} - \sin \theta \mathbf{e}_\mathrm{III},
\]
and let 
\[
    \t = \sin \theta \,\mathbf{e}_\mathrm{I} + \cos \theta\, \mathbf{e}_\mathrm{III}
\]
be a unit vector orthogonal to $\n$ in the same plane, see Figure~\ref{fig:F-Angle}-left.
Then
\[
    \sig\n = \sigI\cos \theta \mathbf{e}_\mathrm{I} -\sigIII\sin \theta \mathbf{e}_\mathrm{III}
\]
is also in the plane $(\mathbf{e}_\mathrm{III}, \mathbf{e}_\mathrm{I})$.
If for $\theta=0$ (\emph{i.e.} $\n = \mathbf{e}_\mathrm{I}$), we choose $\mathbf{i} = \mathbf{e}_\mathrm{III}$ and $\mathbf{j} = \mathbf{e}_\mathrm{II}$ and use the parallel transport of this basis for all $\theta$, Mohr's representation of $\sig\n$ is of the form $(\Sigma,T_1,0)$, \emph{i.e.} $T_2 = 0$.
The decomposition~\eqref{decomp} gives then
\begin{equation}
    \label{eq:GCMohr}
    \Sigma=\frac{\sigIII+\sigI}2+\frac{\sigI-\sigIII}2\cos 2\theta,\quad T_1=\frac{\sigI-\sigIII}2\sin 2\theta.
\end{equation}
Similar constructions with $\sig\n$ in the plane $(\mathbf{e}_\mathrm{I}, \mathbf{e}_\mathrm{II})$ and $(\mathbf{e}_\mathrm{II}, \mathbf{e}_\mathrm{III})$ give respectively the two inner Mohr's circles.
Note that these are not relevant in this work so that in all that follows, we will refer to the outer Mohr's circle as \emph{the} Mohr's circle.

\begin{figure}[h!] 
    \centering
    \includegraphics[width=2.in]{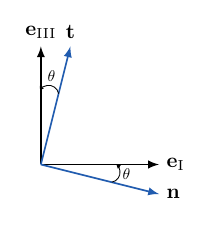} 
    \includegraphics[width=3.in]{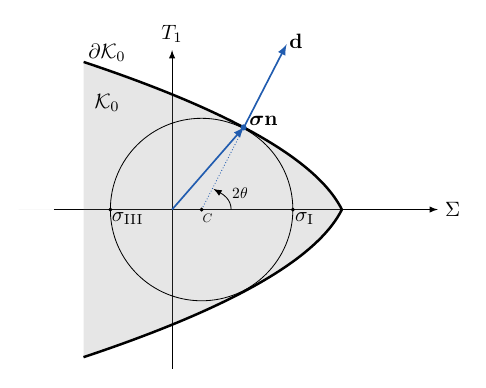} 
    \caption{Left: Orientation of the reference system $\boldsymbol{t}-\boldsymbol{n}$ and definition of the angle $\theta$. Right: Mohr's representation in the section $T_2=0$ of the space $(\Sigma,T_1,T_2)$  showing the section of the intrinsic domain $\mathcal K_0$ and   a Mohr's circle tangent to  its boundary. The stress vector $\sig \n$ corresponding to the point of contact between the circle and $\partial\mathcal K_0$ determines also the orientation  $\theta$ of the normal and the direction of the displacement jump $\mathbf{d}$, which must be in the normal cone of $\mathcal K_0$ at  $\sig \n$.}
    \label{fig:F-Angle}
 \end{figure}

\subsubsection{Intrinsic domain for an isotropic strength criterion}
\label{sec:intrinsic-curve}
Suppose that $\mathbf{0} \in \mathring{\K}_0$ and that $\K_0$ is bounded in at least one direction.
For any unit vector $\n \in \mathbb{S}^2$, we define
\begin{equation}
    \label{eq:defcalK0}
    \mathcal{K}_0(\n) := \left\{\sig\n;\ \sig \in \K_0 \right\}.
\end{equation}
As the image of a convex set by a linear map, $\mathcal{K}_0(\n)$ itself is convex.
Using~\eqref{MajMohr}, we also see that $\mathcal{K}_0(\n)$ is contained in the union of the balls of center $C=((\sigI+\sigIII)/2,0,0)$ and radius $R=(\sigI-\sigIII)/2$.

When $\K_0$ is \emph{isotropic}, \emph{i.e.}
\begin{equation}
    \label{eq:defIsotropy}
    \Q^T\sig\Q \in \K_0\  \forall \sig \in \K_0,\ \Q \in SO(3),
\end{equation}
where $SO(3)$ is the group of rotations in $\mathbb{R}^3$, it is possible to fully describe $\mathcal{K}_0(\n)$ in terms of the outer Mohr's circle.

We begin by establishing the following property:
\begin{proposition}
    \label{prop:Kappa0n}
  Let $\K_0$ be a closed convex isotropic subset of $\mathbb{M}^3_s$ containing $\mathbf{0}$.
  Then $\mathcal K_0(\n)$ in $\mathbb R^3$ is a body of revolution around the axis $(\Sigma,T_1,T_2)=(1,0,0)$ and is independent of $\n$. 
  We refer to it as the \emph{intrinsic domain} of admissible stress vectors.
\end{proposition}

\begin{proof}
    Let us first show that for a fixed $\n$, $\mathcal K_0(\n)$ is a body of revolution around the axis $(\Sigma,T_1,T_2)=(1,0,0)$. 
    Since $\sig = \mathbf 0 \in \K_0$, we have $(0,0,0) \in \mathcal K_0(\n)$. Let $\mathcal I(\n)$ be the intersection of $\mathcal K_0(\n)$ with the axis $(1,0,0)$. 
    Let us take $\Sigma \in \mathcal I(\n)$; then there exists $\sig \in \K_0$ such that $\sig\n \cdot \n = \Sigma$. 
    By isotropy, all tensors $\sig^* = \Q\sig\Q^T$ with $\Q \in \Orth3$ leaving $\n$ invariant, \emph{i.e.}, $\Q\n = \n$, also belong to $\K_0$, and hence $\sig^*\n$ belongs to $\mathcal K_0(\n)$. Moreover, $\sig^*\n \cdot \n = \Sigma$, and defining $\T^* := \sig^*\n - \Sigma\,\n = \Q\T$, where $\T = \sig\n - \Sigma\,\n$, we find $\norm{\T^*} = \norm{\T}$. By varying $\Q$, we generate all vectors $\T^*$ orthogonal to $\n$ having the same norm as $\T$. Consequently, for fixed $\Sigma$, all triplets $(\Sigma, T_1, T_2)$ with the same norm belong to $\mathcal K_0(\n)$, which is the desired property.

    Let us introduce
    $$\overline T_\n(\Sigma) = \sup_{(\Sigma, T_1, T_2) \in \mathcal K_0(\n)} \sqrt{T_1^2 + T_2^2}$$
    which, for a given admissible normal stress, gives the maximum possible norm of the tangential stress. To show that $\mathcal K_0(\n)$ is independent of $\n$, it suffices to show that both $\mathcal I(\n)$ and $\overline T_\n$ are independent of $\n$.
    \begin{itemize}
        \item \emph{ $\mathcal I(\n)$ is independent of $\n$.} Let $\n \in \mathbb S^2$, $\Sigma \in \mathcal I(\n)$, $\Q \in \Orth3$, and $\n^* = \Q\n$. Then there exists $\sig \in \K_0$ such that $\sig\n \cdot \n = \Sigma$, and hence $\sig^* = \Q\sig\Q^T \in \K_0$ and $\sig^*\n^* \cdot \n^* = \Sigma$, so that $\Sigma \in \mathcal I(\n^*)$. Therefore, $\mathcal I(\n) \subset \mathcal I(\n^*)$. Reversing the roles of $\n$ and $\n^*$ gives the reverse inclusion and thus equality. Since $\n$ and $\n^*$ can be chosen arbitrarily in $\mathbb S^2$, $\mathcal I(\n)$ is independent of $\n$. We now denote this interval simply as $\mathcal I$.

        \item \emph{ $\overline T_\n$ is independent of $\n$.} Let $\n \in \mathbb S^2$, $\Sigma \in \mathcal I$, $\Q \in \Orth3$, and $\n^* = \Q\n$. Then there exists $\sig \in \K_0$ such that $\sig\n \cdot \n = \Sigma$. Furthermore, the tangential stress $\T = \sig\n - \Sigma\,\n$ satisfies $\norm{\T} \le \overline T_\n(\Sigma)$. Then $\sig^* = \Q\sig\Q^T \in \K_0$, $\sig^*\n^* \cdot \n^* = \Sigma$, and $\T^* := \sig^*\n^* - \Sigma\,\n^* = \Q\T$. Since $\norm{\T} = \norm{\T^*} \le \overline T_{\n^*}(\Sigma)$, we have $\overline T_\n(\Sigma) \le \overline T_{\n^*}(\Sigma)$. Reversing the roles of $\n$ and $\n^*$ gives the reverse inequality, hence equality. Since $\Sigma$ is arbitrary in $\mathcal I$, the invariance property is established.
    \end{itemize}
\end{proof}

In all that follows, we will focus on isotropic strength domains, which can be described in the following way.

\begin{definition}[Intrinsic domain]
    \label{def:intrinsic}
Let $\K_0 \subset \mathbb{M}^3_s$ be an isotropic strength domain containing $\mathbf{0}$.
The \emph{intrinsic domain} associated with $\K_0$ is
\begin{equation}
    \label{eq:domint}
    \mathcal{K}_0 := \left\{(\Sigma,T_1,T_2) \in \mathbb{R}^3:\ \Sigma\in\mathcal I,\ \sqrt{T_1^2+T_2^2} \le \overline{T}(\Sigma)\right\}.
\end{equation}
\end{definition}

The following fundamental property, illustrated in~Figure~\ref{fig:F-Angle}-right, establishes the link between Mohr's circles and intrinsic domains:
\begin{proposition}
The meridian section $\mathcal {S}_0$ of $\mathcal{K}_0$ associated with the isotropic strength criterion $\K_0$ (see Figure~\ref{fig:F-Angle}) is the envelope of all Mohr's circles that are compatible with the strength criterion, \emph{i.e.} all circles in the Mohr's plane $(\Sigma,T_1,0)$ defined by the equation
$$\left(\Sigma-\frac{\sigI+\sigIII}{2}\right)^2+T_1^2=\frac{(\sigI-\sigIII)^2}{4},
$$
where $(\sigI,\sigIII)$ are the maximum and minimum principal stresses  of some stress tensor $\sig \in \K_0$.
\end{proposition}
\corr{We call $\mathcal S_0$ the \emph{intrinsic section} and $\partial\mathcal S_0\cap\mathcal S_0$ the \emph{intrinsic curve}.}

In the statement of the property, the word {\it circle} can be replaced by {\it disk}, since, by convexity, if the circle is contained in $\mathcal S_0$, then the entire disk is as well.

\begin{proof}
Let $\sigI \ge \sigIII$ be the extreme principal stresses of some $\sig \in \K_0$ with associated principal directions $\mathbf e_\mathrm{I}$ and $\mathbf e_\mathrm{III}$. 
Choosing $\n$ in the plane $(\mathbf e_\mathrm{III}, \mathbf e_\mathrm{I})$, the stress vector $\sig\n$ traces the corresponding Mohr's circle when $\n$ varies in that plane.

We first show that this circle lies in a meridian plane of $\mathcal K_0$, which depends on the chosen basis $(\mathbf i, \mathbf j)$ of the tangent plane to the unit sphere at $\mathbf e_\mathrm{I}$, and on the principal directions $(\mathbf e_\mathrm{I}, \mathbf e_\mathrm{III})$. When $\theta = 0$, the direction of the tangential stress is $\t = \mathbf e_\mathrm{III}$, whose components in the tangent plane to the sphere at $\mathbf e_\mathrm{I}$ are $(\mathbf i \cdot \mathbf e_\mathrm{III}, \mathbf j \cdot \mathbf e_\mathrm{III})$.
As $\theta$ varies, the vector $\n$ moves along the unit sphere, rotating around $\mathbf e_\mathrm{II}$ since it remains in the plane $(\mathbf e_\mathrm{III}, \mathbf e_\mathrm{I})$. 
Because of the parallel transport of  $(\mathbf i, \mathbf j)$, $\t$ remains orthogonal to both $\n$ and $\mathbf e_\mathrm{II}$, and its components in the transported basis remain unchanged. Thus, the  Mohr's circle lies in the meridional plane generated by $(1,0,0)$ and $(0, \mathbf i \cdot \mathbf e_\mathrm{III}, \mathbf j \cdot \mathbf e_\mathrm{III})$.

By applying an orthogonal transformation $\Q$ fixing $\mathbf e_\mathrm{I}$ and mapping $\mathbf e_\mathrm{III}$ to $\mathbf i$, we can align the circle with the plane $T_2 = 0$. Hence, $\mathcal S_0$ contains all Mohr's circles associated with $\K_0$. Therefore, the intrinsic section contains all major Mohr's circles compatible with $\K_0$. 

Finally, by~\eqref{MajMohr}, no point in $\mathcal S_0$ can lie outside the envelope of these circles. Therefore, $\mathcal S_0$ is precisely their envelope.
\end{proof}

In the following, the support function of the convex set $\mathcal K_0$ will play a key role. By definition, this is the function $\H_{\mathcal K_0}$ defined on $\mathbb{R}^3$ and given by  
\begin{equation}\label{PiK}
\H_{\mathcal K_0}(\Delta,D_1,D_2): = 
\sup_{(\Sigma,T_1,T_2) \in \mathcal K_0} \{\Sigma\Delta + T_1D_1 + T_2D_2\},
\end{equation}  
where the triplet $(\Delta,D_1,D_2)$ can be interpreted as the components of the vector $\mathbf d = \Delta\,\n + \mathbf D$ in the same basis used for the stress vector $\sig\n$, such that
$\sig\n \cdot \mathbf{d}=\Sigma\Delta + T_1D_1 + T_2D_2$.
We can also use the support function of $\K_0$ to express $\H_{\mathcal K_0}(\Delta,D_1,D_2)$. Indeed, since for any given $\n$, the set $\sig\n$ describes $\mathcal K_0$ when $\sig$ describes $\K_0$, we obtain  
\[
\H_{\mathcal K_0}(\Delta,D_1,D_2) = \H_{\K_0}(\n \odot \mathbf d).
\]  
In practice, this requires knowledge of $\n$, whereas the definition~\eqref{PiK} is intrinsic to the convex set $\mathcal K_0$ in $\mathbb{R}^3$.

\begin{remark}
    One could redefine the intrinsic domain by including the non-interpenetration condition. This would involve taking the envelope of Mohr's half-circles obtained by restricting $\Sigma$ with $\Sigma \geq (\sigIII + \sigI)/2$. Consequently, a different intrinsic curve would result. However, we will not do this here. We will consider applications where the non-interpenetration condition will be automatically satisfied. 
\end{remark}
\subsection{Characterization of stress states compatible with the jump condition in terms of the Mohr's envelope $\mathcal{K}_0$}
In the case of isotropic strength criteria, Mohr's representation and intrinsic sections allow for the following characterization of stress states admitting displacement jumps.

\begin{proposition}[Stress states compatible with displacement jumps]
    \label{prop:intrinsiccurvejump}
    Let $\K_0$ be convex and isotropic, $\sig \in \partial \K_0$ and $\boldsymbol{\nu} \in \mathrm{N}_{\K_0}(\sig)$.
    Suppose that there exist $\n \in \mathbb{S}^2$ and $\mathbf{d} \not = \mathbf{0}$ such that $\boldsymbol{\nu} = \n \odot \mathbf{d}$. Then, $\sig \n \in \partial \mathcal{K}_0$.
    Conversely, if there exists $\n \in \mathbb{S}^2$ such that $\sig \n \in \partial \mathcal{K}_0$, then for any  $\mathbf{d} \in \mathrm{N}_{\mathcal{K}_0}(\sig\n)$ with $\mathbf{d} \not =\mathbf{0}$, $\n \odot  \mathbf{d} \in \mathrm{N}_{\K_0}(\sig)$.
\end{proposition}
\begin{proof}
    Let $\sig \in\partial \K_0$ and $\boldsymbol{\nu} \in \mathrm{N}_{\K_0}(\sig)$. 
    Suppose that there exist a unit vector $\n$ and a non-zero vector $\mathbf{d}$ such that $\boldsymbol{\nu} =\n \odot \mathbf{d}$.
    Since $\sig \in \partial \K_0$, we have that for all $\sig^* \in \K_0$, 
    \[
        \sig \cdot \boldsymbol{\nu} \ge \sig^* \cdot \boldsymbol{\nu},
    \]
    and using the identity $\sig \cdot \n \odot \mathbf{d} = \sig \n \cdot \mathbf{d}$, we get that 
    \[
        \sig \n \cdot \mathbf{d} \ge \sig^* \n \cdot \mathbf{d},
    \]
    so that $\mathbf{d} \in \mathrm{N}_{\mathcal{K}_0}(\sig\n)$.
    Since $\mathbf{d} \not = \mathbf{0}$, by definition of the normal cone, we must then have that $\sig \n \in \partial \mathcal{K}_0$.

    Conversely, suppose now that $\sig \n \in \partial \mathcal{K}_0$ and let $\mathbf{d} \in \mathrm{N}_{\mathcal{K}_0}(\sig\n)$. 
    For any $\sig^* \in \K_0$, we then have
    \[
        \sig \n \cdot \mathbf{d} \ge \sig^* \n \cdot \mathbf{d},
    \]
    so that 
    \[
        \sig \cdot \n \odot \mathbf{d} \ge \sig^* \cdot \n \odot \mathbf{d},
    \]  
    and $\n \odot \mathbf{d} \in \mathrm{N}_{\K_0}(\sig)$.

\end{proof}
    
As a direct consequence of the preceding proposition, we obtain the following characterization of stress states that allow displacement discontinuities.

\begin{proposition}
For an isotropic strength domain $\K_0$, displacement jumps are only possible for stress tensors $\sig \in \partial \K_0$ such that there exists a normal $\n \in \mathbb{S}^2$ for which $\sig\n \in \partial \mathcal{K}_0$.
This excludes all stress tensors $\sig$ whose largest associated Mohr's circle \corr{lies in the interior of the intrinsic section $\mathcal S_0$}.  
For those whose largest Mohr's circle is tangent to the intrinsic curve, displacement jumps can only occur on the facet whose normal orientation $\n$ with respect to the extreme principal directions corresponds to the angle $\theta$ given by the point of tangency of Mohr's circle with the intrinsic curve (see Figure~\ref{fig:F-Angle}).  
Furthermore, for the non-penetration condition $\mathbf d\cdot\n \geq 0$ to be satisfied, the outward normal $\mathbf d \equiv (\Delta, D_1,D_2)$ to the intrinsic domain $\mathcal{K}_0$ must have its normal component $\Delta \geq 0$, which restricts the angle $\theta$ to the interval $[- \pi/4, \pi/4]$.
\end{proposition}

We illustrate the proposition above with the classical example of the von Mises and Tresca strength criteria.

\begin{example}
    The classical von Mises and Tresca strength criteria are defined by the following strength domains:
    \begin{align*}
            \text{von Mises}: 
            \mathbb{K}_0= &\left\{ \sig\in\mathbb M^3_s :\dfrac{\Vert\sig^D\Vert}{\sqrt{2}}=\sqrt{\dfrac{1}{6}}{\sqrt{(\sigI-\sigII)^2+(\sigII-\sigIII)^2+(\sigI-\sigIII)^2}}\leq\tau_c
            \right\},\\
            \text{Tresca}: 
            \mathbb{K}_0=&\left\{ \sig\in\mathbb M^3_s :\dfrac{\sigI-\sigIII}{2}\leq\tau_c
            \right\}.
    \end{align*}

    In the space of the (unordered) principal stresses $(\sigma_1,\sigma_2,\sigma_3)$, both criteria define cylinders with axis along the line $\sigma_1 = \sigma_2 = \sigma_3$. Figure~\ref{fig:VMvsTresca}-left shows their cross-section in the plane $\sigma_3 = 0$. The corresponding intrinsic domain $\mathcal{K}_0$ in Mohr's space $(\Sigma, T_1, T_2)$ is the envelope of the Mohr's circles associated with all admissible stress states. The Mohr's circle corresponding to a generic state with ordered principal stresses $(\sigI, \sigII, \sigIII)$ has its center on the axis $T_1 = T_2 = 0$ with $\Sigma = (\sigI + \sigIII)/2$ and its radius is $R = (\sigI - \sigIII)/2$. The Tresca criterion, which does not depend on the intermediate stress $\sigII$, directly gives the maximum radius $R = \tau_c$, regardless of the center's position. For the von Mises criterion, given the maximum and minimum principal stresses $\sigI$ and $\sigIII$, the largest Mohr's circle radius is obtained when $\sigII = (\sigI + \sigIII)/2$, giving the same maximum radius $R = (\sigI - \sigIII)/2 = \tau_c$. Therefore, the intrinsic domain $\mathcal{K}_0$ for both criteria is the same and corresponds to the cylinder $\sqrt{T_1^2 + T_2^2} \leq \tau_c$, whose cross-section for $T_2 = 0$ is the gray region bounded by the thick black lines in Figure~\ref{fig:VMvsTresca}-right.

For uniaxial traction, \emph{e.g.} $\sigma_2 = \sigma_3 = 0$, the stress state with maximum principal stress $\sigma_1$ corresponds to $\sqrt{3}\tau_c\, \mathbf{e}_1 \otimes \mathbf{e}_1$ for the von Mises criterion, and to $2\tau_c\, \mathbf{e}_1 \otimes \mathbf{e}_1$ for the Tresca criterion; see the green and blue dots in Figure~\ref{fig:VMvsTresca}-left, respectively. The green and blue circles in Figure~\ref{fig:VMvsTresca}-right represent the corresponding Mohr's circles in the cross-section $T_2 = 0$ of the Mohr's space $(\Sigma,T_1,T_2)$. 
The blue circle, associated with the Tresca criterion, is tangent to the intrinsic domain. Using the notation of Figure~\ref{fig:F-Angle}, this implies that, under uniaxial traction, the Tresca criterion admits displacement jumps along planes with normal direction $\theta = \pm\pi/4$. In contrast, the green circle corresponding to the von Mises criterion has radius $\sqrt{3}\,\tau_c/2$ and lies strictly inside $\mathcal{K}_0$. This indicates that the von Mises criterion does not admit displacement jumps under uniaxial traction. 

The magenta cross in Figure~\ref{fig:VMvsTresca}-left and the dashed magenta circle in Figure~\ref{fig:VMvsTresca}-right represent the stress state of maximum admissible intensity under pure shear, which is identical for both the von Mises and Tresca criteria. In this case, the Mohr's circle is tangent to the intrinsic domain, and both criteria admit displacement jumps along directions $\theta = \pm\pi/4$.

\end{example} 

\begin{figure}[H]
    \begin{center}
        \begin{center}
            \begin{minipage}{0.4\textwidth}
                \centering
                \includegraphics[width=\textwidth]{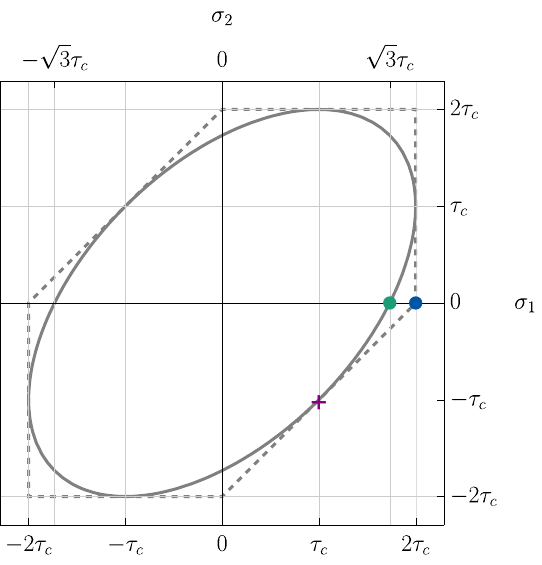}
            \end{minipage}
            \hfill
            \begin{minipage}{0.45\textwidth}
                \centering
                \includegraphics[width=\textwidth]{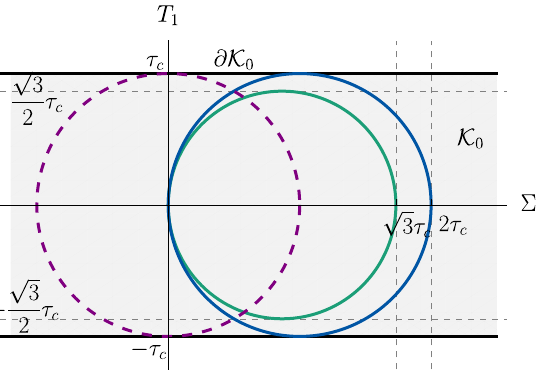}
            \end{minipage}
        \end{center}
        \caption{The von Mises and Tresca criteria define different strength domains (left, shown in the cross-section $\sigma_3 = 0$), but share the same intrinsic domain (right, shown in the cross-section $T_2 = 0$ of the Mohr's space). Green and blue dots (left) mark uniaxial traction states for von Mises and Tresca; the corresponding Mohr's circles are shown in the right panel. The magenta cross (left) represents pure shear; its Mohr's circle is the dashed magenta circle (right).}

    \label{fig:VMvsTresca}
    \end{center}
\end{figure}

\subsection{Examples of isotropic strength criteria and corresponding intrinsic curves}
\label{sec:strenght-examples}
We introduce two main classes of isotropic strength criteria.  
The first, which we name the \emph{Mohr-Coulomb family}, defines the strength directly in the Mohr's plane by prescribing the intrinsic domain \( \mathcal{K}_0 \) via a function \( \overline{T}(\Sigma) \).  
The corresponding admissible stress set \( \K_0 \) in the full stress space is then derived from this construction.
The second class, which we denote the \emph{Drucker-Prager family}, proceeds in the opposite direction: it defines the admissible stress set \( \K_0 \) directly in the stress space, and the intrinsic domain in the Mohr's plane is computed as a consequence. 

We will show that in light of Proposition~\ref{prop:intrinsiccurvejump}, the two classes exhibit fundamentally different behavior in our context.  
For the Mohr-Coulomb family, every stress state \( \sig \in \partial \K_0 \) admits a normal compatible with displacement jumps, whereas this property does not generally hold for the Drucker-Prager family.

\subsubsection{Mohr-Coulomb family}

The Mohr-Coulomb criterion assumes a linear relation between the maximum allowable shear stress $\overline{T}(\Sigma)$ and the normal stress $\Sigma$ in the form 
\begin{equation}
\Vert \mathbf{T}\Vert\leq \overline{T}(\Sigma)=\tau_c-\Sigma\,\tan{\varphi}, 
\label{eq:MC-intrinsic}
\end{equation}
where $\tau_c$ is the \emph{cohesion} and $\varphi$ is the \emph{friction angle}. 
The function $\overline{T}(\Sigma)$ above defines an intrinsic curve, and hence an intrinsic domain $\mathcal{K}_0$ through~\eqref{eq:defcalK0}-\eqref{eq:domint}:
\begin{equation*}
    \mathcal K_0=\left\{(\Sigma,T_1,T_2)\in \mathbb R^3: \sqrt{T_1^2+T_2^2}\le  \tau_c-\Sigma\,\tan{\varphi}\right\}.
\end{equation*}
The corresponding maximum allowable normal stress is $\sigma_c=\tau_c\cot{\varphi}$ and the interval $\mathcal I$ is $(-\infty,\sigma_c]$. 
Admissible stresses are all and only the stress tensors for which the corresponding Mohr's circle lies inside, or is tangent to, the intrinsic section $\mathcal{S}_0$, \emph{i.e.}~the meridian cross-section of $\mathcal K_0$. Imposing this condition and using the representation of the Mohr's circle in terms of the ordered eigenvalues $\{\sigI,\sigII,\sigIII\}$ gives the strength domain
\begin{equation*}
    \K_0=
    \left\{
    \sig\in\mathbb M^3_s: 
    (\sigI-\sigIII)+(\sigI+\sigIII)\sin\varphi-2\tau_c\cos\varphi\leq 0
    \right\}.
\end{equation*}
Since the Mohr's circle is defined only by the maximum and minimum eigenvalues of $\sig$, only the values of $\sigI$ and $\sigIII$ matter for stress criteria defined by the intrinsic curve. The criterion is isotropic, because the direction of the eigenvectors is not taken into account.
The meridian cross-section of $\mathcal{K}_0$ in the Mohr's plane and the domain of admissible stresses $\K_0$ in the stress space are represented in Figure~\ref{fig:MC-domains}.
\begin{figure}[H]
    \begin{center}
        \includegraphics[width=0.25\textwidth]{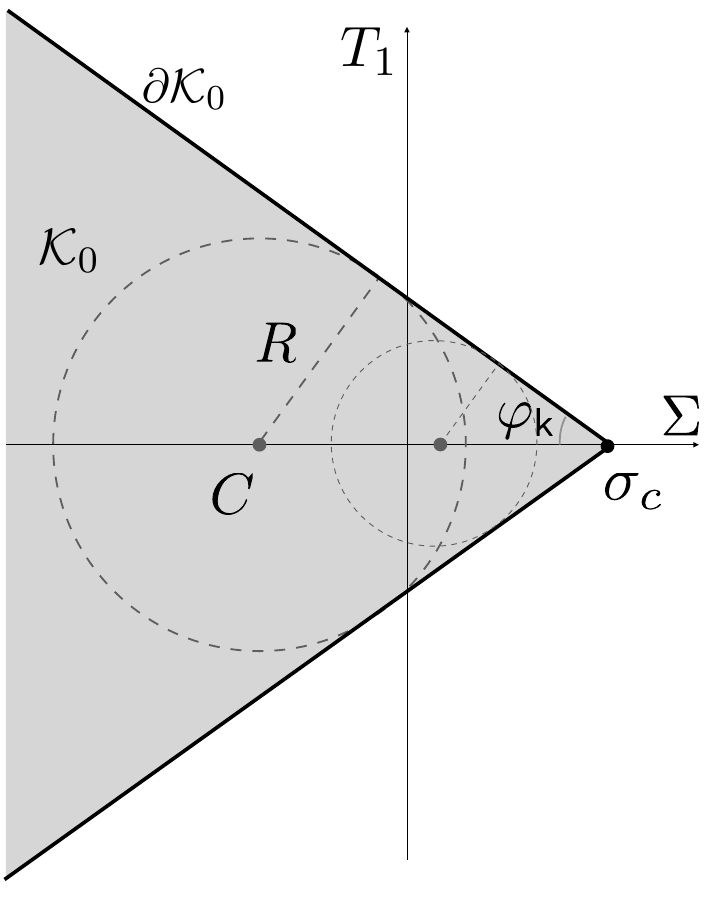}
        \includegraphics[width=0.3\textwidth]{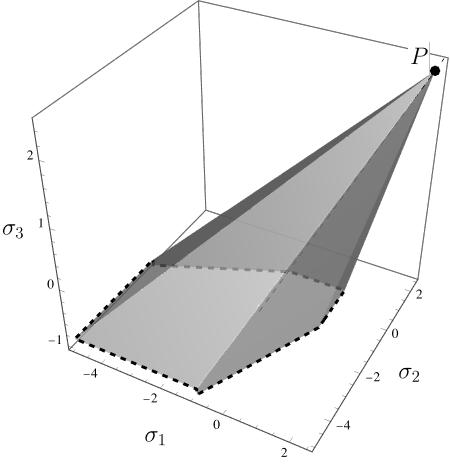}
        \quad
        \includegraphics[width=0.3\textwidth]{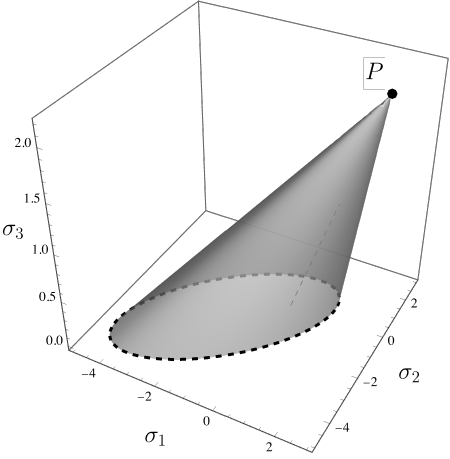}
        \caption{Strength domains for the Mohr-Coulomb and Drucker-Prager criteria and corresponding intrinsic section in Mohr's representation. The two criteria have the same intrinsic domain $\mathcal{K}_0$ (left)  but different strength domains $\K_0$ (center for Mohr-Coulomb and right for Drucker-Prager). 
        The strength domains are represented in the space of the (unsorted)  eigenvalues of $\sig$. The figures are for $\varphi=\pi/8$ and $\tau_c=1$.
    }
    \label{fig:MC-domains}
    \end{center}
\end{figure}

The  support functions of the strength domain $\K_0$ and the intrinsic domain $\mathcal{K}_0$ are given by~\cite{Sal83}:
\begin{eqnarray}
    \notag
\H_{\K_0}(\p)&=&
\begin{cases}
    \tau_c\,\cot{\varphi}\,\mathrm{tr}(\p)
    & \text{ if } \mathrm{tr}(\p)\geq \left(\vert{p_1}\vert+\vert{p_2}\vert+\vert{p_3}\vert\right)\sin{\varphi},\\
    +\infty&\text{ otherwise,}
\end{cases}\\
    \mathrm{H}_{\mathcal{K}_0}(\Delta,D_1,D_2) &=& 
    \begin{cases}
        \tau_c\,\cot{\varphi}\,\Delta
        & \text{ if } \Delta\geq \sqrt{\Delta^2+D_1^2+D_2^2}\sin{\varphi},\\
        +\infty&\text{ otherwise}.
    \end{cases}
\label{eq:MC-K0-supportfunction}
\end{eqnarray}

\begin{remark}
The Mohr-Coulomb criterion is usually used to model soil and rock behavior in conjunction with limit analysis theories to bound the solution for rigid-perfectly-plastic problems with test functions including displacement jumps~\cite{Sal83,Salencon1993CISM,CHEN75}. 
\end{remark}

\begin{remark}
In this Section, we present expressions for the support functions of several classical strength domains $\mathbb{K}_0$, without detailing their generally non-trivial computation. These expressions are available in standard texts on Limit Analysis~\cite{Sal83,Salencon1993CISM,CHEN75}. The derivation relies on the following key result: for a given symmetric matrix $\p$ with ordered eigenvalues $(\pI,\pII,\pIII)$, the following equality holds
\begin{equation}
\max_{\sig\in\mathbb M(\sigI,\sigII,\sigIII)} \sig\cdot\p = \sigI\pI + \sigII\pII + \sigIII\pIII,
\label{eq:suptensors}
\end{equation}
where $\mathbb {M}(\sigI,\sigII,\sigIII)$ denotes the set of all symmetric matrices whose ordered eigenvalues are $(\sigI,\sigII,\sigIII)$.
Given this, the determination of the support function reduces to the tedious solution of relatively simple convex optimization problems in three dimensions, see \emph{e.g.}~\cite[Section 3.6]{CHEN75}. The proof of the above identity, a subtle algebraic result taken for granted in the cited texts, was communicated to us by Patrick Ballard~\cite{Ballard25}, and relies on  Cauchy's interlace theorem, see also~\cite[Theorem~4.3.53, p.~255]{HornJohnson}.  Given the support function for the strength domain in  the stress space $\mathbb{K}_0$, the support function for $\mathcal{K}_0$ is  found by replacing $\p$ with $\mathbf{d}\odot\mathbf{n}$.
\end{remark}

Notable limit cases or variants of the Mohr-Coulomb criterion include the following:
\begin{itemize}
    \item \emph{Tresca criterion (maximal shear)}. It bounds the tangential component of the stress vector only, by imposing $\| \T\| \leq \overline{T}(\Sigma)=\tau_c$. It gives:
        \begin{align*}
        \K_0 & = \left\{\sig\in\mathbb M^3_s: \sigI-\sigIII\leq 2\tau_c \right\},\\
        \mathcal {K}_0 & =\left\{(\Sigma,T_1,T_2)\in\mathbb R^3: \sqrt{T_1^2+T_2^2}\le\tau_c\right\},\\
        \H_{\K_0}(\p) & =
            \begin{cases}
                \tau_c\left(\vert{p_1}\vert+\vert{p_2}\vert+\vert{p_3}\vert\right)  & \text{ if } \mathrm{tr}(\p)= 0\\
                +\infty  & \text{ otherwise,}
            \end{cases}\\
            \H_{\mathcal{K}_0}(\Delta,D_1,D_2) &=
            \begin{cases}
                \tau_c\,\sqrt{\Delta^2+D_1^2+D_2^2} & \text{ if } \Delta= 0\\
                +\infty & \text{ otherwise.}
            \end{cases}
        \end{align*}

\item \emph{Rankine criterion (maximal tension)}. It bounds the maximum normal component of the stress vector by imposing  $\Sigma\leq\sigma_c$. It gives:
    \begin{align*}
        \corr{\K_0} & = \left\{ \sig\in\mathbb M^3_s: \sigI\leq \sigma_c \right\},\\
        \mathcal K_0 & =\left\{(\Sigma,T_1,T_2)\in\mathbb R^3: \Sigma<\sigma_c\right\}\cup\{\corr{(\sigma_c,0,0)}\},\\
        \H_{\K_0}(\p) & =
            \begin{cases}  
                \sigma_c\,\mathrm{tr}(\p) & \text{ if } \corr{\p\succeq\mathbf 0},\\
                +\infty & \text{ otherwise,}
            \end{cases}\\
        \H_{\mathcal{K}_0}(\Delta,D_1,D_2) & =
            \begin{cases}
                \corr{\sigma_c\,\Delta} & \text{ if } \corr{D_1=D_2=0,\ }\Delta\geq 0,\\
                +\infty & \text{ otherwise.}
            \end{cases}
    \end{align*}

\corr{Only opening jumps $\mathbf d=\Delta\n$, $\Delta\ge0$, have a finite cost. These expressions are the limit of~\eqref{eq:MC-K0-supportfunction} as $\varphi\to\pi/2$ at fixed apex $\tau_c\cot\varphi=\sigma_c$.}

\item \emph{Tresca with a tension cut-off}. It is a combination of the two criteria above, obtained when imposing $\Sigma\leq\sigma_c$ and $\norm{\T}\leq\tau_c$:
    \begin{align*}
        \K_0 & = \left\{ \sig\in\mathbb M^3_s: \sup\{\sigI-\sigIII-2\tau_c,\sigI-\sigma_c\}\leq 0 \right\},\\
        \mathcal K_0 & = \left\{ (\Sigma,T_1,T_2)\in\mathbb R^3: \Sigma\le \sigma_c, {(\Sigma-\sigma_c+\tau_c)^+}^2+T_1^2+T_2^2\le\tau_c^2 \right\},\\
        \H_{\K_0}(\p) & =
            \begin{cases}
                \tau_c\left(\vert{p_1}\vert+\vert{p_2}\vert+\vert{p_3}\vert-\mathrm{tr}(\p)\right)+\sigma_c\mathrm{tr}(\p) & \text{ if } \mathrm{tr}(\p)\geq 0,\\
                +\infty&\text{otherwise,}
            \end{cases}\\
    \H_{\mathcal{K}_0}(\Delta,D_1,D_2) & =
        \begin{cases}
            \tau_c\,(\sqrt{\Delta^2+D_1^2+D_2^2}-\Delta)+\sigma_c\,\Delta & \text{ if } \Delta\geq \corr{0}\\
            +\infty & \text{ otherwise,}
    \end{cases}
    \end{align*}
    where $(\cdot)^+=\sup(\cdot,0)$.
\end{itemize}

\begin{remark}
All the criteria above define the admissible stress only in terms of the maximum and minimum eigenvalues $\sigI$ and $\sigIII$.
This implies that for all the stress states $\sig\in\partial\mathbb{K}_0$ there exists a normal $\n$ such that  $\sig\n\in\partial\mathcal{K}_0$. Hence, because of Proposition~\ref{prop:intrinsiccurvejump}, all the points of $\partial\K_0$ have a normal compatible with a displacement jump.
\end{remark}
\subsubsection{Drucker-Prager family}
\label{sec:Mohr-DP}
The Drucker-Prager criterion~\cite{DruPra52} is defined in terms of the norm of the spherical and the deviatoric parts of the stress tensor $\sig$. The strength domain $\K_0$ is defined as
\begin{equation}
    \label{eq:DP}
    \K_0=
    \left\{
        \sig\in\mathbb M^3_s: \frac{\norm{\sig^D}}{\sqrt{2}}\le 3 k\left(\sigma_c-\frac{\mathrm{tr}(\sig)}{3}\right)
    \right\}\mbox{\quad with\quad} 0< k<\frac{1}{2\sqrt3}, \quad\sigma_c>0.
\end{equation}
$\K_0$ is a cone pointed at $\sigma_c\, \mathbf I$, where \corr{$\sigma_c$ is the hydrostatic tensile strength and $\tau_c=3k\,\sigma_c$ the strength in pure shear}, see Figure~\ref{fig:MC-domains}.
The support function is computed by using the decomposition into the spherical and deviatoric parts, which gives
\begin{equation}
    \label{piDP}
    \H_{\K_0}(\mathbf p)=
    \begin{cases}
        \sigma_c\text{tr}(\mathbf p)&\mbox{if } \text{tr}(\mathbf p)\ge3k \sqrt2\norm{\mathbf p^D},\\
        +\infty&\mbox{ otherwise.}
    \end{cases}
\end{equation}
Being formulated in terms of the invariants of the stress tensor, the criterion is isotropic.  
The intrinsic domain can be computed using Mohr's circle representation.  
Given the maximum and minimum principal stresses, $\sigI$ and $\sigIII$, we express the intermediate principal stress as  $\sigII = \chi \sigIII + (1 - \chi) \sigI$, with
$\chi \in [0,1]$.
 Plugging this expression into~\eqref{eq:DP} written in terms of the eigenvalues, one finds that  $\sig\in\partial \K_0$ when 
$f_k(\chi)R=3\sqrt{3}k(\sigma_c-C)$ 
with $f_k(\chi)=2\sqrt{1-\chi+\chi^2}+(1-2\chi)\sqrt{3}\, k,$
where $C=(\sigI+\sigIII)/2$ is the center of the Mohr's circle and $R=(\sigI-\sigIII)/2$ its radius.
Given the center, the maximum radius is found when $\chi$, \emph{i.e.}~$\sigII$, minimizes $f_k(\chi)$. This gives $\chi_k=(1+\sin\varphi_k)/2$, with $\sin\varphi_k={3k}/{\sqrt{1-3k^2}}$. The maximum radius as a function of the position of the center is 
$
R=\sin\varphi_k(\sigma_c-C)$ with $C\le\sigma_c$.
The straight lines $T=\pm\tan\varphi_k(\sigma_c-\Sigma)$ are tangent to all these circles. They represent the intrinsic curve and the intrinsic domain is
 \begin{equation*}
    \mathcal K_0=\left\{(\Sigma,T_1,T_2)\in\mathbb R^3:\Sigma\le\sigma_c, \sqrt{T_1^2+T_2^2}\le(\sigma_c-\Sigma)\tan\varphi_k\right\}\mbox{\quad with\quad} \varphi_k=\arcsin\frac{3 k}{\sqrt{1-3k^2}}.
\end{equation*}
Recalling that $\sigma_c = \tau_c\,\cot{\varphi}$, the intrinsic domain associated with the Drucker-Prager criterion coincides with that of the Mohr-Coulomb criterion when the parameter $k$ is chosen as
$
k = {\sin \varphi}/{\sqrt{3(3 + \sin^2 \varphi)}}.
$
As a result, \corr{the support function of its intrinsic domain} coincides with that of the Mohr-Coulomb model given in~\eqref{eq:MC-K0-supportfunction}, up to a suitable renaming of variables.

The limit cases of the Drucker-Prager criterion analogous to the Tresca and Rankine criteria introduced for the Mohr-Coulomb criterion are the following:
\begin{itemize}
    \item \emph{von Mises criterion}. It bounds the deviatoric part of the stress tensor only. It gives:
    \begin{align*}
        \K_0 & = \left\{ \sig\in\mathbb M^3_s: \frac{\Vert\sig^D\Vert}{\sqrt{2}}\leq \tau_c \right\},\\
        \mathcal K_0 & = \left\{(\Sigma,T_1,T_2)\in\mathbb R^3: \sqrt{T_1^2+T_2^2}\le\tau_c\right\},\\
        \H_{\K_0}(\p) & =
            \begin{cases}
                \tau_c\sqrt{2}\,\Vert\p\Vert & \text{ if } \mathrm{tr}(\p)= 0,\\
                +\infty & \text{ otherwise,}
            \end{cases}\\
        \H_{\mathcal{K}_0}(\Delta,D_1,D_2) & =
            \begin{cases}
                \tau_c\,\sqrt{\Delta^2+D_1^2+D_2^2} & \text{ if } \Delta = 0,\\
                +\infty & \text{ otherwise.}
        \end{cases}
    \end{align*}

\item \emph{Maximum average stress criterion}. It bounds the
trace of the stress tensor. It gives:
    \begin{align*}
        \K_0 & = \left\{ \sig\in\mathbb M^3_s: \frac{\mathrm{tr}(\sig)}{3}\leq \sigma_c \right\},\\
        \H_{\K_0}(\p) & =
            \begin{cases}  
                \sigma_c\,\mathrm{tr}(\p) & \text{ if } \corr{\p=\tfrac13\mathrm{tr}(\p)\,\mathbf I,\ }\mathrm{tr}(\p)\geq 0,\\
                +\infty & \text{ otherwise.}
            \end{cases}
    \end{align*}
    The intrinsic domain in this case is the whole space: $\mathcal{K}_0 = \mathbb{R}^3$. 

\end{itemize}
It is also possible to define a tension cut-off on the von Mises criterion, as done for Tresca's, but this is not reported here explicitly.

\begin{remark}
    \label{rem:jumpcompatibility}
    The intrinsic domains of the Drucker-Prager and the Mohr-Coulomb criteria coincide. 
 However, differently from the Mohr-Coulomb criterion, only the points on $\partial\K_0$ with  
 $\sigII=(\sigI+\sigIII)/2-\sin\varphi_k(\sigI-\sigIII)/2$ can give a stress vector on the boundary of the intrinsic section. Only for these stress states is the normal to $\partial \K_0$ compatible with a displacement jump.
 Similarly, the Tresca and the von Mises criteria have the same intrinsic domain, but the normal to the von Mises strength domain is compatible with a displacement jump if and only if $\sigII=(\sigI+\sigIII)/2$, while all normals to the boundary of the  Tresca domain are jump-compatible.
 \end{remark}

 \begin{remark}
    When $k = 1/(2\sqrt{3})$, then $\varphi_k = \pi/2$ and \corr{the intrinsic domain is $\{(\Sigma,T_1,T_2):\Sigma<\sigma_c\}\cup\{(\sigma_c,0,0)\}$, as for the Rankine criterion: the Mohr's circles reach $\Sigma=\sigma_c$ only at $T_1=T_2=0$}. If $k > 1/(2\sqrt{3})$, all points in the $(\Sigma, T)$-plane are covered by Mohr's circles and the intrinsic domain becomes all of $\mathbb{R}^3$, with no intrinsic curve. Thus, displacement jumps are not possible in this case.
    \end{remark}

\section{The sharp-interface cohesive fracture model}
\label{sec:SharpLimit}
We are now in a position to  identify a sharp interface cohesive model that is compatible with the phase-field model in the limit \(\ell\to  0\).
As throughout the rest of the paper, our approach is for a large part formal. 
Because of the term \(\w(\alpha)/\ell\), as \(\ell \to 0\), the phase-field variable \(\alpha\) must vanish almost everywhere for the phase-field energy functional~\eqref{eq:phase-field-energy3} to remain finite.
Nevertheless, a nontrivial trace of \(\alpha\) persists along the surfaces where the displacement exhibits discontinuities. 
In the following, we infer the energy of the limit model and its evolution conditions by this requirement and the solution of the {multiaxial model problem} of the previous section.
Hence, we present several examples of the response of sharp interface models considering specific strength criteria among those introduced in Section~\ref{sec:strength}. 

\subsection{The limit energy and the governing equations in the bulk and on the jump sets}
Before proceeding with the formulation of a possible limit model, we simplify the energy expressions by introducing a change of variable for the damage parameter. Indeed, the expressions~\eqref{eq:Esur}--\eqref{eq:Phi} involve both functions \(\k(\alpha)\) and \(\w(\alpha)\), although  only a combination of the two is relevant in the limit $\ell\to 0$.
To this end, we define the normalized variable \(\hat\alpha\) and normalized strength degradation function $\hat\k(\hat\alpha)$
\begin{equation}
    \label{eq:hatalpha}
    \hat\alpha={f_\w}(\alpha),\qquad \hat\k(\hat\alpha) = \k(f_\w^{-1}(\hat{\alpha})),
\end{equation}
where ${f_\w}$ is the monotonically increasing  function introduced in Equation~\eqref{eq:Dis} and $f_\w^{-1}$ its inverse. The new damage variable $\hat\alpha$ represents the portion of the  energy dissipated in a localized solution with a maximum damage value  $\alpha$, with respect to the energy dissipated for a fully developed crack where $\alpha$ reaches 1. 
We provide below a specific example, which we will keep throughout this section.
\begin{example}[$\mathsf{M2}$ model]
    \label{ex:M2-sharp}
    For the model $\mathsf{M2}$  introduced in equation~\eqref{eq:model-LS}, the change of variable above  is particularly simple:
    \begin{equation*}
        f_\w(\alpha)=\alpha ^{2},\qquad\hat\k(\hat\alpha):=\k(f_\w^{-1}(\hat\alpha))=1-{\hat \alpha}^{\zeta/2},\quad \zeta\in[1,2).
    \end{equation*}
\end{example}

Assume that  in the limit  of $\ell\to 0$, the solutions of the phase-field model are of finite energy.
Because of the term ${\w(\alpha)}/{\ell}$ in the energy, $\alpha$ must vanish almost everywhere on $\Omega$ when $\ell\to 0$. 
We identify the cracks as the sets of co-dimension $1$ where $\alpha$ is not vanishing and the displacement may jump.
Denoting by $J_\u$ the surface where the displacement jumps, the material behavior is composed of a bulk contribution on $\Omega\setminus J_\u$ and a surface term on $J_\u$. 
We introduce the energy functional of the sharp interface model in the following form:
\begin{equation*}
        \corr{\mathcal{E}_0}(\u,\hat\alpha)=
        \int_{\Omega\setminus J_\u} 
       \psi_{0}(\veps(\u))
        \,\mathrm{d} V
        +
        \int_{J_\u}\phi_\n(\hat\alpha,\jump{\u}) \,\mathrm{d}S.
\end{equation*}

The following hypotheses mirror Hypothesis~\ref{hyp:kw} on the functions \(\k\) and \(\w\) in the phase-field model for the function \(\hat{\k}\) in the sharp-interface model.

\begin{hypothesis}
   The function \(\hat\k\) \corr{defined by the change of variable~\eqref{eq:hatalpha}} satisfies the following properties:

   \begin{enumerate}
    \item \(\hat\k\)  is continuous on $[0,1]$, continuously differentiable on \corr{$(0,1]$}, and satisfies \(\hat\k(0) = 1\) and \(\hat\k(1) = 0\). 
    \item Its derivative \(\hat\k'\) is negative or zero, diverges to \(-\infty\) as \(\hat\alpha \to 0\), and is strictly increasing. In particular, \(\hat\k\) is strictly convex, and \(\hat\k'\) can vanish only at \(\hat\alpha = 1\).
   \end{enumerate}
\end{hypothesis}

Indeed, since
    \begin{equation*}
        \hat\k'(\hat\alpha):=
        \frac{\mathrm{d}\hat\k(\hat\alpha)}{\mathrm{d}\hat\alpha}=
        \frac{\mathrm{d}\k(\alpha)}{\mathrm{d}\alpha} \frac{\mathrm{d}{f_\w^{-1}}(\hat{\alpha})}{\mathrm{d}\hat\alpha}
        =
        \left.c_\w
        \frac{\k'(\alpha)}{\sqrt{\w(\alpha)}}
        \right\vert_{\alpha=f^{-1}_\w(\hat{\alpha})},
    \end{equation*}
the properties of $\hat\k'$ are consistent with Hypothesis~\ref{hyp:kw}.
The divergence of $\hat\k'$ at $\hat\alpha = 0$ follows from the condition \corr{$\sqrt{\w(\alpha)}/\vert\k'(\alpha)\vert\to0$ as $\alpha\to0$ in Hypothesis~\ref{hyp:kw}}.
    
Starting from the solution of the {multiaxial model problem} for the phase-field model with a finite \(\ell\), as established in Section~\ref{sec:model-problem-loc}, we find that consistency with the phase-field model in the limit \(\ell \to 0\) requires the evolution of the sharp-interface problem to satisfy the following conditions, both in the bulk region \(\Omega \setminus J_\u\) and across the cohesive interfaces \(J_\u\):

\begin{enumerate}
\item {\it  Bulk behavior.} 
In the bulk $\Omega\setminus J_\u$ the material behavior is nonlinear elastic without damage (standard Hencky-like law) determined by  the linear elastic stiffness $\A$, the closed convex set of admissible stresses $\K_0$, and the normality rule for the nonlinear deformations:

\begin{equation}
    \label{eq:bulk-limit-behaviour}
    \sig= \A(\veps-\mathbf p),\quad
    \sig\in \K_0,\quad
    (\sig-\sig^*)\cdot\mathbf p\ge 0, \quad\forall\sig^*\in\K_0.
\end{equation}
The constitutive law can also be written in the form 
\begin{equation}
\sig\in\dfrac{\partial\psi_{0}}{\partial\veps}(\veps),
\label{eq:sig-sharp}
\end{equation}
where the derivative in~\eqref{eq:sig-sharp} is to be interpreted as a subdifferential wherever $\psi_0$ is not smooth.
\item {\it Surface behavior.} At a point of the jump set $J_\u$ with normal $\n$, the relation between displacement jumps $\jump{\u}$ and surface traction $\sig\n$ is a non-smooth interfacial law with an activation criterion, without elasticity, but exhibiting a progressive degradation of the activation threshold:
\begin{equation}
    \label{eq:surface-limit-behaviour}
    \jump{\u}\cdot\n\ge0,\quad
    \sig\n\in \corr{\hat\k}(\hat\alpha)\mathcal K_0(\n),\quad
    (\sig\n-\sig^*\n)\cdot\jump{\u}\ge 0\quad \forall \sig^*\n\in \corr{\hat\k}(\hat\alpha)\mathcal K_0(\n)
\end{equation}
and 
\begin{equation*}
    \dfrac{\partial\hat\alpha}{\partial t}\ge0,\quad
    \dfrac{\partial\phi_\n}{\partial\hat\alpha}(\hat\alpha,\jump{\u})\ge0,\quad
    \dfrac{\partial\phi_\n}{\partial\hat\alpha}(\hat\alpha,\jump{\u})\dfrac{\partial\hat\alpha}{\partial t}=0
\end{equation*}
\begin{equation}
    \label{eq:phin}
    \phi_\n(\hat\alpha,\jump{\u})=\hat\k(\hat\alpha)\H_{\mathcal K_0(\n)}(\jump{\u})+\Gc\hat\alpha.
\end{equation}

\corr{For an isotropic strength criterion, $\H_{\mathcal K_0(\n)}(\jump{\u})=\H_{\mathcal K_0}(\Delta,D_1,D_2)$, where $\Delta$ and $(D_1,D_2)$ are the normal and tangential components of the jump in Mohr's representation (Section~\ref{sec:Mohr}). The density~\eqref{eq:phin} depends on $\n$ only through these components.}

In the sharp interface model, the surface energy density \(\Phi(\bar\alpha)\) defined in~\eqref{eq:Phi} becomes:
\begin{equation}
    \label{hatPhi}
    \hat\Phi(\hat\alpha) = \left( \hat\alpha - \frac{\hat\k(\hat\alpha)}{\hat\k'(\hat\alpha)} \right)\Gc.
\end{equation}
The assumption \(\hat\k'(0) = -\infty\) ensures that  \(\hat\Phi(0) = 0\). As a result, the function \(\hat\Phi(\hat\alpha)\) increases from \(0\) to \(\Gc\) as \(\hat\alpha\) varies from \(0\) to \(1\). The variable \(\hat\alpha\) can thus be interpreted as the fraction of surface energy dissipated by damage, relative to the total Griffith surface energy.

Rheologically, this corresponds to a frictional slider without a spring, whose slip threshold gradually decreases. The evolution of the threshold is itself described by a rate-independent law governed by an energy functional, under the standard framework of irreversibility, stability, and energy balance.
\end{enumerate}

The evolution law for the damage variable is merely a reinterpretation of the first-order stability condition applied to the discontinuity surfaces, see~\eqref{STsaut}, combined with the consistency condition provided by the energy balance, which states that damage evolves only when the activation criterion is met. Owing to the convexity of $\phi_\n$ with respect to $\hat\alpha$, for a fixed displacement jump, the evolution of $\hat\alpha$ can be analyzed locally: gradient terms vanish, and the problem reduces to the minimization of the surface energy.\\

{\color{red}
For a jump $\mathbf d$ across a surface of normal $\n$, the minimum of the surface energy over the damage is
\begin{equation}
\label{phistar}
\min_{\hat\alpha\in[0,1]}\phi_\n(\hat\alpha,\mathbf d)=\Gc\,\phi_*\Big(\frac{\H_{\mathcal K_0(\n)}(\mathbf d)}{\Gc}\Big),
\qquad
\phi_*(\delta):=\min_{\hat\alpha\in[0,1]}\big\{\hat\k(\hat\alpha)\,\delta+\hat\alpha\big\},\quad \delta\ge0.
\end{equation}
We denote by $\hat\alpha_*(\delta)$ the minimizer, unique since $\hat\k$ is strictly convex. The function $\phi_*$ is concave, as a minimum of affine functions of $\delta$, and
\begin{equation}
\label{eq:phistar-envelope}
\phi_*'(\delta)=\hat\k\big(\hat\alpha_*(\delta)\big),\qquad \phi_*(\delta)-\delta\,\phi_*'(\delta)=\hat\alpha_*(\delta).
\end{equation}
}

{\color{red}
\begin{remark}[Convexity in the crack normal]\label{rem:convexity-normal}
The energy $\mathcal E_0$ is the candidate limit suggested by flat cracks, and its surface energy is that of a flat crack. A functional of this form is lower semicontinuous only if, in any cell, the flat crack carrying a given jump costs less than every other crack configuration that coincides with it near the boundary of the cell; this condition was introduced by \citet{AmbBra90} for the interfacial energies of piecewise constant fields, see also~\cite{AmbFusPal00}, and studied for displacement fields of bounded deformation in~\cite{DalOrlToa17,FriPerSol21}. At fixed jump, it requires the surface energy to be convex in the area-weighted normal: otherwise a fine zigzag crack spanning the same mean plane dissipates less than the flat crack. Here the density is the concave function $\phi_*$ of a strength $\H_{\mathcal K_0(\n)}(\mathbf d)$ that depends on the angle between jump and normal. Convexity in the normal can then fail, through the softening of the cohesive law combined with the variation of the strength with the orientation. The condition is Herring's criterion for the stability of a flat crystal face against a hill-and-valley structure~\cite{Her51}; it separates weak from strong anisotropy of the crack surface energy in~\cite{LiMau19}. Where it fails, the $\Gamma$-limit carries a relaxed surface energy, smaller than~\eqref{eq:surface-energy-densities-limit-model} for some jumps and normals. Its identification is out of the scope of the present paper and is left to future work. Relaxation results of this type are given in~\cite{BouBraBut95,AmbBra90,BouFonMas98}, and \citet[Prop.~3.2]{DalOrlToa17} give an example of a functional of this form that is not lower semicontinuous in $BD$.
\end{remark}}

\begin{remark}
        In the following, we leverage the evolution conditions~\eqref{eq:bulk-limit-behaviour}--\eqref{eq:surface-limit-behaviour} together with the loading condition~\eqref{eq:loading-condition} of the {multiaxial model problem} to construct explicit solutions for selected examples.
        In the time-discrete setting, a possible variational formulation for the evolution problem of the sharp-interface limit model defines the solution $(\u_{i+1}, \hat\alpha_{i+1})$ at time step $t_{i+1}$, given the damage field $\hat\alpha_{i}$ and the cumulated jump set $\Gamma_{i} := J_{\u_i} \cup \Gamma_{i-1}$ at time step $i$, as the solution to the following minimization problem:
\begin{equation*}
    \inf_{\u \in \mathcal{C}_{i+1}}
    \left(
    \min_{\hat\alpha \in \mathcal{A}_{i+1}(J_\u)}
    \corr{\mathcal{E}_0}(\u, \hat\alpha)
    \right),
\end{equation*}
where $\mathcal{C}_{i+1}$ denotes the space of admissible displacements at step $i+1$, consisting of fields in $BD(\Omega)$ satisfying the hard-device loading condition at time $t_{i+1}$. 
For a candidate displacement $\u$, the damage field in the sharp-interface model is defined on $\Gamma := J_\u \cup \Gamma_{i-1}$. A possible functional setting for the damage field is
\begin{equation*}
    \mathcal{A}_{i+1}(J_\u) := \left\{ \hat\alpha \in L^\infty(J_\u \cup \Gamma_{i-1}) \;: \; \hat\alpha_{i} \leq \hat\alpha \leq 1 \right\}.
\end{equation*}
Although a time-continuous variational formulation based on a threefold evolution principle---irreversibility, stability, and energy balance---could also be considered, its statement requires several additional modeling choices and is therefore not pursued here. 
\end{remark}

\subsection{The cohesive law on the jump set}
We analyze here the cohesive interface law obtained in the sharp interface model. To this end, we solve the evolution problem for the damage $\hat\alpha$ at a generic point of the jump set $J_\u$ for a given displacement jump history $\{\jump{\mathbf u_t}\}_{t\in(t_0,t_1)}$. We treat first the general case, and then provide an explicit example for a specific model. We briefly discuss the effect of the damage irreversibility on the cohesive law when the displacement jump  is not monotonically increasing.

\begin{proposition}[Evolution of damage and cohesive force at a point on a discontinuity surface] 
    \label{P-Evolhat1}
    Suppose that the evolution of the displacement jump at a point of a discontinuity surface with normal vector $\n$ is given, for $t \ge 0$, by
    \[
    \jump{\mathbf u_t} = t \,\mathbf d_0,
    \]
    where $\mathbf d_0$ is a unit vector satisfying $\mathbf d_0 \cdot \n \ge 0$, and such that
    \(
    F_0 := \mathrm{H}_{\mathcal K_0(\n)}(\mathbf d_0) < +\infty.
    \)
    Then, the following properties hold:
\begin{enumerate}
    \item The damage evolution problem at this point, starting from an initial damage level $\hat\alpha_0 < 1$, admits a unique time-continuous solution, which is given by
\begin{equation*}
    \hat\alpha_t = 
    \begin{cases}
        \hat\alpha_0 & \text{if } t \le t_0:= \dfrac{\Gc}{|\hat\k'(\hat\alpha_0)| F_0}, \\
        (-\hat\k')^{-1} \left( \dfrac{\Gc}{t F_0} \right) & \text{if } t_0 < t < t_1, \\
        1 & \text{if } t \ge t_1:=\dfrac{\Gc}{|\hat\k'(1)| F_0}.
    \end{cases}
\end{equation*}
\item  The solution satisfies
\[
\hat\alpha_t = \argmin_{\hat\alpha_0 \le \hat\alpha \le 1} \, \phi_\n(\hat\alpha, t\, \mathbf d_0).
\] 
\item For $t \ge t_0$, the cohesive force, \textit{i.e.}, the component of the stress vector collinear with $\mathbf d_0$, has magnitude $F_t = \hat\k(\hat\alpha_t) F_0$, which decreases with time and vanishes at the finite time $t_1$ provided that $\hat\k'(1) \ne 0$. This magnitude coincides with the derivative of the minimum  surface energy density with respect to $t$, which is itself equal to $\hat\Phi(\hat\alpha_t)$ as defined in~\eqref{hatPhi}:
\begin{equation*}
\phi_\n(\hat\alpha_t, t\,\mathbf d_0) = \hat\Phi(\hat\alpha_t), \qquad
F_t = \frac{d}{dt} \phi_\n(\hat\alpha_t, t\,\mathbf d_0).
\end{equation*}
As a consequence, $\phi_\n(\hat\alpha_t, t\,\mathbf d_0)$ is a strictly increasing and strictly concave function of $t$ over the interval $(t_0, t_1)$, and remains constant at $\Gc$ for $t \ge t_1$, provided $t_1$ is finite.
\end{enumerate}
\end{proposition}
\begin{proof} 
\begin{enumerate}
    \item This is essentially a classical result, which follows from the strict convexity of $\phi_\n$ with respect to $\hat\alpha$. We briefly recall the main steps of the proof; see~\cite[Proposition 2.4]{BouFraMar08} for a similar result. One first shows that there exists a unique, continuous, and non-decreasing solution $\hat\alpha_t$ starting from $\hat\alpha_0$ that satisfies the inequality
    \[
    - t \,\hat\k'(\hat\alpha_t) F_0 \le \Gc, \quad \forall t \ge 0,
    \]
    and increases only when equality holds. If $\hat\alpha_0 > 0$, then the inequality is strict at $t=0$, and remains so by continuity until a time $t_0$ such that $t_0 |\hat\k'(\hat\alpha_0)| F_0 = \Gc$, which is finite due to the assumption on $\mathbf d_0$ and the fact that $\hat\k'$ can only vanish at $1$. Then, for $t > t_0$, equality holds at all times due to the negativity and monotonicity of $\hat\k'$, the non-decreasing character of $\hat\alpha$, and continuity (see the argument used to prove equality in~\eqref{STsaut2}). This equality yields an explicit expression for $\hat\alpha_t$ for all $t > t_0$, and $\hat\alpha_t$ reaches $1$ at a time $t_1$ such that $t_1 |\hat\k'(1)| F_0 = \Gc$, which is finite if $\hat\k'(1) < 0$ and infinite if $\hat\k'(1) = 0$. In the case where $\hat\alpha_0 = 0$, we simply have $t_0 = 0$, and damage starts evolving immediately from $t = 0$.
    
    \item We now show that the function $\hat\alpha_t$ introduced above indeed minimizes the energy at each time $t$. This is obvious for $t = 0$. For $t > 0$, the energy is a continuously differentiable and strictly convex function of $\hat\alpha$ at fixed $t$. As such, it reaches its minimum
    at $\hat\alpha_0$ if the derivative at $\hat\alpha_0$ is nonnegative,
 at $1$ if the derivative at $1$ is nonpositive,
     and at the unique point in $(\hat\alpha_0, 1)$ where the derivative vanishes, otherwise.
    Since the derivative reads
    \[
    \frac{\partial \phi_\n}{\partial \hat\alpha}(\hat\alpha, t \,\mathbf d_0) = t \,\hat\k'(\hat\alpha) F_0 + \Gc,
    \]
    this brings us back to the previous discussion involving the inequality criterion, which precisely governs the activation of damage and the evolution of $\hat\alpha_t$.
    
    \item During the evolution of $\hat\alpha$, the stress vector remains on the boundary of the intrinsic domain $\hat\k(\hat\alpha_t) \mathcal K_0(\n)$, which shrinks as damage increases. Since
    \[
    \sig_t \n \cdot \mathbf d_0 = \hat\k(\hat\alpha_t) \H_{\mathcal K_0(\n)}(\mathbf d_0),
    \]
    the magnitude of the cohesive force is $\hat\k(\hat\alpha_t) F_0$, and hence decreases with time.
    
    In terms of energy, we have
    \[
    \phi_\n(\hat\alpha_t, t\, \mathbf d_0) = t \hat\k(\hat\alpha_t) F_0 + \hat\alpha_t \Gc,
    \]
    which coincides with $\hat\Phi(\hat\alpha_t)$ defined in~\eqref{hatPhi}, since $t \,\hat\k'(\hat\alpha_t) + \Gc = 0$ when $t_0 < t < t_1$, and $\hat\alpha_t = 1$ for $t \ge t_1$. This directly implies that the energy increases with $t$.
    Differentiating $\phi_\n$ 
    with respect to $t$ yields that the cohesive force $F_t = \hat\k(\hat\alpha_t) F_0$ corresponds to the derivative of the energy, due to the consistency condition:
    \(
    \left( t \,\hat\k'(\hat\alpha_t) F_0 + \Gc \right) \frac{d \hat\alpha_t}{dt} = 0.
    \)
    From this, one deduces the concavity of the energy as a function of time.
    
\end{enumerate}
\end{proof}

Note that only the cohesive force necessarily vanishes as the displacement amplitude increases; the component of the stress vector orthogonal to $\mathbf d_0$ may remain unaffected by the damage process. An illustration of this behavior will be provided with the von Mises-type model developed in the following section.

This result can be extended to arbitrary displacement jump paths.
An analogous uniqueness result still holds, and the evolution of the damage variable can be made explicit as a function of the displacement jump trajectory. We state the result without proof, as the argument is a straightforward extension of the previous one.

\begin{proposition}
    \label{prop:Evolhat2}
    Starting from an initial damage level $\hat\alpha_0 \in [0,1)$, consider a displacement jump trajectory $t \mapsto \mathbf d_t$ originating from $\mathbf 0$, and satisfying $\mathbf d_t \cdot \n \ge 0$ and $\mathrm{H}_{\mathcal K_0(\n)}(\mathbf d_t) < +\infty$. Then, at time $t$, the damage variable $\hat\alpha_t$ depends only on 
    \[
        {H}^{*}_t: = \max_{s \le t} \mathrm{H}_{\mathcal K_0(\n)}(\mathbf d_s),
    \]
    and is given by
    \[
    \hat\alpha_t =
    \begin{cases}
    \hat\alpha_0 & \text{if } {H}^{*}_t \le \dfrac{\Gc}{|\hat\k'(\hat\alpha_0)|}, \\[1ex]
    (-\hat\k')^{-1}\left(\dfrac{\Gc}{ {H}^{*}_t}\right) & \text{if } \dfrac{\Gc}{|\hat\k'(\hat\alpha_0)|} <  {H}^{*}_t < \dfrac{\Gc}{|\hat\k'(1)|}, \\[1ex]
    1 & \text{if }  {H}^{*}_t\ge \dfrac{\Gc}{|\hat\k'(1)|}.
    \end{cases}
    \]    
\end{proposition}
\begin{remark}
The knowledge of the displacement jump is sufficient to determine the damage variable. As a result, the question of the regularity of the damage reduces to that of the displacement jump itself. From a mathematical standpoint, the introduction of a memory variable on the discontinuity surfaces does not entail additional difficulties.
\end{remark}

\subsubsection{Example of a cohesive model and its associated response}
We consider the following family of functions for $\hat\k$, depending on a parameter $\hat\zeta$:
\begin{equation}
\label{eq:kappak}
\hat\k(\corr{\hat\alpha}) = 1 - \corr{\hat\alpha}^{\hat\zeta} \quad \text{with} \quad 0 < \hat\zeta < 1,
\end{equation}
which satisfies all the required properties. 
As shown in Example~\ref{ex:M2-sharp}, this model can be regarded as the sharp interface limit of the model $\mathsf{M2}$~\eqref{eq:model-LS} with a parameter  $\zeta=2{\hat\zeta}$\corr{, for $\hat\zeta\in[1/2,1)$}.
 
The limiting case $\hat\zeta \to 1$ corresponds to the Dugdale model, which is excluded here since it yields cohesive forces that are not strictly decreasing as damage increases. The limiting case $\hat\zeta \to 0$ leads to cohesive forces that drop immediately to zero (as in Griffith's model), yet starting from a finite initial value. As a consequence, the total cohesive energy vanishes, making the model physically unacceptable.

When $0 < \hat\zeta < 1$, since $\hat\k'(1) = -\hat\zeta < 0$, complete rupture with vanishing cohesive force is reached in finite time. The corresponding surface energy density $\hat\Phi(\corr{\hat\alpha})$, which characterizes the degradation level of the surface, reads
\[
\hat\Phi(\corr{\hat\alpha}) = \frac{1}{\hat\zeta} \left( \corr{\hat\alpha}^{1 - \hat\zeta} - (1 - \hat\zeta)\corr{\hat\alpha} \right) \Gc.
\]

As shown in Proposition~\ref{P-Evolhat1}, the determination of the cohesive law requires solving the minimization problem~\corr{\eqref{phistar}}. With the choice~\eqref{eq:kappak} for $\hat\k$, the solution writes as:
\[
\corr{\hat\alpha_*} =
\begin{cases}
(\hat\zeta \delta)^{1/(1 - \hat\zeta)} & \text{if } 0 \le \hat\zeta\delta \le 1, \\
1 & \text{if } \hat\zeta\delta \ge 1,
\end{cases}
\qquad
\phi_*(\delta) =
\begin{cases}
\delta - \left( \dfrac{1}{\hat\zeta} - 1 \right) (\hat\zeta \delta)^{1/(1 - \hat\zeta)} & \text{if } 0 \le \hat\zeta\delta \le 1, \\[1ex]
1 & \text{if } \hat\zeta\delta \ge 1,
\end{cases}
\]

Under the conditions of Proposition~\ref{P-Evolhat1}, consider a surface element with normal $\n$ subjected to a displacement jump of increasing amplitude in the direction $\mathbf d_0$, such that $F_0 = \H_{\mathcal K_0(\n)}(\mathbf d_0) < +\infty$. Then, the surface energy and the cohesive force are
\begin{equation}
\label{phiFdelta}
\corr{\Gc\,\phi_*(\delta)\quad\text{and}\quad F=F_0\,\phi_*'(\delta)},
\quad \text{with} \quad
\delta = \frac{F_0 \|\mathbf d\|}{\Gc}.
\end{equation}
The graphs of \corr{the surface energy and the cohesive force} are shown in Figure~\ref{fig:ModeleCohesif} for three values of the parameter: $\hat\zeta = 1/4$, $1/2$, and $3/4$, corresponding to $\zeta=\corr{1/2}$, $1$, $\corr{3/2}$, respectively\corr{; the value $\hat\zeta=1/4$ lies outside the range of $\mathsf{M2}$}. The surface energy reaches $\Gc$, and the cohesive force vanishes when $\delta$ reaches $1/\hat\zeta$. In particular, for $\hat\zeta = 1/2$, the decay of the cohesive forces is linear.

\begin{figure}[th] 
   \centering
   \includegraphics[width=\textwidth]{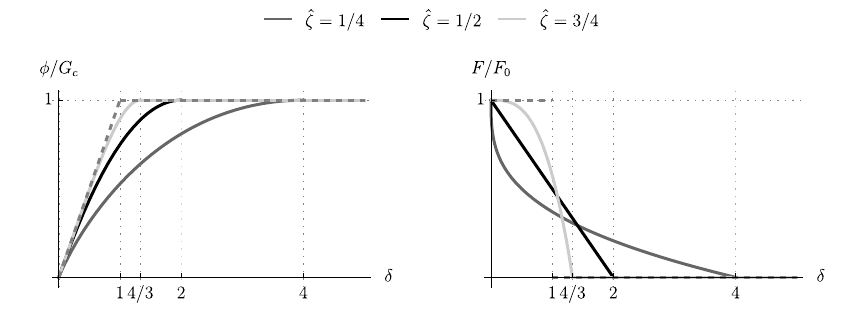} 
   \caption{Response of a surface element to an imposed displacement jump of increasing amplitude in a given direction, for cohesive models from the family~\eqref{eq:kappak} with $\hat\zeta = 1/4$, $1/2$, or $3/4$: the surface energy (left) and the cohesive force (right) are shown as functions of the non-dimensional displacement jump amplitude $\delta:=F_0\Vert\mathbf{d}\Vert/\Gc$. The dashed curve corresponds to the Dugdale model, \emph{i.e.}, the limiting case $\hat\zeta = 1$.}
   \label{fig:ModeleCohesif}
\end{figure}
If we repeat this loading test by unloading when $\delta$ reaches $\delta_1$, progressively reducing the displacement jump amplitude to zero, then, since damage no longer evolves, the cohesive force remains constant at the value reached at $\delta_1$, namely $F_0 \phi_*'(\delta_1)$, and the surface energy decreases linearly. When the displacement jump reaches zero, all the elastic surface energy has been recovered, but a residual surface energy $\hat\alpha_1 \Gc$ remains, corresponding to the damage level $\hat\alpha_1 = (\hat\zeta \delta_1)^{\corr{1/(1-\hat\zeta)}}$ attained during loading; see Figure~\ref{fig:EssaiCohesif} and Proposition~\ref{prop:Evolhat2}.
If the surface element is then reloaded in the same direction, the cohesive force stays fixed at $F_0 \phi_*'(\delta_1) = \hat\k(\hat\alpha_1) F_0$, and damage does not evolve as long as $\delta \le \delta_1$. As soon as $\delta > \delta_1$, damage starts increasing again, and both the surface energy and the cohesive force are once more given by~\eqref{phiFdelta}.
If we alternate such loading phases with unloading-reloading phases, we obtain a response as illustrated in Figure~\ref{fig:EssaiCohesif}: unloading-reloading phases are reversible, while overloading phases are irreversible, leading eventually to full damage $\hat\alpha = 1$ with surface energy $\Gc$ and vanishing cohesive force.

\begin{figure}[t] 
   \centering
   \includegraphics[width=\textwidth]{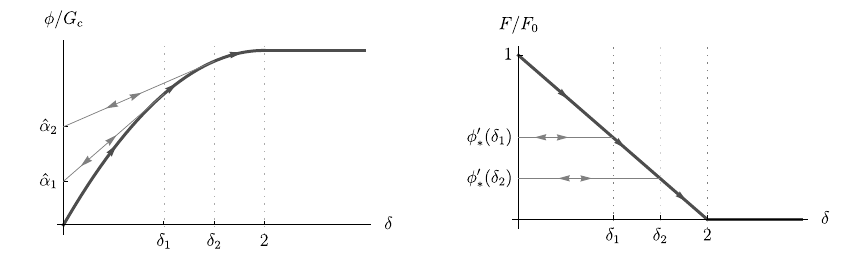} 
   \caption{Response of a surface element subjected to loading phases interspersed with unloading-reloading cycles in a prescribed displacement-jump direction for the model corresponding to $\hat\zeta = 1/2$ in equation~\eqref{eq:kappak}. 
   The unloading-reloading takes place at constant cohesive stress (light-gray lines), being purely elastic.}
   \label{fig:EssaiCohesif}
\end{figure}

\subsection{Examples with specific strength criteria}
\label{sec:Examples}
We present several examples of isotropic strength criteria to illustrate the proposed approach. 
For each case, we define the constitutive behavior of the sharp-interface model and discuss the dependence on the loading direction in the stress space through the solution of the multiaxial model problem of Section~\ref{sec:model-problem-formulation}. 

To characterize the cohesive law, we  adopt Mohr's representation of the stress vector and the displacement jump introduced in Section~\ref{sec:Mohr}, namely
$\boldsymbol{\sigma}\mathbf{n} = \Sigma\,\mathbf{n} + \mathbf{T}$ and $\mathbf{d} = \Delta\,\mathbf{n} + \mathbf{D}$,
where $\Sigma$ and $\Delta$ are the normal components, and $\mathbf{T}$ and $\mathbf{D}$ denote the tangential parts.

\subsubsection{von Mises criterion}
\label{sec:VM}

For the von Mises criterion,  the strength domain $\K_0$  in the stress space and the associated support functions $\H_{\K_0}(\p)$  are given in Section~\ref{sec:Mohr-DP}.
Thus, assuming isotropic elastic behavior, the bulk strain energy density is expressed as
\[
\varphi_{_0}(\boldsymbol{\varepsilon},\mathbf{p})=
\begin{cases}
    \mu_0\|\boldsymbol{\varepsilon}^D-\mathbf{p}\|^2+\frac{\kappa_0}{2}(\textrm{tr}\,\boldsymbol{\varepsilon})^2+\sqrt{2}\tau_c\|\mathbf{p}\|&\text{if}\quad\textrm{tr}(\mathbf{p})=0,\\
    +\infty&\text{ otherwise}.
\end{cases}
\]
where $\mu_0>0$ and $\kappa_0>0$ denote the shear and bulk moduli, respectively. Minimizing $\varphi_0(\boldsymbol{\varepsilon},\mathbf{p})$ with respect to $\mathbf{p}$ at fixed $\boldsymbol{\varepsilon}$ gives alternative expressions in terms of $\boldsymbol{\varepsilon}$  only:
\begin{equation}
\psi_{_0}(\boldsymbol{\varepsilon})=\frac{\kappa_0}{2}(\textrm{tr}\,\boldsymbol{\varepsilon})^2+
\begin{cases}
\mu_0\|\boldsymbol{\varepsilon}^D\|^2 & \text{if }\|\boldsymbol{\varepsilon}^D\|\leq\dfrac{\tau_c}{\sqrt{2}\mu_0}\\[8pt]
\sqrt{{2}}\tau_c\|\boldsymbol{\varepsilon}^D\|-\dfrac{\tau_c^2}{2\mu_0} & \text{ otherwise.}
\end{cases}
\label{psiVM}
\end{equation}
The energy density is continuously differentiable and exhibits linear growth with respect to the deviatoric strain once the latter is sufficiently large. Stresses follow by differentiation.

This criterion is isotropic since it involves only invariants of $\boldsymbol{\sigma}$ and is unbounded in the direction of spherical tensors. The associated intrinsic domain $\mathcal{K}_0$ in Mohr's representation is the circular cylinder given in Section~\ref{sec:Mohr-DP}.
Its meridian section is the strip $\mathcal{S}_0=\mathbb{R}\times[-\tau_c,\tau_c]$. Figure~\ref{fig:examples-VM} shows the strength domain (gray region in left figure) and the corresponding intrinsic section (gray region in right figure).

The surface energy density~\corr{\eqref{eq:phin}} of the limit cohesive model is determined by the support function of the intrinsic domain $\mathcal{K}_0$ in the Mohr's coordinates.
Hence, for a jump surface with normal $\mathbf{n}$ it is given by
\begin{equation*}
\phi_{\mathbf{n}}(\corr{\hat\alpha},\mathbf{d})=
\begin{cases}
\hat{\k}(\corr{\hat\alpha})\tau_c\|\mathbf{d}\|+\corr{\hat\alpha}\,\Gc & \text{if }\mathbf{d}\cdot\mathbf{n}=0,\\
+\infty & \text{ otherwise.}
\end{cases}
\end{equation*}

The condition $\Delta=\mathbf{d}\cdot\mathbf{n}=0$  implies  that the displacement jump must be purely tangential.
Cohesive forces correspond to tangential stresses $\mathbf{T}$, while normal stresses $\Sigma$ do not perform work on displacement jumps.
The cohesive law for the tangential stress is in the form characterized  by Proposition~\ref{prop:Evolhat2} and illustrated in Figures~\ref{fig:ModeleCohesif}-\ref{fig:EssaiCohesif} for $\hat\k$ given in~\eqref{eq:kappak} with 
$$
F=\Vert \mathbf{T}\Vert,\quad\delta=\frac{\tau_c \Vert \mathbf{D}\Vert}{\Gc},\quad F_0=\tau_c.
$$

\noindent This cohesive model has two limitations that deserve specific comments:
\begin{itemize}
    \item 
The strength domain $\K_0$ is unbounded in the direction of spherical tensors. A primary consequence is that if the stress direction imposed is the identity tensor $\mathbf{I}$, crack nucleation does not occur, and the material response remains linearly elastic regardless of the loading magnitude. Another consequence is that, on any fully damaged crack ($\hat{\alpha}=1$), where cohesive forces vanish, the stress vector is not necessarily zero since normal stresses remain unaffected by damage.

\item As discussed in Remark~\ref{rem:jumpcompatibility}, for $\K_0$ corresponding to the von Mises criterion it is possible to create a displacement jump when the stress $\boldsymbol{\sigma}\in\partial\K_0$ only if the ordered eigenvalues of $\sig$ verify the condition $\sigII=(\sigI+\sigIII)/2$. This means that it allows for crack nucleation only when this condition is met.
\end{itemize}

To illustrate these points, we consider a test on a cube  as in the {multiaxial model problem} of Section~\ref{sec:fundamentalProblem}, but now with the sharp-interface model obtained for $\ell\to 0$. As in Section~\ref{sec:fundamentalProblem}, for any loading direction $\s$ in the stress space, we assume that the cube size $L$ is sufficiently small with respect to the elasto-cohesive length $\ell_{ch}^\s$ so as to avoid snap-back in the evolution.

\begin{figure}[h!] 
    \begin{center}
    \begin{minipage}[c]{0.5\textwidth}
    \centering
    \includegraphics[width=\textwidth]{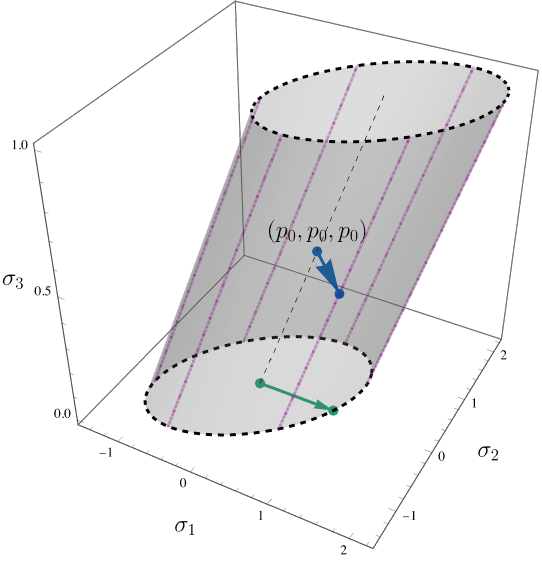}
  \end{minipage}%
  \hspace{0.02\textwidth}%
  \begin{minipage}[c]{0.4\textwidth}
    \centering
    \vspace*{0.02\textheight}
    \includegraphics[width=\textwidth]{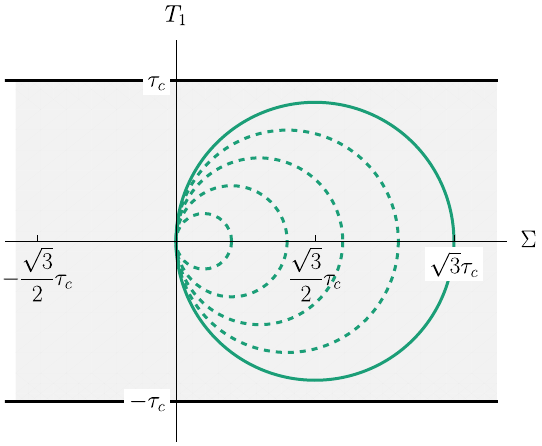}\medskip\\

    \vfill

    \includegraphics[width=\textwidth]{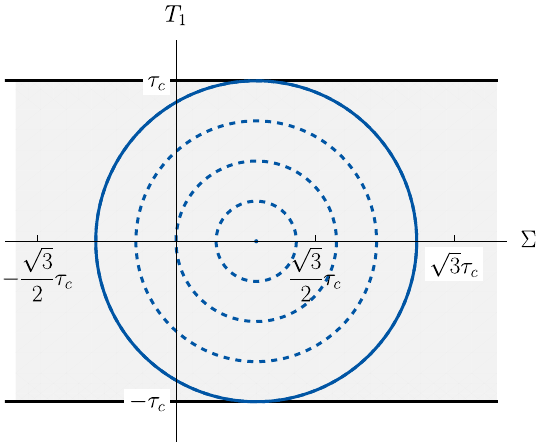}
  \end{minipage}
    \end{center}
    \caption{
Example with the von Mises strength criterion. {Left:} The cylinder represents the strength surface $\partial \K_0$ in stress space. The purple lines highlight portions of the surface where the normal is compatible with a displacement jump. The green path corresponds to a uniaxial loading example, while the blue path represents a shear loading superimposed with an isotropic (spherical) prestress $p_0$. {Right:} Mohr's representation showing the intrinsic section $\mathcal{S}_0$ and the evolving Mohr's circles corresponding to the stress paths: green for uniaxial tension (top) and blue for shear with superimposed isotropic prestress $p_0$ (bottom).
}
   \label{fig:examples-VM}
\end{figure}
 As a first example, let us consider the case of a uniaxial traction test where $\s = \mathbf{e}_1 \otimes \mathbf{e}_1$ for some unit vector $\mathbf{e}_1$, in green in Figure~\ref{fig:examples-VM}. The stress tensor $\sig_t = \sigma_t \s$ has the trivial eigenvalues $\sigI = \sigma_t$, $\sigII = \sigIII = 0$, and its magnitude is controlled through the loading condition~\eqref{eq:loading-condition} on the average strain $\bar{\veps}^\s(\u)$, coinciding here with the space-average of \corr{the normal strain $\mathbf e_1\cdot\veps\,\mathbf e_1$}.
The stress reaches the boundary of the strength domain $\K_0$ when $\sigma_t = \sqrt{3}\,\tau_c$. 
However, since $\sigII \neq (\sigI + \sigIII)/2$, it is not possible to nucleate a crack in this case. For larger loading, the response remains elastic with constant stress $\sqrt{3}\,\tau_c\,(\mathbf{e}_1 \otimes \mathbf{e}_1)$ and homogeneous strains including a nonlinear contribution governed by the normality rule.
The damage remains null, and the evolution is similar to that discussed in~\cite{Francfort2015,FRANCFORT2016125} for the case of von Mises elasto-plasticity.

Let us now consider a mixed loading test where the stress tensor  includes a fixed isotropic pre-stress $p_0$ and an evolving shear component $\tau_t$
\begin{equation*}
\sig_t = p_0 \,\mathbf{I} + \corr{2}\tau_t \,\n\odot\t,
\end{equation*}
where the amplitude of the shear component $\tau_t\geq 0$ is controlled by the loading condition  $\corr{\bar\varepsilon^\s}(\u_t)=t$ on the average strain in the direction $\s= \corr{\sqrt{2}\,\n\odot\t}$\corr{, so that $\sig_t\n\cdot\t=\tau_t$ and $\s\cdot\s=1$}.
In this case, represented in blue in Figure~\ref{fig:examples-VM}, $\sigI = p_0+\tau_t$, $\sigII =p_0$, $ \sigIII = p_0-\tau_t$. Since $\sigI+\sigIII=2\sigII$, it is possible to nucleate a cohesive crack with normal $\n$\corr{, the cell having a face normal to $\n$,} when the stress reaches the yield criterion, \emph{i.e.} when $\tau_t=\tau_c$. 
The shear then gradually decreases as $\tau_t = \corr{\hat{\k}}(\hat{\alpha}_t)\tau_c$ with increasing damage. Meanwhile, the normal stress $\boldsymbol{\sigma}_t\mathbf{n} \cdot \mathbf{n}$ remains constant and equal to $p_0$. When $\hat{\alpha} = 1$, the shear stress vanishes, but the normal stress still equals $p_0$. In other words, the von Mises criterion leads to the formation of a cohesion-less shear band rather than a genuine Griffith crack, since the resulting jump set can still support normal stresses. The residual intrinsic domain reduces to the \corr{line $\mathbb R\times\{(0,0)\}$}.
Displacement jumps remain tangential, and at full damage ($\hat{\alpha} = 1$), the tangential component becomes completely unconstrained, while the normal jump remains zero. In terms of surface energy, the condition $\phi_{\mathbf{n}}(1, \mathbf{d}_0) = \Gc$ for some directions $\mathbf{d}_0$ is not sufficient; it must hold for all directions to characterize a true Griffith-type crack.\\

\noindent Two remedies can partially address this intrinsic limitation of the von Mises criterion:

\begin{enumerate}
\item Elastic moduli, notably the bulk modulus, can be made damage-dependent, as in~\cite{AleMarVid15}, ensuring the material cannot sustain stresses once damage reaches unity. However, this approach suffers from drawbacks: volumetric behavior no longer explicitly involves damage in the limit $\ell\to 0$ (thus allowing any spherical stress). This retains the weaknesses of standard Griffith's theory and the related classical phase-field approximations.

\item Combining the von Mises criterion with a constraint on normal stresses (\emph{e.g.} maximum tension criterion). Here, both tangential and positive normal stresses continuously decrease to zero as damage approaches unity. Only negative normal stresses remain possible at the limit, which inherently cannot vanish due to non-penetration conditions. This results in a satisfactory cohesive crack model, which is particularly relevant when replacing the von Mises criterion with Tresca, as will be discussed in Section~\ref{sec:examples-tresca-cap}.
\end{enumerate}

\subsubsection{The antiplane framework and link with an existing Gamma-convergence result}

Let us consider the case of an antiplane setting as a particular reduction of the general von Mises model. The domain is taken to be a cylindrical region of the form $\Omega = \omega \times \mathbb{R}$, where $\omega$ is a bounded planar domain in the $(x,y)$-plane spanned by the basis vectors $\mathbf{e}_x$ and $\mathbf{e}_y$, and $\mathbf{e}_z$ denotes the unit vector orthogonal to this plane. The displacement field is assumed to be of the form $\mathbf{u}(\mathbf{x}) = u(x,y)\,\mathbf{e}_z$, and is thus described by the scalar function of two scalar variables $u$.
Only the $xz,yz$ (and symmetric) components of the strain and stress tensors are non-zero, and these components depend solely on the in-plane coordinates $(x, y)$. The nonlinear strain is thereby purely deviatoric and can be represented by a planar vector $\mathbf{p}$;  the model becomes effectively two-dimensional.
As a result, only the shear modulus $\mu_0$ contributes to the energy; the bulk modulus plays no role. 

In this setting, a $\Gamma$-convergence result is established in~\cite{DalOrlToa16} for a phase-field model based on the von Mises strength criterion, showing convergence towards a cohesive fracture model. Although currently unique in the literature, this result remains limited in scope: it addresses only anti-plane configurations and does not incorporate the irreversibility of the phase-field. 

The modeling assumptions adopted in~\cite{DalOrlToa16} differ in several respects from those considered here, as summarized below for comparison.
The authors consider the phase-field energy 
\[
\mathcal{E}_\ell(u, \mathbf{p}, \alpha) = \int_\Omega \left( \frac{1 - \alpha}{2}\mu_0\, (\nabla u - \mathbf{p}) \cdot (\nabla u - \mathbf{p}) + \k(\alpha)\, {\tau_c}\, \|\mathbf{p}\| +\frac{\Gc}{4c_\w} \left(\frac{\w(\alpha)}{\ell}\,  +  \ell\, \|\nabla \alpha\|^2 \right)\right) \dS.
\]

Unlike the present model, the shear modulus in~\cite{DalOrlToa16} depends explicitly on the phase-field variable and is defined as $\mu(\alpha) = (1 - \alpha)\mu_0$. The last term proportional to $\Gc$ coincides with the one considered in the present work, under similar assumptions on the potential function $\w(\alpha)$. As for the function $\k(\alpha)$ governing the evolution of the yield criterion, it is assumed to be merely non-increasing, with $\k(0) = 1$, and allowed to vanish only at $\alpha = 1$ (but not necessarily satisfying $\k(1) = 0$).
For technical reasons, admissible values of $\alpha$ are bounded from below by a parameter that tends to $0$ faster than $\ell$. This ensures that the shear modulus remains strictly positive. 
In this setting, the internal variable $\mathbf{p}$ can be eliminated by minimizing the energy $\mathcal{E}_\ell(u, \mathbf{p}, \alpha)$ with respect to $\mathbf{p}$, yielding a two-field formulation. It is then shown in~\cite{DalOrlToa16} that minimizers $(u_\ell, \alpha_\ell)$ of the phase-field energy converge to a minimizer $u$ of a limiting energy defined over functions of bounded variation:

\[
\mathcal{E}_0(u) = \int_{\Omega \setminus S_u} \hat{\psi}_0(\nabla u)\,dS + {\tau_c}  \,|D^c u|(\Omega) + \int_{S_u} \corr{\phi_{\mathrm a}}(|\jump{u}|)\,\dS,
\]
with 
\begin{equation} \label{Phi0}
\corr{\phi_{\mathrm a}}(d) = \min\left\{ \Gc,\ \phi_*\left( \frac{\tau_c d}{\Gc} \right) \Gc \right\},
\end{equation}
where $\phi_*$ is defined in~\eqref{phistar}. Here, $\hat{\psi}_0$ denotes the elastic energy density for antiplane deformations as defined in~\eqref{psiVM}, and $|D^c u|(\Omega)$ is the Cantor part of the gradient of $u$.
To obtain the expression~\eqref{Phi0} for the surface energy, one must use the change of variable $\alpha \mapsto \hat{\alpha}$ defined in~\eqref{eq:hatalpha}. \\

\noindent This result prompts several remarks:

\begin{itemize}
\item The elimination of the phase-field variable is possible only because irreversibility is not enforced. However, it is conceivable that the $\Gamma$-convergence result could be extended to account for irreversibility.
\item When $\hat{\k}(1) = 0$, one has $\phi_*(\delta) \le 1$, and thus the surface energy is exactly \corr{$\phi_{\mathrm a}(d)=\Gc\,\phi_*(\tau_c d/\Gc)$}, which corresponds to the limiting model proposed in the previous section. Without this assumption, $\phi_*$ grows linearly and exceeds $1$ beyond a certain value of $d$, reflecting the fact that the material can still carry stress and the cohesive force does not vanish. In that case, the only way to recover the Griffith-type surface energy is via the dependence of the shear modulus on $\alpha$ and its vanishing as $\alpha \to 1$. One may therefore reasonably conjecture that the $\Gamma$-convergence result would still hold if $\mu$ were independent of $\alpha$, provided $\hat{\k}(1) = 0$. Of course, a full proof would require reworking the entire argument.
\item The authors report that they were unable to extend the $\Gamma$-convergence result to the fully three-dimensional setting, citing difficulties associated with enforcing the constraint $\mathrm{tr}\left(\mathbf{p}\right) = 0$ in 3D. This limitation may stem from the fact that the limiting model cannot exhibit fully cohesive behavior in three dimensions as discussed above: the outer normals to the von Mises strength domain $\K_0$ are not, in general, compatible with a displacement jump in the 3D setting, whereas such compatibility holds in the special case of antiplane deformation. \corr{A second difficulty is the lack of lower semicontinuity of the candidate surface energy outside the antiplane setting: for the isotropic strength domain $\{\sig:\|\sig\|\le\sigma_c\}$ and $\phi_*(\delta)=\min\{\delta,1\}$, \citet[Prop.~3.2]{DalOrlToa17} show that the corresponding functional is not lower semicontinuous in $BD$ (see Remark~\ref{rem:convexity-normal}). In the antiplane setting the surface energy is convex in the normal.}
\item Throughout the present work, we omit the Cantor term in the expression of the energy of the limiting cohesive model, although its presence cannot be excluded. 
\end{itemize}

\subsubsection{A linearly incompressible material with a Drucker-Prager criterion}\label{S-Incompressible}
\label{sec:example-DP-incomp}
We consider a material where the stresses are limited by the Drucker-Prager criterion introduced in Section~\ref{sec:strength}. The corresponding expression of the strength domain $\K_0$ in the stress space and the  support function are given by Equations~\eqref{eq:DP} and~\eqref{piDP}. The intrinsic domain $\mathcal{K}_0$ and its support function coincide with those of the Mohr-Coulomb criterion. 
At any point on the intrinsic curve other than the apex $(\sigma_c, 0)$, the outer unit normal $\mathbf{d}_0$ of $\mathcal{S}_0$ satisfies $\mathbf{d}_0 \cdot \mathbf{n} = \sin \varphi_k$, while at the apex, the only requirement is $\mathbf{d}_0 \cdot \mathbf{n} \ge \sin \varphi_k$, see Figure~\ref{fig:examples-DP}-left.
\begin{figure}[t]
    \begin{center}
\includegraphics[width=.8\textwidth]{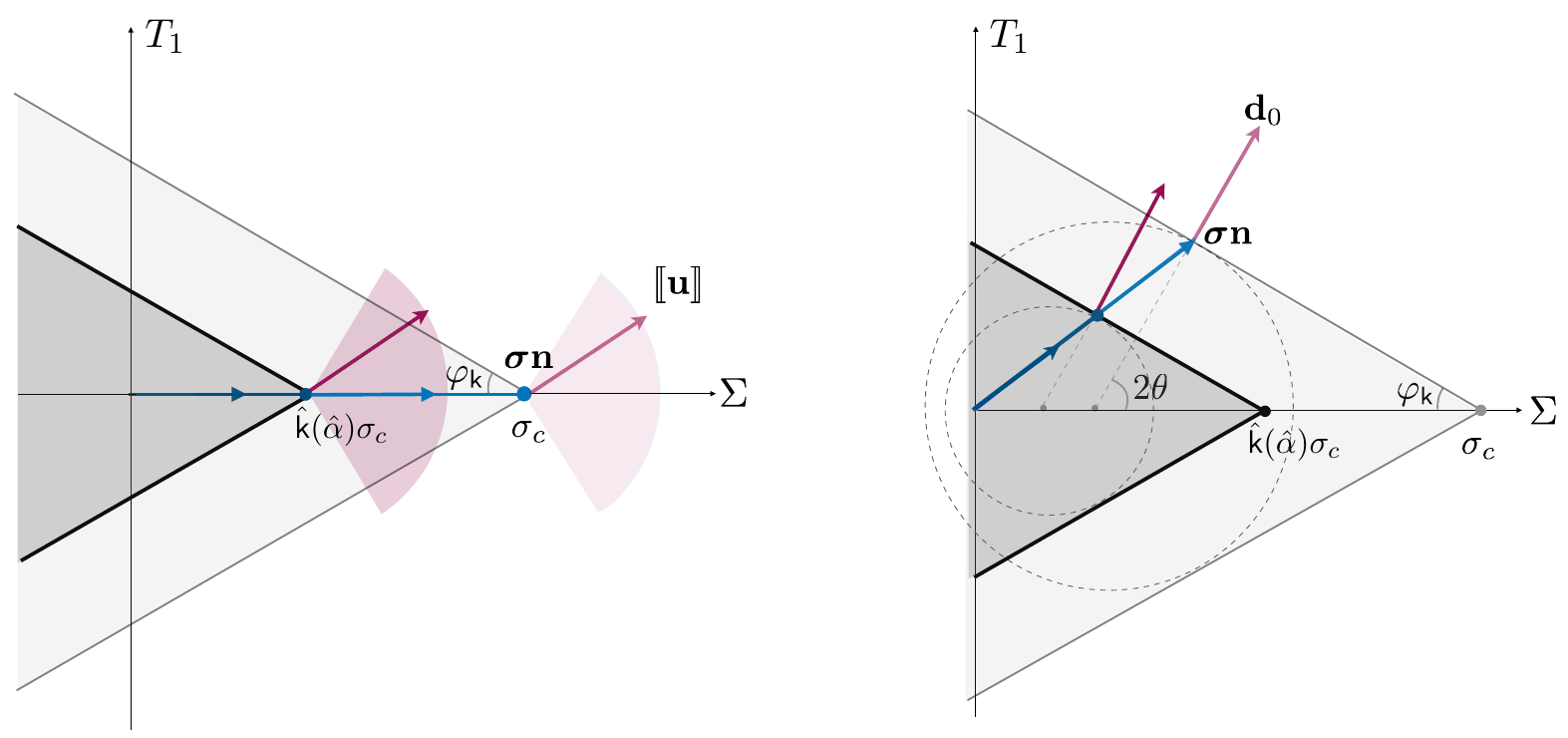}
    \end{center}
\caption{
Mohr's representation of the stress evolution under the Drucker-Prager criterion for an incompressible material in the multiaxial model problem defined in Problem~\ref{pb:fundamental-problem}. Left: Purely isotropic loading direction. Right: Loading direction with a deviatoric component, compatible with a displacement jump, see Eq.~\eqref{eq:DP-direction}. Blue arrows indicate the loading paths; pink arrows represent admissible normals to the intrinsic surface (not uniquely defined at the apex). A cohesive crack nucleates when the stress state reaches the boundary of the intrinsic surface. The intrinsic surfaces, shown in gray, progressively shrink during loading due to damage accumulation.}
\label{fig:examples-DP}
\end{figure}

We consider the specific case in which the material exhibits \textit{incompressible} linear behavior, characterized by an infinite bulk modulus ($\kappa_0 \to \infty$) and a Poisson's ratio of $1/2$. Although this scenario slightly deviates from the general framework, minor modifications enable the reuse of the main analytical results. Nonlinear compressibility arises once the stress reaches the initial strength threshold. The volumetric strain energy density, expressed as a function of the strain and internal strain variables, is given by:
\[
\varphi_{0}(\boldsymbol{\varepsilon}, \mathbf{p}) = 
\begin{cases}
\mu_0 (\boldsymbol{\varepsilon}^D - \mathbf{p}^D) \cdot (\boldsymbol{\varepsilon}^D - \mathbf{p}^D) + \H_{\K_0}(\mathbf{p}) & \text{if } \mathrm{tr}(\boldsymbol{\varepsilon}) = \mathrm{tr}(\mathbf{p}), \\\\
+\infty & \text{ otherwise}.
\end{cases}
\]
Using~\eqref{piDP}, this becomes a function of $(\boldsymbol{\varepsilon}, \mathbf{p}^D)$:
\begin{equation*}
\varphi_{0}(\boldsymbol{\varepsilon}, \mathbf{p}^D) = 
\begin{cases}
\mu_0 \norm{\boldsymbol{\varepsilon}^D - \mathbf{p}^D}^2 + \sigma_c\,\mathrm{tr}(\boldsymbol{\varepsilon}) & \text{if } \mathrm{tr}(\boldsymbol{\varepsilon}) \ge 3\sqrt{2}k \|\mathbf{p}^D\|, \\\\
+\infty & \text{ otherwise}.
\end{cases}
\end{equation*}

Minimizing $ \varphi_{0} $ over $ \mathbf{p}^D $ yields the elastic strain energy:
\begin{equation*} 
\psi_{0}(\boldsymbol{\varepsilon}) =
\begin{cases}
\mu_0 \left( \left( \| \boldsymbol{\varepsilon}^D \| - \dfrac{\mathrm{tr}(\boldsymbol{\varepsilon})}{3\sqrt{2}k} \right)^+ \right)^2 + \sigma_c \, \mathrm{tr}(\boldsymbol{\varepsilon}) & \text{if } \mathrm{tr}(\boldsymbol{\varepsilon}) \ge 0, \\\\
+\infty & \text{ otherwise},
\end{cases}
\end{equation*}
where $x^+ = \max\{0, x\}$. The stress-strain relationship is summarized below:
\[
\begin{cases}
\boldsymbol{\varepsilon}^D = \dfrac{\boldsymbol{\sigma}^D}{2\mu_0} \text{ and } \mathrm{tr}(\boldsymbol{\varepsilon}) = 0 & \text{ if } \|\boldsymbol{\sigma}^D\| < k\sqrt{2}(3\sigma_c - \mathrm{tr}(\boldsymbol{\sigma})), \\
\boldsymbol{\varepsilon}^D = \left( \dfrac{1}{2\mu_0} + \dfrac{\mathrm{tr}(\boldsymbol{\varepsilon})}{3\sqrt{2}k\|\boldsymbol{\sigma}^D\|} \right) \boldsymbol{\sigma}^D \text{ and } \mathrm{tr}(\boldsymbol{\varepsilon}) \ge 0 & \text{ if } \|\boldsymbol{\sigma}^D\| = k\sqrt{2}(3\sigma_c - \mathrm{tr}(\boldsymbol{\sigma})) \ne 0, \\
\mathrm{tr}(\boldsymbol{\varepsilon}) \ge 3\sqrt{2}k\|\boldsymbol{\varepsilon}^D\| & \text{ if } \boldsymbol{\sigma}^D = \mathbf{0} \text{ and } \mathrm{tr}(\boldsymbol{\sigma}) = 3\sigma_c.
\end{cases}
\]

We now study the {multiaxial model problem} posed on \corr{the periodic cell} $\Omega_L$ \corr{of Problem~\ref{pb:fundamental-problem}} with one face normal to $\mathbf{n}$, under two distinct loading directions: first with purely spherical stress, then with a deviatoric component. In both cases, we seek a solution with uniform stress.

\paragraph{{Spherical stress loading direction.}} Consider the {multiaxial model problem} on the \corr{cell} where the imposed stress direction is the identity tensor $\mathbf{I}$, and loading is controlled via the average volumetric deformation $\int_{\Omega_L} \mathrm{tr}(\boldsymbol{\varepsilon}_t)\,dV = tL^3$. We look for a solution with spherical stress: $\boldsymbol{\sigma}_t(\mathbf{x}) = \sigma_t(\mathbf{x})\mathbf{I}$. Equilibrium then requires $\nabla\sigma_t = 0$, so the field is uniform: $\boldsymbol{\sigma}_t(\mathbf{x}) = \sigma_t\mathbf{I}$. The loading condition imposes a volume change from $t>0$. 

Let us examine the homogeneous response: since the material is initially incompressible, stresses reach the threshold $\sigma_c$ immediately for $t > 0$, and there is no linear elastic phase. In stress space, this corresponds to the apex $(\mathrm{tr}(\boldsymbol{\sigma})/3=\sigma_c, \boldsymbol{\sigma}^D = \mathbf{0})$ of the cone $\mathbb{K}_0$, and in Mohr's representation, to the apex $(\Sigma = \sigma_c, 0, 0)$ of $\mathcal{K}_0$. As this point lies on the intrinsic curve, the homogeneous solution is unstable and a cohesive crack nucleates immediately for $t > 0$ with damage appearing where displacements are discontinuous. Thus, $\sigma_t < \sigma_c$ for all $t > 0$. 

In the non-cracked part, $\boldsymbol{\varepsilon}_t = \mathbf{0}$\corr{; since $\u_t=\mathbf E_t\x+\mathbf w_t$ with $\mathbf w_t$ periodic, the imposed volume change $\mathrm{tr}\,\mathbf E_t=t>0$ is carried by a displacement jump as soon as $t>0$. All directions being eigenvectors of the stress tensor, any plane is a candidate crack. As in Section~\ref{sec:model-problem-loc}, we take the crack on the cross-section $\Gamma=\{\x\in\Omega_L:\x\cdot\n=0\}$. The strain vanishes on both sides of $\Gamma$, so the jump is a rigid-body motion, and the periodicity of $\mathbf w_t$ in the two directions parallel to $\Gamma$ excludes a relative rotation, as for~\eqref{sautu}: the jump is a uniform translation,
\[ \jump{\mathbf{u}_t}(\mathbf{x}') = \boldsymbol{\xi}_t,\quad\forall\,\mathbf x'\in\Gamma. \]
The loading implies $\boldsymbol{\xi}_t \cdot \mathbf{n} = tL$, so the jump is nonzero on $\Gamma$.} This requires stress to reach the threshold: $\sigma_t = \hat{\k}(\alpha_t(\mathbf{x}'))\corr{\,\sigma_c}$, and thus damage is uniform on $\Gamma$: $\alpha_t(\mathbf{x}') = \alpha_t$.

Using~\eqref{piDP}, the damage criterion becomes:
\begin{equation}\label{criloc}
\hat{\k}'(\alpha_t)\sigma_c\corr{\,\boldsymbol\xi_t\cdot\n} + \Gc \ge 0,
\end{equation}
with equality when damage evolves. Since damage appears for $t > 0$, equality holds near $t=0$ and persists as long as $\alpha_t < 1$. From~\eqref{criloc} \corr{and $\boldsymbol\xi_t\cdot\n=tL$} we deduce the damage evolution and the normal displacement jump:
\[ \alpha_t = (-\hat{\k}')^{-1}\left( \corr{\frac{\Gc}{\sigma_c t L}} \right), \quad \corr{\boldsymbol\xi_t\cdot\n = tL}. \]
The tangential translation remains \corr{arbitrary}, constrained only by:
\[ \corr{\|\boldsymbol\xi_t\| \sin\varphi_k \le \boldsymbol\xi_t\cdot\n = tL}. \]
\corr{The plane crack is one critical point among many. The surface energy of this example is that of the Mohr-Coulomb criterion with $\varphi=\varphi_k$; its strength depends on the orientation of the jump, and it may fail to be convex in the crack normal, in which case the cohesive law must be relaxed (Remark~\ref{rem:convexity-normal}).}

\paragraph{{Non-spherical stress loading direction.}} 
We now impose a non-spherical stress direction $\mathbf{s}$ on the \corr{cell}.
Only  specific loading directions are compatible with displacement jumps. Indeed, only the points on $\partial\K_0$ with  
 $$2\sigII=\sigI(1-\sin\varphi_k)+\sigIII(1+\sin\varphi_k)$$  
 give stress vectors on the boundary of the intrinsic section (see Remark~\ref{rem:jumpcompatibility} in Section~\ref{sec:Mohr-DP}). Hence, a crack can nucleate only if the loading direction respects this condition. 
 
 Let us consider a loading direction $\s$ in the stress space with eigenvalues
\[ s_\mathrm{I} = 1 \corr{+} \eta,\quad s_\mathrm{II} = 1 - \eta\sin\varphi_k,\quad s_\mathrm{III} = 1 \corr{-} \eta,\quad \eta > 0 \]
and eigenvectors $\mathbf{e}_\mathrm{I}$, $\mathbf{e}_\mathrm{III}$ in the plane defined by directions $\mathbf{n}$ and $\mathbf{t}$:
\begin{equation}
    \label{eq:DP-direction}\mathbf{e}_\mathrm{I} = \corr{\cos\theta\,\mathbf{n} + \sin\theta\,\mathbf{t}},\quad \mathbf{e}_\mathrm{III} = \corr{-\sin\theta\,\mathbf{n} + \cos\theta\,\mathbf{t}},\quad \theta = \frac{\pi}{4} - \frac{\varphi_k}{2}. 
\end{equation}
The parameter $\eta$ measures the non-sphericity of $\mathbf{s}$\corr{; here $\s$ is not normalized}. We are particularly interested in small $\eta$, \emph{i.e.}, near-spherical directions. 
We seek a uniform stress solution $\boldsymbol{\sigma}_t = \sigma_t \mathbf{s}$.
The evolution of the states and of the intrinsic domain is represented in Figure~\ref{fig:examples-DP}-right.

In the first loading phase, while stresses are inside the initial strength domain (light gray region), there is no volume change. Then $2\mu_0\boldsymbol{\varepsilon}_t = \sigma_t \mathbf{s}^D$ and $\boldsymbol{\sigma}_t = \sigma_t \mathbf{s}$. Since $t = \mathbf{s} \cdot \boldsymbol{\varepsilon}_t$, we find:
\[ t = \eta^2\left(1 + \frac{1}{3}\sin^2\varphi_k\right)\frac{\sigma_t}{\mu_0}. \]
This phase ends for the average deformation $\varepsilon_e^\s$ such that $\boldsymbol{\sigma}_{t}\mathbf{n}$ reaches the boundary of $\mathcal{K}_0$ at the tangency point with the Mohr's circle. It follows from the geometrical construction in Figure~\ref{fig:examples-DP}-right:
\[ \sigma_\mathrm{I} = (1 \corr{+} \eta)\sigma_e^{\corr{\s}},\quad \sigma_\mathrm{III} = (1 \corr{-} \eta)\sigma_e^\s,\quad \sigma_e^\s = \frac{\sin\varphi_k}{\sin\varphi_k + \eta}\,\sigma_c. \]
Here $\sigma_e^\s$ is the stress intensity at $\varepsilon_e^\s$. Note that $\varepsilon_e^\s$ is of second order and $\sigma_e^\s \approx \sigma_c$ for small $\eta$.

For $t > \varepsilon_e^\s$, a cohesive crack forms on section $\Gamma$ with normal $\mathbf{n}$. The displacement jump is collinear with the outward normal to $\mathcal{K}_0$, \emph{i.e.}, $\mathbf{d}_0 = \sin\varphi_k\,\mathbf{n} + \cos\varphi_k\,\mathbf{t}$ (pink arrow in the figure). Hence, $\sigma_t < \sigma_e^\s$ as soon as $t > \varepsilon_e^\s$. In $\Omega_L \setminus \Gamma$, the stress lies inside $\mathbb{K}_0$ so nonlinear strains vanish and we have:
\[ \boldsymbol{\varepsilon} = \frac{\sigma_t}{2\mu_0}\,\mathbf{s}^D \quad\text{in}\quad \Omega_L \setminus \Gamma. \]
\corr{As in Section~\ref{sec:model-problem-loc}, the uniformity of the strain and the periodicity of $\mathbf w_t$ make the displacement jump across $\Gamma$ a uniform translation, collinear with $\mathbf{d}_0$}: $\jump{\mathbf{u}_t} = d_t \mathbf{d}_0$.

The loading condition becomes:
\[ t = \|\mathbf{s}^D\|^2 \frac{\sigma_t}{2\mu_0} + \frac{d_t}{L}\,\mathbf{s}\mathbf{n}\cdot\mathbf{d}_0 = \corr{\frac{\sigma_t}{\mu_0}}\left(1 + \frac{1}{3}\sin^2\varphi_k\right)\eta^2 + \frac{d_t}{L}(\sin\varphi_k + \eta). \]
For small $\eta$, the leading order term gives $d_t \approx tL / \sin\varphi_k$. Applying Proposition~\ref{P-Evolhat1}, we deduce that damage grows uniformly on $\Gamma$ from $0$ to $1$, reached when $d_t = \Gc / (\sigma_c\corr{\sin\varphi_k} |\hat{\k}'(1)|)$.\\

The fundamental distinction between spherical and non-spherical loading lies in the conditions for crack initiation. Under spherical stress, crack nucleation occurs automatically once the stress reaches the failure criterion. In contrast, for non-spherical (i.e., deviatoric) stress states, crack initiation occurs only for specific loading directions. The presence of deviatoric components imposes additional constraints on the admissible crack orientations and associated displacement jumps. Furthermore, under the assumption of uniform stress, the solution can be determined up to symmetry, as the volumetric strain governs the surface displacement jump and, consequently, the damage.

Replacing the Drucker-Prager criterion with the Mohr-Coulomb criterion—whose normals are compatible with a displacement jump for all loading directions—allows crack initiation to occur without restrictions on the loading direction. This situation will be illustrated in the following example, which considers the specific case of the Tresca criterion with a tension cut-off.

\subsubsection{Tresca criterion with a tension cut-off}
\label{sec:examples-tresca-cap}
The Tresca criterion with a tension cut-off is an intrinsic criterion requiring the maximum tension not to exceed a tensile strength $\sigma_c$ and the maximum shear not to exceed a shear strength $\tau_c$. 
The strength criterion, the intrinsic section and the corresponding support functions have been introduced in Section~\ref{sec:Mohr-DP}. Figure~\ref{fig:tresca-cup} shows the intrinsic section in Mohr's representation.

As before, we discuss
 the solution to the multiaxial loading problem for a generic loading direction $\s$.
If this direction corresponds to a hydrostatic compression, \emph{i.e.} if $\mathbf s=-\mathbf I/\sqrt3$, then the response will remain purely elastic, as the criterion will never be reached. In all other cases, a cohesive crack will nucleate once a critical loading level $\varepsilon_e^\s$ is reached, and will then degrade until rupture. Let us distinguish between several cases, denoting by $(s_\mathrm{I},s_\mathrm{II},s_\mathrm{III})$ the ordered eigenvalues of  $\s$ and by $(\mathbf  e_\mathrm{I},\mathbf  e_\mathrm{II},\mathbf e_\mathrm{III})$ the corresponding principal directions. In all cases, the stresses are uniform within the cube: $\sig=\sigma\mathbf s$, where $\sigma$ denotes their magnitude (the direction $\mathbf s$ being normalized).

\begin{figure}[htbp]
\centering
{\includegraphics[width=0.4\textwidth]{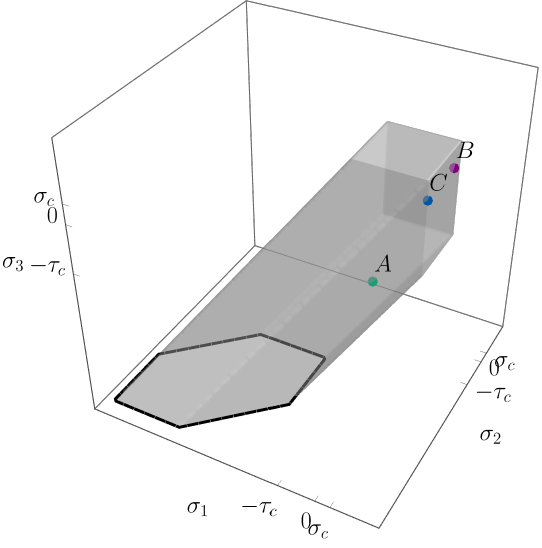}}%
\hspace{0.05\textwidth}
\raisebox{1.cm}{\includegraphics[width=0.4\textwidth]{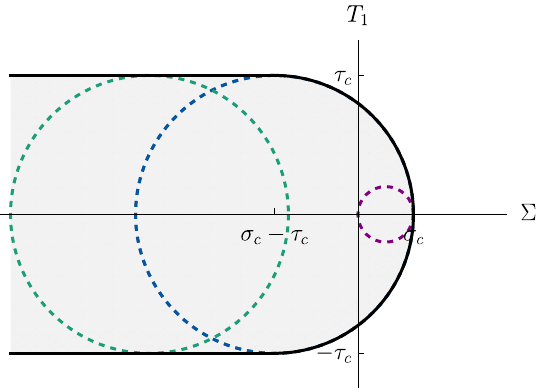}}
\caption{Tresca criterion with a tension cut-off: Strength criterion (left) and  intrinsic section (right). 
The dashed circles in the intrinsic section of the right panel represent the outer Mohr's circles at the onset of crack nucleation for the three stress states $A$, $B$, and $C$ in the left panel, which are examples of the critical states for the three classes of loading directions described in Section~\ref{sec:examples-tresca-cap}.
}
\label{fig:tresca-cup}
\end{figure}

\begin{enumerate}
\item[(A)] If $0\ne\dfrac{s_\mathrm{I}-s_\mathrm{III}}{2\tau_c}>\dfrac{s_\mathrm{I}}{\sigma_c}$, the outer Mohr's circle is tangent to the intrinsic curve at a point where $\norm{\T}=\tau_c$ when the loading reaches the critical value $\varepsilon_e^\s={2\tau_c}/({A_0^\s(s_\mathrm{I}-s_\mathrm{III})})$, see point $A$ in  Figure~\ref{fig:tresca-cup}. A cohesive crack will then appear, with orientation $\mathbf n_*$ forming an angle $\theta=\pm\pi/4$ with the principal directions $\mathbf e_\mathrm{I}$ and $\mathbf e_\mathrm{III}$. This corresponds to shear failure.

\item[(B)] If $0\le\dfrac{s_\mathrm{I}-s_\mathrm{III}}{2\tau_c}<\dfrac{s_\mathrm{I}}{\sigma_c}$, the outer Mohr's circle is tangent to the intrinsic curve at the point $(\sigma_c,0)$ when the loading reaches the critical value $\varepsilon_e^\s={\sigma_c}/({A_0^\s\,s_\mathrm{I}})$, see point $B$ in Figure~\ref{fig:tresca-cup}. A cohesive crack will then appear with an orientation aligned with the principal direction $\mathbf e_\mathrm{I}$. This corresponds to tensile failure.

\item[(C)]If $\dfrac{s_\mathrm{I}-s_\mathrm{III}}{2\tau_c}=\dfrac{s_\mathrm{I}}{\sigma_c}>0$, the maximum shear and maximum tension criteria are reached simultaneously, and the orientation of the crack is arbitrary within the plane defined by the principal directions $\mathbf e_\mathrm{I}$ and $\mathbf e_\mathrm{III}$, see point $C$ in Figure~\ref{fig:tresca-cup}.
\end{enumerate}
\corr{In each case, as in Section~\ref{sec:model-problem-loc}, the cell is taken with a face normal to the crack normal; in case (C), this choice selects the crack orientation within the plane $(\mathbf e_\mathrm{I},\mathbf e_\mathrm{III})$.}

In all cases, at the end of the loading process, the formed crack can no longer sustain stresses other than hydrostatic pressure \corr{(convention of Section~\ref{sec:PFstrengthdegradation} at full damage)}.

The cohesive model obtained with this criterion appears to be a complete model to describe the nucleation and propagation of cohesive cracks, which is particularly appealing for numerical applications. \corr{Its strength depends on the orientation of the jump, and its surface energy may fail to be convex in the crack normal for some values of $\sigma_c/\tau_c$; the cohesive law~\eqref{eq:phin} must then be relaxed (Remark~\ref{rem:convexity-normal}).} The strength surface can be easily tuned to fit the experimental data for the tensile and shear strengths $\sigma_c$ and $\tau_c$. The fracture toughness $\Gc$ appears as the additional material parameter, together with the Young modulus and Poisson ratio characterizing the linear elastic behavior. A further generalization of the present criterion to allow for additional flexibility could include, for example, a Mohr-Coulomb model with a tension cut-off~\cite{Sal83,CHEN75}.

\section{Conclusions}
\label{sec:Conclusion}

We begin by summarizing the overall construction and the results obtained before concluding and outlining a few perspectives.

Aiming to develop a fracture model capable of accounting for both the initiation and propagation of displacement discontinuity surfaces---while incorporating a bulk strength criterion expressed in terms of stress---we started from a nonlinear Hencky-type elastic law. 
This law is based on a constant linear elastic stiffness, which governs linear deformations with respect to stress, and a fixed convex set in stress space that introduces a material strength criterion and determines nonlinear deformations via the normality rule.

We then introduced a phase-field variable, which modifies the strength criterion through a softening law and entails an energetic cost that, due to the presence of a small parameter, forces localization of the phase-field variable within narrow bands of controlled thickness. 
Unlike existing phase-field fracture models, the phase-field variable does not degrade the linear elastic stiffness, but the strength, \emph{i.e.} the rate of the linear growth of the energy for large strains.
The evolution law for displacements, stresses, and the phase-field model was formulated using a standard variational approach based on three general principles: the irreversibility of the phase-field variable, an energy-based stability criterion for admissible states, and an energy balance between strain energy, dissipated energy, and external work. 

By solving an evolution problem on a \corr{periodic cubic cell}, with an arbitrary convex strength criterion, we showed that a cohesive crack nucleates as soon as the stress reaches the boundary of the strength domain, provided that an outward normal to the convex set at the contact point is compatible with a displacement jump. 
For isotropic strength criteria, this condition is equivalent to requiring that the stress vector lies on the boundary of the intrinsic domain $\mathcal{K}_0$ (see Section~\ref{sec:intrinsic-curve}) in Mohr's representation. 
Conversely, when this normal is not compatible with a displacement jump, strain localization and hence crack nucleation cannot occur; the response remains homogeneous up to strain levels diverging as the small parameter tends to zero.
The nucleated crack is cohesive, in the sense that the stress vector is a continuous function of the displacement jump, progressively decreasing from the initial strength to zero. 
We characterize this cohesive law in terms of the constitutive parameters of the models and the strength criterion. 

 
\medskip

Based on these results, we proposed a \corr{candidate} limiting fracture model that meets the initial objectives\corr{; it is lower semicontinuous, as a $\Gamma$-limit must be, only if its surface energy is convex in the crack normal (Remark~\ref{rem:convexity-normal})}.
In the bulk, it corresponds to a nonlinear Hencky-type elastic behavior, while on the surface, it corresponds to a cohesive fracture law incorporating damage. 
The phase-field variable becomes a damage parameter taking values between $0$ and $1$, acting only along displacement jumps, and whose evolution is governed by a local (without gradient terms), irreversible law. 
It contributes to the surface energy, which includes a reversible elastic component associated with the displacement jump and an irreversible part increasing with damage. 
When it reaches $1$, the surface energy equals Griffith's energy, and cohesive forces vanish (although the stress vector may not vanish, depending on the chosen strength criterion).

\medskip
The method is constructive: the limiting model is \emph{deduced} from the initial ingredients---namely, the phase-field model, the strength criterion, and the softening law coupling them. The cohesive fracture evolution law obtained inherits these initial choices.

Our model is consistent with and extends existing results on free-discontinuity problems involving cohesive surface energies:
\begin{itemize}
    \item \citet{BouBraBut95} began with a variational model involving linear elasticity in the bulk and a cohesive surface energy with a finite slope at the origin. 
    They showed that the principle of energy minimization, combined with the arbitrariness of the crack set, necessitates a relaxed formulation in which the bulk behavior incorporates a strength calibrated to the slope at the origin of the cohesive law.
    Instead, we started  with a nonlinear elastic law in the bulk, incorporating an arbitrary prescribed strength domain, and deduced the corresponding cohesive law.
    Our formulation suggests a natural path to extend the scalar (anti-plane) results of~\citet{BouBraBut95} to the vector-valued elasticity settings with arbitrary strength criteria.
    
    \item \citet{ChaLavMar06} established that cohesive fracture models derived from energy minimization naturally produce a strength domain represented by an intrinsic curve in the Mohr diagram. 
    However, their analysis assumes linear bulk elasticity, and no relaxation result is available in the setting of vector-valued elasticity. 
    Our model starts from a convex strength domain in the bulk and leads to the same class of cohesive surface energies, while naturally incorporating the self-consistent bulk behavior in the sense of the relaxation result of~\citet{BouBraBut95}.  
\end{itemize}

Let us now revisit the assumptions made and distinguish between those that could potentially be relaxed and those that appear to be fundamental.

\begin{itemize}
\item  We assumed that nonlinear deformations are reversible. This was done to simplify the presentation, but there is no conceptual obstacle to introducing irreversibility and thereby starting from an elastic-perfectly plastic model. 
    This is essentially what was proposed in \cite{AleMarVid14,AleMarVid15} for the von Mises criterion. 
    It would lead to a cohesive energy containing irreversibility with respect to displacement jumps.
    To date, this inherited irreversible cohesive law has not been clearly identified in the general vector-valued setting.

\item We assumed that the linear stiffness tensor $\A$ is invariant and independent of the phase-field. This assumption enabled a simplified presentation and allowed us to demonstrate the nucleation of cohesive cracks without requiring degradation of the elastic coefficients. 
Nonetheless, building on the present results, reintroducing stiffness degradation remains a promising direction for future work.
 It may allow the inclusion of additional effects relevant to the complex phenomenology of fracture mechanics, such as fatigue, unloading behavior, and more intricate strength criteria.

\item  We assumed that the strength criterion shrinks homothetically as damage evolves. This greatly simplified the analysis, but nothing prevents considering a more general dependence. In any case, crack nucleation depends only on the initial strength criterion, while its evolution enters only through the cohesive law. This highlights the need to introduce some degree of phenomenology—or to more thoroughly account for physical mechanisms—in order to better guide the choice of the cohesive law and the nonlinear evolution of the strength domains.

\item  The choice of the strength criterion appears essential. Indeed, the fact that crack nucleation can only occur when the stress vector lies on the boundary of the intrinsic domain (a necessary and sufficient condition), and not merely when stress reaches the threshold (a necessary but not always sufficient condition), establishes a hierarchy among criteria. For example, the Tresca criterion is ``superior'' to the von Mises criterion, in the sense that a crack may (and even must) nucleate as soon as stress reaches the threshold for the former, whereas for the latter, the intermediate principal stress must also be suitably positioned. This casts strength criteria in a new light. Criteria such as the Tresca or Mohr-Coulomb with a tension cut-off appear suitable to produce acceptable models for the most common multiaxial material behavior, see  Section~\ref{sec:strenght-examples} and Section~\ref{sec:examples-tresca-cap}. \corr{For the Tresca criterion with a tension cut-off, the surface energy may fail to be convex in the crack normal for some values of $\sigma_c/\tau_c$, in which case the cohesive law must be relaxed (Remark~\ref{rem:convexity-normal}).}

\item  The entire construction relies essentially on the variational approach: existence of an energy functional, normality rule, stability condition, energy balance, \ldots\  depend on it. In return, this provides a framework for addressing evolution problems from both mathematical and numerical perspectives. It raises the hope of obtaining true convergence results from the phase-field model to the cohesive fracture model, generalizing the first such result obtained in \cite{DalOrlToa16} in the restricted setting of the von Mises criterion in anti-plane shear without irreversibility. \corr{Such a result must identify the relaxed surface energy where the condition of Remark~\ref{rem:convexity-normal} fails.} From a numerical standpoint, we retain the initial phase-field model with its parameter, which can be chosen arbitrarily small, allowing for computations with regular fields that approximate displacement discontinuities without any \emph{a priori} assumption on their spatiotemporal evolution. It remains to implement and test this method on examples to validate the approach. In any case, departing from this well-established and now well-understood ``variational'' path based on energy minimization principles to venture into uncharted territory seems highly risky. We showed that this is not necessary to obtain a consistent cohesive fracture model with arbitrary (convex) strength criteria. We shall not advise attempting it.
\end{itemize}

In conclusion, provided that complete mathematical results and conclusive numerical tests support this construction, we believe that the proposed approach can serve as the backbone of a general theory of fracture. 
Such a theory would bridge various concepts that have traditionally been developed separately, including strength criteria, plasticity laws, limit analysis, Griffith's theory, ductile fracture, cohesive laws, and damage mechanics.

\section{Acknowledgements}
The authors wish to thank Patrick Ballard for sharing with them a proof of~\eqref{eq:suptensors}.
BB  acknowledges the support of the Natural Sciences and Engineering Research Council of Canada (NSERC), RGPIN-2022-04536 and the Canada Research Chair program.
CM and CZ received funding from the European Union’s Horizon 2020 research
and innovation program under the Marie Skłodowska-Curie grant agreement No. 861061 -
NEWFRAC Project. 
CM and JJM gratefully acknowledge the \emph{Fondation des Treilles} for hosting their research residency on variational fracture mechanics.

\bibliography{24-Nucleation}
\end{document}